\documentclass[11pt,a4paper,titling,fleqn]{article} 

\usepackage[paper=a4paper,left=25mm,right=25mm,top=25mm,bottom=25mm]{geometry}

\usepackage{amsthm,amsmath,amsfonts,amssymb}
\usepackage{graphicx,psfrag,epsf}
\usepackage{enumerate}
\usepackage[round]{natbib}
\usepackage{url}
\usepackage{hyperref} 
\hypersetup{colorlinks=true,allcolors=[rgb]{0,0,1}}
\usepackage{color,xcolor}
\usepackage{dsfont}
\usepackage{multirow}
\usepackage{array}
\usepackage{booktabs}

\allowdisplaybreaks

\selectfont

\theoremstyle{plain}
\newtheorem{theorem}{Theorem}[section]
\newtheorem{proposition}[theorem]{Proposition}
\newtheorem{corollary}[theorem]{Corollary}
\newtheorem{lemma}[theorem]{Lemma}
\theoremstyle{remark}
\newtheorem{remark}[theorem]{Remark}

\newtheorem{assumption}[theorem]{Assumption}
\newtheorem{example}[theorem]{Example}

\newcommand{\C}{\mathbb{C}}
\newcommand{\E}{\mathbb{E}}
\newcommand{\N}{\mathbb{N}}
\newcommand{\R}{\mathbb{R}}

\newcommand{\PP}{\mathbb{P}}
\newcommand{\XX}{\mathbf{X}}
\newcommand{\ZZ}{\mathbf{Z}}

\newcommand{\xx}{\mathbf{x}}
\newcommand{\zz}{\mathbf{z}}

\newcommand{\cA}{\mathcal{A}}
\newcommand{\cU}{\mathcal{U}}

\newcommand{\Ran}{\mathsf{Ran}}
\newcommand{\eqd}{\stackrel{\mathrm{d}}=}
\newcommand{\sgn}{\mathsf{sgn}}
\newcommand{\Var}{\mathrm{Var}}
\newcommand{\de}{\mathrm{\,d}}
\newcommand{\supp}{\mathop{\mathrm{supp}}}

\definecolor{light-gray}{gray}{0.95}
\definecolor{darkblue}{rgb}{0,0,.5}
\definecolor{foxred}{rgb}{0.7, 0.11, 0.11}
\newcommand{\Mail}[1]{{\color{foxred} #1}}

\begin{document}


\title{\bf A new coefficient of separation}

\author{%
	Sebastian Fuchs\footnote{University of Salzburg, Austria, Email: \Mail{sebastian.fuchs@plus.ac.at}}
	\qquad
	Carsten Limbach\footnote{University of Salzburg, Austria, Email: \Mail{carsten.limbach@plus.ac.at}} 
	\qquad
	Patrick B. Langthaler\footnote{University of Salzburg, Austria, Email: \Mail{patrick.langthaler@plus.ac.at}}
	\vspace{5mm}
}

\maketitle

\begin{abstract}
A coefficient is introduced that quantifies the extent of separation of a random variable $Y$ relative to a number of variables $\mathbf{X} = (X_1, \dots, X_p)$ by skillfully assessing the sensitivity of the relative effects of the conditional distributions. 
The coefficient is as simple as classical dependence coefficients such as Kendall's tau, also requires no distributional assumptions, 
and consistently estimates an intuitive and easily interpretable measure, which is $0$ if and only if $Y$ is stochastically comparable relative to $\mathbf{X}$, that is, the values of $Y$ show no location effect relative to $\XX$, and $1$ if and only if $Y$ is completely separated relative to $\mathbf{X}$.
As a true generalization of the classical relative effect, in applications such as medicine and the social sciences the coefficient facilitates comparing the distributions of any number of treatment groups or categories. It hence avoids the sometimes artificial 
grouping of variable values such as patient's age into just a few categories, which is known to cause inaccuracy and bias in the data analysis.
The mentioned benefits are exemplified using synthetic and real data sets.
\end{abstract}

\noindent%


\section{Introduction} \label{Sec.Intro}
In statistics, the problem of measuring the effect size arises in various situations and has a long and fruitful history.
A popular tool for comparing the distributions of two treatment groups / populations, represented by the random variables $Z_1$ and $Z_2$ (usually assumed to be independent), 
consists in evaluating the functional 
\begin{equation} \label{Def.Rel.Eff}
  \Psi(\PP^{Z_1}, \PP^{Z_2}) 
  := \int_{\mathbb{R}} \PP(Z_1 < z) + \frac{1}{2} \; \PP(Z_1 = z) \; \de \PP^{Z_2}(z)\,,
\end{equation} 
introduced for the first time in \citet{wilcoxon1945} and \citet{mann1947}.
Over time different authors investigated $\Psi$ under various names, for example relative effect~\citep{birnbaum1957}, common language effect~\citep{mcgraw1992} and winning probability or statistical preference~\citep{deschuymer2003}. 
Here we follow the lead of~\citet{birnbaum1957} and \citet{Brunner2019} and use the term \emph{(nonparametric) relative effect}.
The relative effect has been used extensively for statistical designs involving simple group comparisons or factorial designs for independent groups~\citep{rankFD}, repeated measures designs~\citep{nparLD} and multivariate data~\citep{npmv}. 
Areas of application include medicine~\citep{schimke2022severe, kedor2022prospective, agathocleous2017ascorbate, tasdogan2020metabolic, schiroli2019precise}, anthropology~\citep{lucquin2018impact}, ecology~\citep{page2020comprehensive, augusto2022tree}, and the social sciences~\citep{seidel2014modeling, zhang2021spread, richetin2015should}.

The relative effect $\Psi$ in \eqref{Def.Rel.Eff} quantifies the stochastic tendency of one variable to take on greater values than the other, and can therefore be employed for describing the degree of separation or conformity of two treatment groups / populations: 
If \(\Psi(\PP^{Z_1}, \PP^{Z_2})\) equals $1/2$, then $Z_1$ and $Z_2$ are said to be \emph{stochastically comparable} \citep{Brunner2019};
if, instead, \(\Psi(\PP^{Z_1}, \PP^{Z_2}) \in \{0,1\}\), then $Z_1$ and $Z_2$ are said to be \emph{completely separated} \citep{Brunner2019,navarro2023}.

Looking at \eqref{Def.Rel.Eff} from a different perspective and replacing the pair $(Z_1,Z_2)$ with $(X,Y)$, where $Y$ is a merge of $Z_1$ and $Z_2$ and the (discrete, two-valued) random variable $X$ determines the belonging to one of two treatment groups $x_1$ and $x_2$, the comparison of $Z_1$ and $Z_2$ converts to a comparison of the two conditional distributions $Y|X=x_1$ and $Y|X=x_2$.
Therefore, the initial problem in \eqref{Def.Rel.Eff} translates to quantifying the degree of separation of $Y$ relative to $X \in \{x_1,x_2\}$ which allows an extension of \eqref{Def.Rel.Eff} for comparing \emph{any} number of treatment groups instead of just two.
Incorporating vector-valued $\XX = (X_1, \dots, X_p)$ even offers the possibility of comparing \emph{any} number of treatment groups arranged according to different criteria or covariates.

The main contribution of this paper is a novel measure of separation \(\Lambda = \Lambda(Y|\XX)\) for a random variable $Y$ relative to a set of other variables $X_1, \dots, X_p$.
The extent of separation of $Y$ relative to $\XX$ is determined by how sensitive the relative effect $\Psi(\PP^{Y|\XX=\xx_1}, \PP^{Y|\XX=\xx_2})$ of the conditional distributions $\PP^{Y|\XX=\xx_1}$ and $\PP^{Y|\XX=\xx_2}$ is on $\xx_1$ and $\xx_2$, and then averaging over all possible values of $\xx_1$ and $\xx_2$. 
The main features of our measure \(\Lambda\) are the following:
\begin{enumerate}[({A}1)]
\item \label{REM:MainProp:1} 
\(0 \leq \Lambda(Y|\XX) \leq 1\).
\item \label{REM:MainProp:2}
\(\Lambda(Y|\XX)=0\) if and only if $Y$ is \emph{stochastically comparable relative to} $\XX$,
i.e. \linebreak 
$ \Psi \big(\PP^{Y|\XX=\xx_1}, \PP^{Y|\XX=\xx_2}\big) = \frac{1}{2}$
for all almost all $(\xx_1,\xx_2)$ with $\xx_1 \neq \xx_2$.
\item \label{REM:MainProp:3}
\(\Lambda(Y|\XX)=1\) if and only if \(Y\) is \emph{completely separated relative to} \(\XX\), i.e. \linebreak
$\Psi \big(\PP^{Y|\XX=\xx_1}, \PP^{Y|\XX=\xx_2}\big) \in \{0,1\}$ for almost all $(\xx_1,\xx_2)$ with $\xx_1 \neq \xx_2$.
\end{enumerate}
The measure $\Lambda$ is a direct and natural generalization of the relative effect in \eqref{Def.Rel.Eff} and the first of its kind to quantify the extent of separation of $Y$ relative to vector-valued $\XX$. 
As a demonstration of how $\Lambda$ works, the extent of separation measured by $\Lambda$ is illustrated in Remark \ref{Rem:Illust.Lambda} with the help of $X \in \{x_1,x_2\}$ and two normal distributions $Y|X=x_1$ and $Y|X=x_2$ with varying difference between mean values (Figure \ref{Ex.BehrensFisher.Pic1}), and by means of a normally distributed vector $(\XX,Y)$ (Figure \ref{Ex.Norm.Pic}).

\begin{figure}
    \centering
    \includegraphics[width=\linewidth]{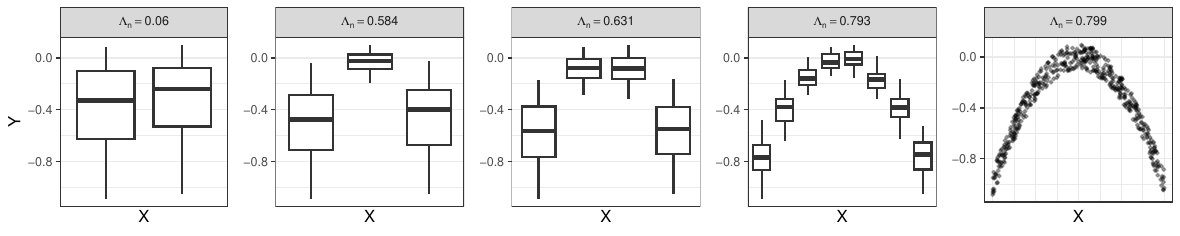}
    \caption{The rightmost panel shows a scatterplot of sample size $512$ depicting the pair $(X,Y)$ with $X \sim \mathcal{U}(-1, 1)$ and $Y = -X^2 + \varepsilon$, $\varepsilon \sim \mathcal{U}(-0.1, 0.1)$. Boxplots for a discretized $X$-variable with 2, 3, 4, and 8 categories respectively are shown in the first four panels. For each level of discretization the estimated values of $\Lambda$ are given at the top of the panel.}
    \label{fig:intro}
\end{figure}

The benefit of the measure $\Lambda$ over the relative effect in \eqref{Def.Rel.Eff} when quantifying the degree of separation of $Y$ relative to $\XX$ is demonstrated in Figure \ref{fig:intro} with synthetic data. Here, as is often the case in applications such as medicine, the variable $X$ with arbitrary value range (continuous or discrete), such as patient's age, BMI, blood pressure or weight, is grouped into two (left panel) or a few categories (mid panels) as part of a pre-processing step. The estimated degrees of separation range from $0.060$ in the leftmost panel with just two groups to $0.799$ in the rightmost panel for the original data without grouping. 
It is worth noting that the relative effect in \eqref{Def.Rel.Eff} can only be determined for two groups, whereas $\Lambda$ is applicable to all five variants of grouping while producing only a single score in each case.
Figure \ref{fig:intro} demonstrates that 
arbitrary grouping can lead to a considerable inaccuracy or distortion when estimating the degree of separation and is not needed for our 
$\Lambda$. To put it another way, when applied to the original data, $\Lambda$ avoids the sometimes artificial grouping of the variable values into just a few categories and thus reduces researchers' degree of freedom and avoids problems inherent with grouping of continuous variables.
We refer to Section \ref{Sec.RealDataEx} for a real data example in which the effects of discretizing patient's age are examined.

This paper presents a true generalization $\Lambda$ of the relative effect in \eqref{Def.Rel.Eff} that is capable of quantifying the extent of separation of a random variable $Y$ relative to a set of random variables $X_1, \dots, X_p$ and that requires no distributional assumptions.
We present closed-form expressions for $\Lambda$ in various settings (Section \ref{SubSec:Coeff:Normla}), prove invariance properties for $\Lambda$ (Section \ref{SubSec:Coeff:Invariance}) and show that, under additional continuity assumptions on $Y$ or $\XX$, the concept complete separation is closely related to or even coincides with that of perfect dependence, that is, $Y$ is almost surely a function of $\XX$ (Section \ref{Sec:PS.PD}). In the latter case, $\Lambda$ qualifies as a measure quantifying the extent of functional dependence (Section \ref{SubSec:PS.PD:MoD}), making it suitable for a variable selection method.
Since $\Lambda$ evaluates conditional distributions, it is not surprising that $\Lambda$ generally fails to be continuous with respect to weak convergence. 
In Section \ref{Sec:Cont} it is then shown that $\Lambda$ is instead continuous with respect to conditional weak convergence introduced in \citep{Sweeting-1989} yielding robustness of $\Lambda$ against small pertubations of $Y$.
In Section \ref{Sec:Estimation} a strongly consistent estimator $\Lambda_n$ of $\Lambda$ is presented that relies on the graph-based estimation principle introduced in \citep{chatterjee2021AoS}.
The coefficient $\Lambda_n$ exhibits a simple expression, is fully non-parametric, has no tuning parameters and can be computed in $O(n \log n)$ time. 
For comparative purposes, a consistent estimator $\widehat{\Lambda}_n$ for $\Lambda$ is discussed that is based on the classical relative effect estimation \citep{Brunner2019}, but only proves to perform well for $\XX$ being discrete. 
Synthetic data is used to evaluate the overall performance of \(\Lambda_n\) in various scenarios and when compared to the classical relative effect estimator $\widehat{\Lambda}_n$ in discrete settings.
Further simulation studies address and demonstrate robustness of $\Lambda$ against small perturbations of $Y$, highlight certain phenomena for vector-valued $\XX$, perform a variable selection based on $\Lambda$, identify contrasts with Chatterjee's correlation coefficient \citep{chatterjee2021AoS} and illustrate their potential in the examination of heteroscedasticity (Section \ref{Sec.Sim}).

All proofs and additional results are available in the online Supplementary Material.

Throughout the paper, let \(Y\) be a non-degenerate random variable and \(\XX=(X_1,\ldots,X_p)\) be a \(p\)-dimensional random vector, \(p \in \N = \{1,2,\dots\}\) being arbitrary, with at least one non-degenerate coordinate, both defined on a common probability space \((\Omega,\cA,P)\).
We refer to \(Y\) as the response variable and \(\XX\) as the vector of predictor variables. Variables not in bold type refer to one-dimensional real-valued quantities.

\section{The coefficient} \label{Sec:Coeff}

We propose the following quantity as a measure of the degree of separation of $Y$ relative to the random vector \(\XX\):
\begin{equation} \label{Def:REM}
  \Lambda(Y|\XX)  
  := \alpha^{-1} \, \int_{\R^p \times \R^p} \big( \Psi \big(\PP^{Y|\XX=\xx_1}, \PP^{Y|\XX=\xx_2}\big)             - \Psi \big(\PP^{Y|\XX=\xx_2}, \PP^{Y|\XX=\xx_1}\big) \big)^2 \de (\PP^\XX \otimes \PP^\XX) (\xx_1,\xx_2) 
\end{equation}
where
$\alpha := \alpha(\XX) := 1 - \PP(\XX = \XX^\ast)$
with \(\XX^\ast\) denoting an independent copy of \(\XX\).
Since it is assumed that at least one coordinate of \(\XX\) is non-degenerate, the normalizing constant fulfills $0 < \alpha(\XX) \leq 1$ and the quantity \(\Lambda(Y|\XX) \) is well-defined.
We note in passing that \(\alpha(\XX) = 1\) if at least one coordinate of \(\XX\) has a continuous cdf.
The value \(1-\alpha(\XX)\) can be interpreted as the degree of discrete mass concentration of \(\XX\).

As a first main result, the following theorem encapsulates the key characteristics of \(\Lambda\) and provides alternative representations.

\begin{theorem}[Measure of separation]~~ \label{Thm:REM:Main}
\(\Lambda\) defined in \eqref{Def:REM} equals
\begin{eqnarray}
  \Lambda(Y|\XX)
  & = & 4 \, \alpha^{-1} \, \left( \int_{\R^p \times \R^p} \left( \Psi \big(\PP^{Y|\XX=\xx_1}, \PP^{Y|\XX=\xx_2}\big) - \frac{1}{2} \right)^2 
      \de (\PP^\XX \otimes \PP^\XX) (\xx_1,\xx_2) \right)
  \label{Thm:REM:Main:Eq1}
  \\
  & = & 4 \, \alpha^{-1} \left( \int_{\R^p \times \R^p} \Psi \big(\PP^{Y|\XX=\xx_1}, \PP^{Y|\XX=\xx_2}\big)^2 \de (\PP^\XX \otimes \PP^\XX) (\xx_1,\xx_2) - \frac{1}{4} \right)
  \label{Thm:REM:Main:Eq2}
\end{eqnarray}
and fulfills the properties (A\ref{REM:MainProp:1}), (A\ref{REM:MainProp:2}), and (A\ref{REM:MainProp:3}), i.e.
\begin{enumerate}[({A}1)]
\item \(0 \leq \Lambda(Y|\XX) \leq 1\).
\item \(\Lambda(Y|\XX)=0\) if and only if $Y$ is stochastically comparable relative to $\XX$.
\item \(\Lambda(Y|\XX)=1\) if and only if $Y$ is completely separated relative to $\XX$.
\end{enumerate}
\end{theorem}

We shall also say that the random vector \((\XX,Y)\) is \emph{stochastically comparable} if $Y$ is stochastically comparable relative to $\XX$, and is \emph{completely separated} if \(Y\) is completely separated relative to \(\XX\).

\begin{remark}[Interpreting the coefficient]\label{Rem:REM:Main}
\begin{enumerate}
\item In light of \eqref{Thm:REM:Main:Eq1}, \(\Lambda\) relates the relative effects of the conditional distributions \(\Psi (\PP^{Y|\XX=\xx_1}, \PP^{Y|\XX=\xx_2})\) to that of the unconditional distribution \(\Psi(\PP^{Y},\PP^{Y}) = 1/2\). 

\item The definition of \(\Lambda\) in \eqref{Def:REM} suggests interpreting \(\Lambda(Y|\XX)\) as a measure quantifying the sensitivity of the relative effects \(\Psi (\PP^{Y|\XX=\xx_1}, \PP^{Y|\XX=\xx_2})\). $\Lambda$ hence is inspired by but considerably differs from a similar construction principle involving the sensitivity of the conditional distributions proposed in \citep{ansari2023DepM}.

\item \eqref{Thm:REM:Main:Eq2} motivates considering $\Lambda$ as a measure quantifying the variability of the relative effects \(\Psi (\PP^{Y|\XX=\xx_1}, \PP^{Y|\XX=\xx_2})\) as it fulfills
\begin{equation*}
  \Lambda(Y|\XX)
  = \frac{4 \, \Var \left( \E \left( \mathds{1}_{\{Y < Y^\ast\}} + \frac{1}{2} \mathds{1}_{\{Y = Y^\ast\}} \, | \, \XX,\XX^\ast \right) \right)}{1 - \PP(\XX = \XX^\ast)}\,,
\end{equation*}
where $(\XX^\ast,Y^\ast)$ denotes an independent copy of $(\XX,Y)$.
$\Lambda$ hence is related but, again, considerably differs from similar functionals introduced in \citep{emura2021,limbach2024,shih2024}.
\end{enumerate}
\end{remark}

\medskip
\begin{remark}[Extreme cases]\label{Rem:REM:Main2} 
\begin{enumerate}
\item (\emph{Stochastic comparability and independence}). \label{Rem:REM:Main.3}
If $Y$ and $\XX$ are independent, then the conditional distributions are stochastically comparable implying \(\Psi (\PP^{Y|\XX=\xx_1}, \PP^{Y|\XX=\xx_2}) = 1/2\) for almost all $(\xx_1, \xx_2) \in \R^p \times \R^p$ and hence $\Lambda(Y|\XX)=0$.
The converse direction does not hold, in general. For instance, the Fr{\'e}chet copula discussed in Example \ref{Ex:REM:Kendall2} in the Supplementary Material with parameter values $\alpha = \beta > 0$ fulfills $\Lambda(Y|X)=0$ but $X$ and $Y$ fail to be independent.

\item (\emph{Complete separation and perfect dependence}). \label{Rem:REM:Main.4}
If $Y$ is completely separated relative to $\XX$, then knowledge about the value of $\XX$ provides some information 
about the value of $Y$.
In other words, if the treatment groups are completely separated, then knowledge about the belonging to a certain treatment group provides information about the value of the response in this group.
Complete separation is thus closely related to the concept of perfect dependence: $Y$ is said to be perfectly dependent on $\XX$ if there exists some measurable function $f$ such that $Y = f(\XX)$ almost surely.
The relation between complete separation and perfect dependence is being investigated in Subsection \ref{Sec:PS.PD}.
It turns out that the two dependence concepts are generally not connected (Example \ref{Ex.REM.PS.vs.PD} and Figure \ref{Ex.REM.PS.vs.PD.Pic} in the Supplementary Material), but depending on the degree of continuity of \(Y\) and \(\XX\), stronger and stronger connections occur, with the extreme case of equivalence in the presence of continuous data.

\item (\emph{Attainability of the maximum value}). \label{Rem:REM:Main.5}
It should be noted that not every Fr{\'e}chet class is equipped with a completely separated random vector \((\XX,Y)\), with the result that the maximum value for $\Lambda(Y|\XX)$ in such Fr{\'e}chet classes is strictly smaller than $1$; see Example \ref{Ex.REFF.Not1} in the Supplementary Material for an illustration.
\end{enumerate}
\end{remark}

\medskip
\begin{remark}[Discrete predictor variables]~~ \label{Rem:REM:Main.2}
For a discrete predictor vector with finite range, covering the situation of a finite number of treatment groups, $\Lambda$ in \eqref{Def:REM} considerably simplifies:
Suppose \(\XX\) is discrete with finite range \(\{\xx_1 \dots, \xx_m\}\), \(m \geq 2\), such that \(\PP(\XX=\xx_i) = q_i\), \(i \in \{1,\dots,m\}\).
Then $\alpha = 1-\sum_{i=1}^m q_i^2$ and
\begin{equation}\label{formula_discrete_x}
  \Lambda(Y|\XX)
  = \frac{2 \, \sum_{i=1}^{m} \sum_{j = i+1}^{m} q_i \, q_j \, \left( 2 \, \Psi \big(\PP^{Y|\XX=\xx_i}, \PP^{Y|\XX=\xx_j}\big) - 1 \right)^2}
           {1-\sum_{i=1}^m q_i^2} \,.
\end{equation}
In particular, if the predictor vector takes on only two distinct values, then \(\Lambda(Y|\XX)\) mimics the relative effect of two distributions, and can therefore be understood as a true generalization of the relative effect defined in \eqref{Def.Rel.Eff}:
Suppose \(m=2\), then $\alpha = 1-q_1^2-(1-q_1)^2 = 2 \, q_1 (1-q_1)$ and hence
\begin{align}\label{special_REM}
  \Lambda(Y|\XX)
  & = \left( 2\, \Psi \big(\PP^{Y|\XX=\xx_1}, \PP^{Y|\XX=\xx_2}\big) - 1 \right)^2 
  \\
  & = \left( \Psi \big(\PP^{Y|\XX=\xx_1}, \PP^{Y|\XX=\xx_2}\big) - \Psi \big(\PP^{Y|\XX=\xx_2}, \PP^{Y|\XX=\xx_1}\big) \right)^2\,, \notag
\end{align}
where the second identity is due to \eqref{def_REFF} in the Supplementary Material, noting that both expressions are independent of the choice of \(q_1\) and therefore of the concrete splitting of the \(\xx\)-values.
It should be noted that while for $m = 2$ the value of $\Lambda(Y|\XX)$ does not depend on $q_1$ and $q_2$, i.e.~the distribution of $\XX$, this is not the case in general. This situation is discussed in more detail in Remark \ref{Rem:REM:Main.2App} in the Supplementary Material.
\end{remark}

\subsection{Closed-form expressions} \label{SubSec:Coeff:Normla}

The performance of \(\Lambda\) is now demonstrated by examining the case of two treatment groups, each of which is normally distributed (Example \ref{Ex:REM:BehrensFisher}), and the case of a continuous random vector \((\XX,Y)\) following a multivariate normal distribution (Proposition \ref{Prop:REM:Kendall}).

The first example deals with the so-called \emph{parametric Behrens-Fisher} situation, a well-known example for illustrating the degree of separation (or the so-called `stochastic tendency') between two treatment groups (whose distributions are usually assumed to be independent and Gaussian distributed).

\begin{example}[Behrens-Fisher (BF) situation; Normal distribution]\label{Ex:REM:BehrensFisher}~~
Consider the random variables \(X\) and \(Y\) with $\PP(X = 1) = 1 - \PP(X = 2) = q > 0$ and the two treatment groups represented by the conditional distributions
$\PP^{Y \mid X=1} = N(\mu_1,\sigma^2_1)$ and $\PP^{Y \mid X=2} = N(\mu_2,\sigma^2_2)$.
Then straightforward calculation together with \eqref{special_REM} yields
\begin{equation} \label{Ex:REM:BehrensFisher:F}
  \Lambda(Y|X)
    = \left( 2 \, \Phi\left(\frac{\mu_1-\mu_2}{\sqrt{\sigma_1^2 + \sigma_2^2}}\right) - 1 \right)^2
    = \left( 2 \, \Phi\left(\frac{\mu_2-\mu_1}{\sqrt{\sigma_1^2 + \sigma_2^2}}\right) - 1 \right)^2\,,
\end{equation}
where \(\Phi\) denotes the cdf of the standard normal distribution.
Thus, the value \(\Lambda(Y|X)\) only depends on the absolute difference between the two parameters $\mu_1$ and $\mu_2$ and the variability of the two distributions.
More precisely, \(\Lambda(Y|X) = 0\) if and only if $\mu_1 = \mu_2$, and \(\Lambda(Y|X) < 1\).
Figure \ref{Ex.BehrensFisher.Pic1} illustrates the BF situation by means of two Gaussian distributions with varying degrees of separation.
Figure \ref{Ex.BehrensFisher.Pic} in the Supplementary Material depicts values for \(\Lambda(Y|X)\) and varying \(\sigma = \sigma_1 = \sigma_2 \in \{1,2,3,4\}\).
\end{example}

Example \ref{Ex:REM:BehrensFisher2} in the Supplementary Material provides further illustrating examples of the BF situation with different assumptions on the distributions of $Y|X=1$ and $Y|X=2$.

As for the BF situation and the comparison of two normally distributed treatment groups in Example \ref{Ex:REM:BehrensFisher}, $\Lambda(Y|\XX)$ also has a closed-form expression for \((\XX,Y)\) following a multivariate normal distribution, which is proven as a consequence of Remark \ref{Rem:REM:Kendall}.

\begin{proposition}[Closed-form expression for the multivariate normal distribution]~~\label{Prop:REM:Kendall}\\
Assume \((\XX,Y) \sim \mathcal{N}({\bf 0},\Sigma)\) has positive definite covariance matrix 
\(\Sigma = \left(\begin{smallmatrix} 
  \Sigma_{11} & \Sigma_{12} \\
  \Sigma_{21} & \sigma_Y^2  \end{smallmatrix}\right)\) with $\sigma_Y>0$.
Then 
\begin{equation}\label{Prop:REM:Kendall:Eq}
  \Lambda(Y|\XX)
  = \frac{2}{\pi} \, \arcsin(\rho^2)\,,
\end{equation}
with parameter \(\rho = \sqrt{\Sigma_{21} \Sigma_{11}^{-1} \Sigma_{12} / \sigma_Y^2}\). 
Thus, $\Lambda(Y|\XX)=0$ if and only if $\rho=0$ if and only if $Y$ and $\XX$ are independent, and $\Lambda(Y|\XX) < 1$.
In particular, $Y$ and $\XX$ are independent if and only if $Y$ is stochastically comparable relative to $\XX$.
\end{proposition}{}

Example \ref{Ex:REM:Kendall2} in the Supplementary Material lists further closed-form expressions of \(\Lambda(Y|\XX)\) for various distributional assumptions on \((\XX,Y)\).

\begin{remark}[Illustrating the extent of separation]\label{Rem:Illust.Lambda}
For a better understanding, the extent of separation measured by $\Lambda$ is now explained in more detail by means of two normal distributions with varying difference between mean values (BF situation in Example \ref{Ex:REM:BehrensFisher} and Figure \ref{Ex.BehrensFisher.Pic1}), and by means of a normally distributed random vector $(X,Y)$ (Proposition \ref{Prop:REM:Kendall} and Figure \ref{Ex.Norm.Pic}):
In the left panel of Figure \ref{Ex.BehrensFisher.Pic1} the two distributions / treatment groups strongly overlap with no tendency of one group taking on greater values than the other, i.e.~they are stochastically comparable, whereas in the right panel the two distributions / treatment groups only overlap to a small extent, i.e.~they are strongly separated.
Instead, Figure \ref{Ex.Norm.Pic} covers the situation of comparing an infinite number of conditional distributions $Y|X=x$.
The higher the correlation $\rho$ between the variables $X$ and $Y$, the stronger the degree of separation of $Y$ relative to $X$, ranging from stochastic comparability with no location effect relative to $X$ in the left panel to strong separation in the panel on the right.
\begin{figure}[h]
  \centering
  \includegraphics[scale=0.18]{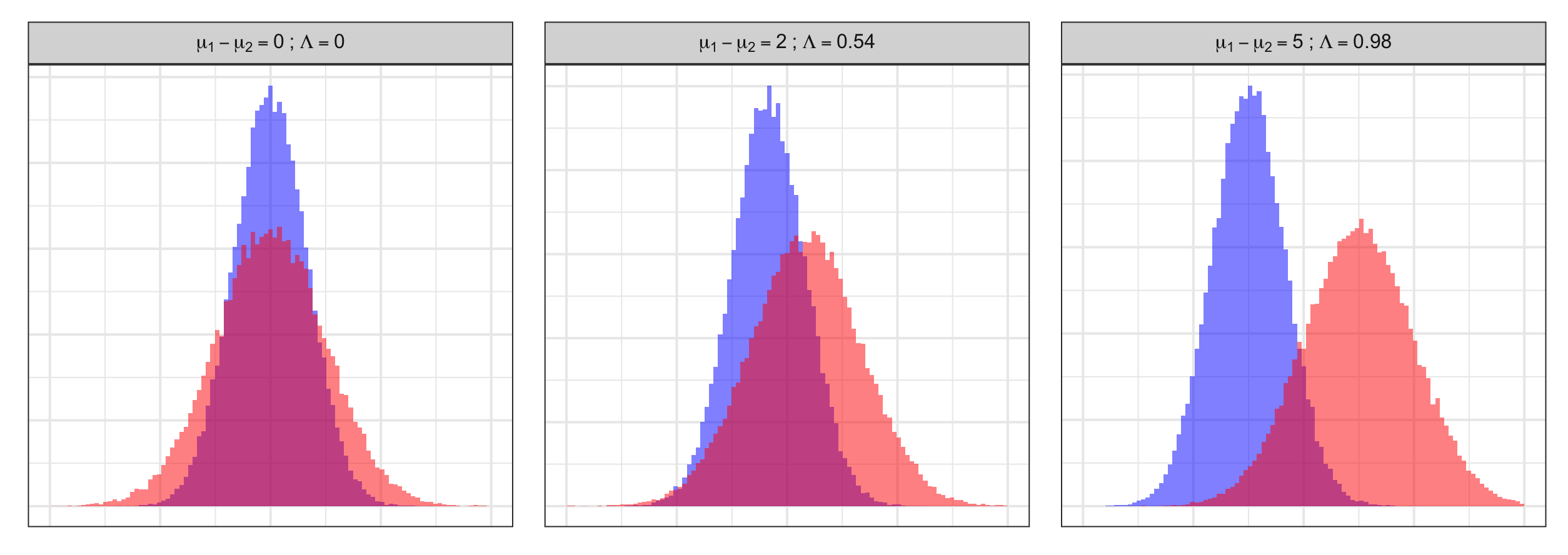}
  \caption{Histograms depicting two normal distributions $Y|X=1 \sim N(\mu_1,2.25)$ (red) and $Y|X=2 \sim N(\mu_2,1)$ (blue) having varying distance between mean values \(\mu_1 - \mu_2 \in \{0, 2, 5\}\). The values $\Lambda(Y|X)$ result from \eqref{Ex:REM:BehrensFisher:F} in Example \ref{Ex:REM:BehrensFisher}.
  }
  \label{Ex.BehrensFisher.Pic1}  
\end{figure}
\begin{figure}[h]
  \centering
  \includegraphics[scale=0.18]{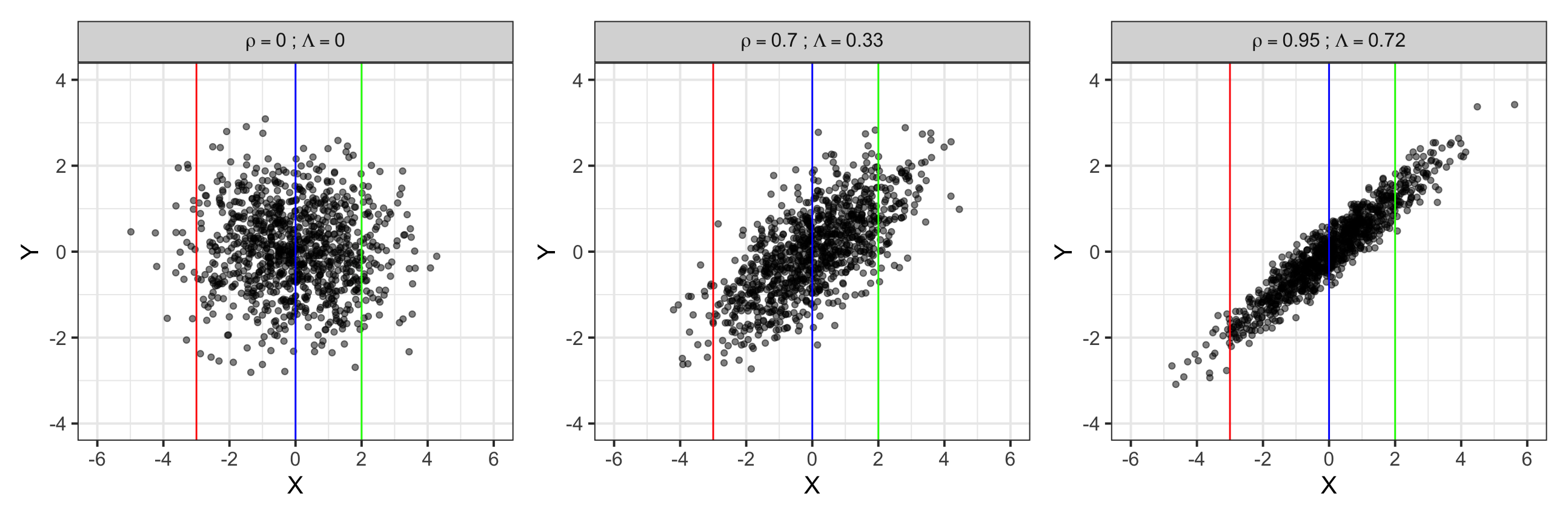}
  \caption{Scatterplots depicting normally distributed random vectors $(X,Y)$ with varying correlation parameter \(\rho \in \{0,0.7,0.95\}\). As representative examples,  the colored lines show the conditional distributions $Y|X=-3$ (red), $Y|X=0$ (blue) and $Y|X=2$ (green). The values $\Lambda(Y|X)$ result from \eqref{Prop:REM:Kendall:Eq} in Proposition \ref{Prop:REM:Kendall}.}
  \label{Ex.Norm.Pic}
\end{figure}
\end{remark}

\subsection{Invariance properties} \label{SubSec:Coeff:Invariance}

The investigation of invariance properties for $\Lambda$ is motivated by the fact that, in combination with the consistency of the proposed estimator in Section \ref{Sec:Estimation}, such transformation of the initial data has no effect on the (asymptotic) dependence value (cf. Remark \ref{Estimator.Remark.General} \eqref{Estimator.Remark.General:1}).

We first show that \(\Lambda(Y|\XX)\) remains unchanged when replacing the random variables by their individual distributional transforms.
Denote by \(F_Z\) the cdf of a random variable \(Z\).\pagebreak

\begin{proposition}[Distributional invariance]\label{Prop.REM:DI}
The map \(\Lambda\) fulfills
$ \Lambda(Y|\XX) = \Lambda \big(F_Y(Y)| {\bf F}_\XX(\XX) \big) $
where ${\bf F}_\XX = (F_{X_1}, \dots, F_{X_p})$.
\end{proposition}

Another important property of \(\Lambda(Y|\XX)\) is its invariance under strictly increasing transformations of \(Y\) and under bijective transformations of \(\XX\).

\begin{proposition}[Invariance under strictly increasing / bijective transformations]\label{Prop.REM:InvB}~~\\
The map \(\Lambda\) fulfills 
$ \Lambda(Y|\XX) = \Lambda(g(Y)|\bold{h}(\XX)) $
for every strictly increasing transformation \(g: \mathbb{R} \to \mathbb{R}\) and every bijective transformation \(\bold{h}: \R^p \to \R^p\). 
\end{proposition}

\begin{remark}\label{Rem:REM:Inv}
\begin{enumerate}
\item (Invariance in the predictor vector).
Invariance of $\Lambda(Y|\XX)$ under bijective transformations of $\XX$ reflects the idea that the value $\Lambda(Y|\XX)$ only depends on the information that is contained in the $\sigma$-algebra generated by $\XX$.

\item (Invariance in the response). 
In contrast, bijective transformations of $Y$ allow to directly influence the extent of separation of $Y$ relative to $\XX$, making such invariance not a desirable property:  
For an illustration, consider the discrete random vector $(X,Y)$ with 
$\PP(X=1,Y=1) = \PP(X=1,Y=2) = \PP(X=2,Y=3) = 1/3$ and a bijection $h$ mapping $1$ to $1$, $2$ to $3$ and $3$ to $2$.
Then, $(X,Y)$ is completely separated (i.e. $\Lambda(Y|X) = 1$), while $(X,h(Y))$ is stochastically comparable (i.e. $\Lambda(h(Y)|X) = 0$).
To round off the discussion, it is worth mentioning that the stated invariance of $\Lambda(Y|\XX)$ in the response is a by-product resulting from the construction of $\Lambda$.
\end{enumerate}
\end{remark}

\subsection{Complete separation and perfect dependence} \label{Sec:PS.PD}

The relationship between complete separation and perfect dependence is being investigated as announced in Remark \ref{Rem:REM:Main2}\;\eqref{Rem:REM:Main.4}:
The two dependence concepts are generally not connected (Example \ref{Ex.REM.PS.vs.PD} in the Supplementary Material), but depending on the degree of continuity of \(Y\) and \(\XX\), stronger and stronger connections occur (Theorems \ref{Thm1.REM.PS.vs.PD} and \ref{Thm2.REM.PS.vs.PD}), with the extreme case of equivalence in the presence of continuous data (Corollary \ref{Cor.REM.PS.vs.PD}).

We first recall two simple observations: there exist Fr{\'e}chet classes with no completely separated representative (Example \ref{Ex.REFF.Not1} in the Supplementary Material) and, in general, neither perfect dependence implies complete separation nor vice versa (Example \ref{Ex.REM.PS.vs.PD} in the Supplementary Material). 
Interestingly, unlike the general case, 
if one of the coordinates of \(\XX\) has a continuous cdf, then complete separation implies perfect dependence.

\begin{theorem}[Continuous predictor variables]~~\label{Thm1.REM.PS.vs.PD}
Consider $(\XX,Y)$ and suppose that one of the coordinates of \(\XX\) has a continuous cdf.
Then, the following assertions hold:
\begin{enumerate}[(i)]
\item \label{Thm1.REM.PS.vs.PD:I1}
There exists a random vector \((\XX^\circ,Y^\circ)\) with \(\XX^\circ \eqd \XX\) and \(Y^\circ \eqd Y\) such that \(Y^\circ\) perfectly depends on \(\XX^\circ\).

\item \label{Thm1.REM.PS.vs.PD:I2}
A random vector \((\XX^\circ,Y^\circ)\) with \(\XX^\circ \eqd \XX\) and \(Y^\circ \eqd Y\) such that \(Y^\circ\) is completely separated relative to \(\XX^\circ\) does not necessarily exist.

\item \label{Thm1.REM.PS.vs.PD:I3}\label{Thm1.REM.PS.vs.PD:I4}
If \(Y\) is completely separated relative to \(\XX\),
then \(Y\) perfectly depends on \(\XX\) and its cdf is continuous.

\item \label{Thm1.REM.PS.vs.PD:I5}
Perfect dependence of \(Y\) on \(\XX\) does not imply complete separation of \(Y\) relative to \(\XX\).
\end{enumerate}
\end{theorem}

If instead the response has a continuous cdf, then there always exists a representative in the Fr{\'e}chet class that is completely separated, and perfect dependence implies complete separation.

\begin{theorem}[Continuous response variable]~~\label{Thm2.REM.PS.vs.PD}
Consider $(\XX,Y)$ and suppose that \(Y\) has a continuous cdf.
Then, the following assertions hold:
\begin{enumerate}[(i)]
\item \label{Thm2.REM.PS.vs.PD:I1}
There exists a random vector \((\XX^\circ,Y^\circ)\) with \(\XX^\circ \eqd \XX\) and \(Y^\circ \eqd Y\) such that \(Y^\circ\) is completely separated relative to \(\XX^\circ\).

\item \label{Thm2.REM.PS.vs.PD:I2}
A random vector \((\XX^\circ,Y^\circ)\) with \(\XX^\circ \eqd \XX\) and \(Y^\circ \eqd Y\) such that \(Y^\circ\) perfectly depends on \(\XX^\circ\) does not necessarily exist.

\item \label{Thm2.REM.PS.vs.PD:I3}\label{Thm2.REM.PS.vs.PD:I4}
If \(Y\) perfectly depends on \(\XX\),
then \(Y\) is completely separated relative to \(\XX\) and one of the coordinates of $\XX$ has a continuous cdf.

\item \label{Thm2.REM.PS.vs.PD:I5}
Complete separation of \(Y\) relative to \(\XX\) does not imply perfect dependence of \(Y\) on \(\XX\).
\end{enumerate}
\end{theorem}

A combination of Theorems \ref{Thm1.REM.PS.vs.PD} and \ref{Thm2.REM.PS.vs.PD} yields the following far reaching result.

\begin{corollary}[Continuous response and predictor variables]~~\label{Cor.REM.PS.vs.PD}
Consider $(\XX,Y)$ and suppose that \(Y\) and one of the coordinates of $\XX$ have a continuous cdf.
Then, 
\begin{enumerate}[(i)]
\item 
There exists a random vector \((\XX^\circ,Y^\circ)\) with \(\XX^\circ \eqd \XX\) and \(Y^\circ \eqd Y\) such that \(Y^\circ\) is completely separated relative to \(\XX^\circ\).

\item 
\(Y\) is completely separated relative to \(\XX\) if and only if \(Y\) perfectly depends on \(\XX\).
\end{enumerate}
\end{corollary}

Under the continuity assumption of Corollary \ref{Cor.REM.PS.vs.PD}, $\Lambda(Y|\XX) = 1$ if and only if $Y$ perfectly depends on \(\XX\).
This motivates interpreting $\Lambda$ in this case as a measure quantifying the degree of functional dependence of $Y$ on $\XX$ (cf. Subsection \ref{SubSec:PS.PD:MoD}).

\subsection{Concordance representation} \label{SubSec:Coeff:Kendall}

Yet another representation of $\Lambda$, which will be of decisive importance for its estimation, is presented in Theorem \ref{Thm:REM:Kendall}. 
It states that \(\Lambda(Y|\XX)\) resembles the difference between the probability of concordance and the probability of discordance of the pair \((Y,Y')\) with $Y$ and $Y'$ having same conditional distribution and being conditionally independent given $\XX$, i.e.
\begin{align}\label{DefMarkovProduct}
  (Y|\XX = \xx) \stackrel{d}{=} (Y'|\XX=\xx) \textrm{ for } \PP^\XX\textrm{-almost all } \xx \in \mathbb{R}^p \textrm{ and } Y \perp Y' \,|\, \XX\,.
\end{align}
We refer to the reduced vector \((Y,Y')\) as the \emph{Markov product} of \((\XX,Y)\) \citep{fuchs2023JMVA}.
Recall that the difference between the probability of concordance and the probability of discordance can be determined using Kendall's tau  (see, e.g.~\citep{fuchs2021SPP}).

\begin{theorem}[Concordance representation]~~\label{Thm:REM:Kendall}
Let \(Y'\) be as in \eqref{DefMarkovProduct}.
Then
\begin{equation}\label{Thm:REM:Kendall.Eq1}
  \Lambda(Y|\XX)
  = \alpha^{-1} \Big[ \PP \big( (Y_1 - Y_2) (Y_1' - Y_2') > 0 \big) 
             - \PP \big( (Y_1 - Y_2) (Y_1' - Y_2') < 0 \big) \Big]
\end{equation}
where \((Y_1,Y_1')\) and \((Y_2,Y_2')\) are independent copies of \((Y,Y')\). 
\end{theorem}

\begin{remark}[Continuous case]\label{Rem:REM:Kendall}
According to Theorem \ref{Thm:REM:Kendall},
for \((\XX,Y)\) with a continuous cdf and connecting copula \(A\),
\(\Lambda(Y|\XX)\) resembles Kendall's tau \(\tau\) of the transformed copula \(\psi(A)\) that is associated with the Markov product \((Y,Y')\), i.e.~$\Lambda(Y|\XX)= \tau (\psi(A))$. 
We refer to \citep{emura2021, limbach2024, shih2024} for previous work on applying Kendall's tau to $(Y,Y')$ in the case of continuous data and $p=1$, motivating the use of Kendall's tau for quantifying the degree of functional dependence of $Y$ on $X$. We address this type of use of $\Lambda$ in Subsection \ref{SubSec:PS.PD:MoD} below and demonstrate far-reaching additional properties.
\end{remark}

\subsection{$\Lambda$ as a measure of functional dependence} \label{SubSec:PS.PD:MoD}

Viewing \(\Lambda\) as a measure of functional dependence in the sense of Corollary \ref{Cor.REM.PS.vs.PD}, i.e. $\Lambda(Y|\XX) = 1$ if and only if $Y$ perfectly depends on \(\XX\), motivates investigating additional desirable properties such as the information gain inequality (Proposition \ref{Thm.REM:IGI_CI}) and the data processing inequality (Corollary \ref{Cor.REM:DPI}).
Taking the situation in Corollary \ref{Cor.REM.PS.vs.PD} as a basis, we require the following continuity assumption:

\begin{assumption} \label{Assumption.1}
The random vector $(\XX,Y)$ meets the condition that \(Y\) and one of the coordinates of $\XX$ have a continuous cdf.
\end{assumption}

Under Assumption \ref{Assumption.1} (cf.~\citep{emura2021,limbach2024,shih2024} for $p=1$)
\begin{enumerate}[(i)]
\item \(0 \leq \Lambda(Y|\XX) \leq 1\).
\item \(\Lambda(Y|\XX)=0\) if \(Y\) and \(\XX\) are independent.
\item \(\Lambda(Y|\XX)=1\) if and only if \(Y\) is perfectly dependent on \(\XX\).
\end{enumerate} 
In addition, $\Lambda$ fulfills the information gain inequality and the conditional independence property as follows:

\begin{proposition}[Information gain inequality; conditional independence property]~~\label{Thm.REM:IGI_CI}\\
Under Assumption \ref{Assumption.1},
\begin{enumerate}[(i)]
\item (Information gain inequality) 
$\Lambda(Y|\XX) \leq \Lambda(Y|(\XX,\ZZ))$
holds for all \(\XX,\ZZ\) and \(Y\).

\item (Conditional independence property)
$\Lambda(Y|\XX) = \Lambda(Y|(\XX,\ZZ))$
holds for all \(\XX,\ZZ\) and \(Y\) such that $Y \perp \ZZ \mid \XX$.
\end{enumerate}    
\end{proposition}

The information gain inequality reflects the idea that additional information in terms of a larger number of predictor variables increases the extent of functional dependence.
The conditional independence property further states that the degree of functional dependence remains unchanged if no additional information is provided through adding further predictor variables. 
We refer to \citep{chatterjee2021AoS, deb2022, ansari2023MFOCI, fuchs2023JMVA, strothmann2022} for more information on these properties in the context of variable selection methods based on measures of functional dependence.

The data processing inequality \citep{cover2006} implies that a transformation of the predictor variables cannot improve the predictability, and self-equitability \citep{kinney2014} states that ``the statistic should give similar scores to equally noisy relationships of different types'' \citep{reshef2011}.

\begin{corollary}[Data processing inequality; self-equitability]~~\label{Cor.REM:DPI}
Under Assumption \ref{Assumption.1},
\begin{enumerate}
\item (Data processing inequality)
$ \Lambda(Y|{\bf h}(\XX)) \leq \Lambda(Y|\XX) $
holds for all $\XX,Y$ and all measurable functions ${\bf h}$ for which \(\alpha({\bf h}(\XX)) = 1\). 

\item (Self-equitability)
$ \Lambda(Y|{\bf h}(\XX)) = \Lambda(Y|\XX) $
holds for all $\XX,Y$ and all measurable functions ${\bf h}$ for which \(\alpha({\bf h}(\XX)) = 1\) and such that $Y \perp \XX \mid {\bf h}(\XX)$. 
\end{enumerate}    
\end{corollary}

Slightly more general versions of Proposition \ref{Thm.REM:IGI_CI} and Corollary \ref{Cor.REM:DPI} not assuming continuity are presented in Proposition \ref{Prop.REM:IGI_CI} and Corollary \ref{CorApp.REM:DPI} in the Supplementary Material.

Under Assumption \ref{Assumption.1}, $\Lambda$ hence qualifies as a measure of functional dependence, making it suitable for a variable selection method that is based on the degree of functional dependence of $Y$ on $\XX$. Section \ref{App.Sim.VarSel} in the Supplementary Material presents a simulation study in which $\Lambda$ is used as the underlying dependence measure in a model-free feature ranking and variable selection procedure.

\section{Continuity} \label{Sec:Cont}

The focus now lies on continuity of $\Lambda$ and thus on the robustness of $\Lambda(Y|\XX)$ against small perturbations of the distribution of $(\XX,Y)$.
Since $\Lambda$ evaluates conditional distributions, it is not continuous with respect to weak convergence of the unconditional distributions (see Proposition \ref{Cont.Prop.NWC} which mimics the situation in \cite[Corollary 1.1]{buecher2024} for Chatterjee's correlation coefficient). 
Therefore, continuity of $\Lambda$ requires a stronger notion of convergence, and a suitable candidate turns out to be the concept of conditional weak convergence introduced in \citep{Sweeting-1989}, which, together with weak convergence of the unconditional distributions, implies weak convergence of the corresponding Markov products as defined in \eqref{DefMarkovProduct}; see \citep{ansari2025Cont}.
The latter condition together with a convergence of the ranges of $\XX$ and $Y$ then yields the desired continuity of $\Lambda(Y|\XX)$.

We first demonstrate that, in general, $\Lambda$ fails to be continuous with respect to weak convergence of the unconditional distributions.

\begin{proposition} \label{Cont.Prop.NWC}
Suppose $(X,Y)$ has independent and continuous marginal cdfs and let $\lambda \in [0,1]$.
Then $\Lambda(Y|X) = 0$, however, there exists a sequence $(X_n,Y_n)_{n \in \N}$ with continuous marginal cdfs weakly converging to $(X,Y)$ and such that $\Lambda(Y_n|X_n) = \lambda$ for all $n \in \N$.
\end{proposition}

For a cdf \(F\), we denote by \(F^{-1}\) its pseudo-inverse, i.e., 
$F^{-1}(t):=\inf \{x \in \mathbb{R} \, : \, F(x) \geq t\}$.

We now establish conditions for the convergence of the normalizing constant $\alpha(\XX)$ in \eqref{Def:REM}. 
The first result requires the presence of at least one continuous marginal cdf.

\begin{proposition} \label{Cont.Lemma.1C:1}
Consider the random vector $\XX = (X_1,...,X_p)$ and a sequence of random vectors \((\XX_n = (X_{1,n},...,X_{p,n}))_{n\in\mathbb{N}}\). If there is some \(k \in \{1,\dots,p\}\) such that
\begin{enumerate}[(i)]
\item \(F_{X_k}\) is continuous and 
\item \(F_{X_{k,n}} \circ F_{X_{k,n}}^{-1}(t) \to t = F_{X_k} \circ F_{X_k}^{-1}(t)\) for \(\lambda\)-almost all \(t \in (0,1)\),
\end{enumerate}
then $\lim_{n \to \infty} \alpha(\XX_n) = \alpha(\XX)$.
\end{proposition}

We can relax the continuity assumption in Proposition \ref{Cont.Lemma.1C:1} at the cost of weak convergence of the joint distribution denoted by $\xrightarrow{d}$.

\begin{proposition} \label{Cont.Lemma.1C}
Consider the random vector $\XX = (X_1,...,X_p)$ and a sequence of random vectors \((\XX_n = (X_{1,n},...,X_{p,n}))_{n\in\mathbb{N}}\) with 
\begin{enumerate}[(i)]
    \item \(\XX_n \xrightarrow{d} \XX\) and
    \item for each \(k \in \{1,\dots,p\}\), \(F_{X_{k,n}} \circ F_{X_{k,n}}^{-1} (t) \to F_{X_k} \circ F_{X_k}^{-1}(t)\) for \(\lambda\)-almost all \(t \in (0,1)\).
\end{enumerate}
Then $\lim_{n \to \infty} \alpha(\XX_n) = \alpha(\XX)$.
\end{proposition}

Propositions \ref{Cont.Lemma.1C:1} and \ref{Cont.Lemma.1C} provide sufficient conditions for the convergence of the denominator.
For convergence of the numerator, we shall need the following definition:
For \(d\in \N\,,\) denote by \(\cU(\R^d)\) a class of bounded, continuous, weak convergence-determining functions mapping from \(\R^d\) to \(\C\,.\) \pagebreak
A sequence \((f_n)_{n\in \N}\) of functions mapping from \(\R^d\) to \(\C\) is said to be \emph{asymptotically equicontinuous} on an open set \(V\subset \R^d\,,\) if for all \(\varepsilon>0\) and \(\xx \in V\) there exist \(\delta(\xx,\varepsilon)>0\) and \(n(\xx,\varepsilon)\in \N\) such that whenever \(|\xx'-\xx|\leq \delta(\xx,\varepsilon)\) then \(|f_n(\xx')-f_n(\xx)|<\varepsilon\) for all \(n > n(\xx,\varepsilon)\,.\) Further, \((f_n)_{n\in \N}\) is said to be \emph{asymptotically uniformly equicontinuous} on \(V\) if it is asymptotically equicontinuous on \(V\) and the constants \(\delta(\varepsilon)=\delta(\xx,\varepsilon)\) and \(n(\varepsilon)=n(\xx,\varepsilon)\) do not depend on \(\xx\,.\)

The following result provides sufficient conditions for continuity of $\Lambda$ and is based on \cite[Theorem 3.1]{ansari2025Cont} and a characterization of conditional weak convergence in \citep{Sweeting-1989}.

\begin{theorem}[Continuity of \(\Lambda\)]~~\label{Cont.Thm.1}
Consider a \((p+1)\)-dimensional random vector \((\XX, Y)\) and a sequence \(( \XX_n, Y_n)_{n\in\mathbb{N}}\) of \((p+1)\)-dimensional random vectors. Let \(V \subseteq \mathbb{R}^p\) be open such that \(\PP(\XX \in V) = 1\). If
\begin{enumerate}[(i)]
\item\label{Cont.Thm.1.A1} \((\XX_n, Y_n) \xrightarrow{d} (\XX, Y)\), and
\item\label{Cont.Thm.1.A2} \((E[u(Y_n)|\XX_n = \xx])_{n\in\mathbb{N}}\) is asymptotically equicontinuous on \(V\) for all \(u \in \mathcal{U}(\mathbb{R})\), 
\end{enumerate}
then 
\begin{enumerate}[(i)]\setcounter{enumi}{2}
\item\label{Cont.Thm.1.A12} the sequence of Markov products \((Y_n,Y_n')_{n \in \mathbb{N}}\) of \((\XX_n,Y_n)_{n \in \mathbb{N}}\) converges weakly to the Markov product \((Y,Y')\) of \((\XX,Y)\).
\end{enumerate}
If condition \eqref{Cont.Thm.1.A12} holds and, additionally, 
\begin{enumerate}[(i)] \setcounter{enumi}{3}
\item\label{Cont.Thm.1.A3} \(F_{Y_n} \circ F_{Y_n}^{-1}(t) \to F_Y \circ F_Y^{-1}(t)\) for \(\lambda\)-almost all \(t \in (0,1)\), and 
\item\label{Cont.Thm.1.A4} $\alpha(\XX_n) \to  \alpha(\XX)$,
\end{enumerate}
then \(\lim_{n \to \infty} \Lambda(Y_n|\XX_n) = \Lambda(Y|\XX)\).
\end{theorem}

As a direct application of Theorem \ref{Cont.Thm.1}, we now verify continuity of $\Lambda$ within the class of multivariate normal distributions.

\begin{corollary}[Continuity of \(\Lambda\) for the multivariate normal distribution]~~\label{Cont.Cor.MVN}
Suppose \((\XX,Y)\sim N(\boldsymbol{\mu},\Sigma)\) and, for \(n\in \N\), let \((\XX_n,Y_n)\sim N(\boldsymbol{\mu}_n,\Sigma_n)\) with \(\Sigma\) and \(\Sigma_n\) being positive definite. 
If \(\Sigma_n \to \Sigma\) then \(\lim_{n \to \infty} \Lambda(Y_n|\XX_n) = \Lambda(Y|\XX)\).
\end{corollary}

Corollary \ref{Cont.Cor.MVN} is a direct consequence of \citep{ansari2025Cont}, verifying conditions \eqref{Cont.Thm.1.A1},\eqref{Cont.Thm.1.A2} and \eqref{Cont.Thm.1.A3} in Theorem \ref{Cont.Thm.1}, and the continuity of the marginal cdfs from which condition \eqref{Cont.Thm.1.A4} in Theorem \ref{Cont.Thm.1} is obtained.
Following \citep{ansari2025Cont}, continuity of $\Lambda$ can even be achieved in the larger class of elliptical distributions and in the class of $l^1$-norm symmetric distributions, where in the latter case \eqref{Cont.Thm.1.A12}--\eqref{Cont.Thm.1.A4} in Theorem \ref{Cont.Thm.1} are used.

As another direct application of Theorem \ref{Cont.Thm.1}, robustness of $\Lambda(Y|\XX)$ against small pertubations of the response $Y$ is established in Corollary \ref{Cont.Cor.RobustY} and illustrated in a simulation study in Section \ref{Sec.Sim.Sc4} in the Supplementary Material.

\begin{corollary}[Robustness of \(\Lambda\) against small pertubations of the response]~~\label{Cont.Cor.RobustY}
Consider a \((p+1)\)-dimensional random vector \((\XX, Y)\) and a sequence \(( \epsilon_n)_{n\in\mathbb{N}}\) of random variables. If
\begin{enumerate}[(i)]
\item\label{Cont.Cor.RobustY.A1}\label{Cont.Cor.RobustY.A2} \(\epsilon_n \xrightarrow{d} 0\) 
and, for every $n \in \N$, $\epsilon_n  \,\perp\, (\XX,Y) $, 
\item\label{Cont.Cor.RobustY.A3} \(F_{Y+ \epsilon_n} \circ F_{Y+ \epsilon_n}^{-1}(t) \to F_Y \circ F_Y^{-1}(t)\) for \(\lambda\)-almost all \(t \in (0,1)\),
\end{enumerate}
then \(\lim_{n \to \infty} \Lambda(Y + \epsilon_n|\XX) = \Lambda(Y|\XX)\).
\end{corollary}

\begin{remark}[Robustness of \(\Lambda\) in special cases]\label{Cont.Cor.RobustY:Rem} 
Consider a \((p+1)\)-dimensional random vector \((\XX, Y)\) and a sequence \(( \epsilon_n)_{n\in\mathbb{N}}\) of continuous random variables weakly converging to $0$ and such that $\epsilon_n  \,\perp\, (\XX,Y)$ for every $n \in \N$.
\begin{enumerate}
\item 
Corollary \ref{Cont.Cor.RobustY} is directly applicable to random vectors \((\XX,Y)\) with $Y$ having a continuous cdf.
\item 
Robustness of \(\Lambda\) for the BF situation discussed in Example \ref{Ex:REM:BehrensFisher} and Example \ref{Ex:REM:BehrensFisher2} in the Supplementary Material assuming that $Y|X=1$ and $Y|X=2$ have continuous cdfs directly follows from \eqref{special_REM} in Remark \ref{Rem:REM:Main.2}.
\end{enumerate}
\end{remark}

\section{Estimation} \label{Sec:Estimation}
A strongly consistent estimator \(\Lambda_n\) for \(\Lambda\) is introduced which is built upon the concordance representation \eqref{Thm:REM:Kendall.Eq1} given in Theorem \ref{Thm:REM:Kendall} and that relies on a graph-based estimation principle. 
For comparative purposes, a consistent estimator \(\widehat{\Lambda}_n\) for \(\Lambda\) is presented that is based on the classical relative effect estimation.

In the following, consider a $(p+1)$-dimensional random vector $(\XX,Y)$
with i.i.d.~copies $(\XX_1,Y_1)$, $\dots$, $(\XX_n,Y_n)$, $n \geq 2$.
Recall that \(Y\) is assumed to be non-degenerate and that \(\XX=(X_1,\ldots,X_p)\) has at least one non-degenerate coordinate.

As estimator for $\Lambda$ we propose the statistic $\Lambda_n$ given by
\begin{equation} \label{Estimator.Def.Lambda}
  \Lambda_n(Y|\XX)
  := 
       \frac{\binom{n}{2}^{-1} \, \sum_{k < l} \, \sgn(Y_k - Y_l) \, \sgn(Y_{N(k)} - Y_{N(l)})}
            {1 - \binom{n}{2}^{-1} \, \sum_{k < l} \, \mathds{1}_{\{\XX_k = \XX_l\}}}
\end{equation}
where, for each $k \in \{1,\dots,n\}$, the number $N(k)$ denotes the index $l$ such that $\XX_l$ is the nearest neighbor of $\XX_k$ with respect to the Euclidean metric on $\mathbb{R}^p$.
Since there may exist several nearest neighbors of $\XX_k$, ties are broken uniformly at random.

\begin{theorem}[Consistency] \label{Estimator.Thm.Consistency}
It holds that $ \lim_{n \to \infty} \Lambda_n (Y|\XX) = \Lambda(Y|\XX)$ almost surely.
\end{theorem}

\begin{remark} \label{Estimator.Remark.General}
\begin{enumerate}
\item 
The estimation principle underlying \(\Lambda_n\), exploiting the nearest neighbor structure formed by the realizations of the predictor variables, is inspired by the elegant estimation procedure used for Azadkia \& Chatterjee's so-called `simple measure of conditional dependence' \(\xi_n\) introduced in \citep{chatterjee2021AoS}.
We make ample use of this connection in the proof of Theorem \ref{Estimator.Thm.Consistency}.

\item \label{Estimator.Remark.General:2}
It is immediately apparent from the definition in \eqref{Estimator.Def.Lambda} that \(\Lambda_n\) mimics the concordance representation presented in \eqref{Thm:REM:Kendall.Eq1}.
The numerator therefore takes the form of a U-statistic including both the observations of the response variables and their nearest neighbors, and the denominator compensates for the occurrence of ties in the predictor variables.
If there exist at least some $k,l \in \{1,\dots,n\}$ with $\XX_k \neq \XX_l$, then the denominator in \eqref{Estimator.Def.Lambda} is strictly positive and \(\Lambda_n(Y|\XX)\) is well-defined.

\item 
The statistic $\Lambda_n$ can be computed in $O(n\log n)$ time.
This is achieved since nearest neighbors can be determined in $O(n\log n)$ time \citep{friedman} and using the function \emph{cor.fk} in the R package \emph{pcaPP}, which allows the numerator in \eqref{Estimator.Def.Lambda} to be determined in $O(n\log n)$ time \citep{filzmoser_R}. 
In situations where only a few treatment groups occur, we use a modified version of the estimator in \eqref{Estimator.Def.Lambda} that excludes comparisons within a treatment group.
This modification has no effect on the consistency of the estimator.

\item \label{Estimator.Remark.General:1}
The invariance results in Subsection \ref{SubSec:Coeff:Invariance} together with the proven strong consistency in Theorem \ref{Estimator.Thm.Consistency} justify transforming the initial data without changing the (asymptotic) dependence value.
This includes, for example, standardizing the predictor variables to ensure that the nearest neighbor search is performed with comparable variable ranges. Alternatively, this can be achieved by replacing the original observations by their ranks (Proposition \ref{Prop.REM:DI}) or by transforming them by strictly monotone functions (Proposition \ref{Prop.REM:InvB}).
In practical applications, data pre-processing procedures such as standardization or rank transformation ensure more robust outcomes.
\end{enumerate}
\end{remark}

\medskip
\begin{remark}[Asymptotic normality]\label{Estimator.Remark.AN}~~
We conjecture that 
\begin{equation} \label{Estimator.Remark.AN:Eq}
  \frac{\Lambda_n - \E(\Lambda_n)}{\sqrt{\Var(\Lambda_n)}} 
\end{equation}
behaves asymptotically normal. 
At this moment, we do not know how to prove this conjecture.
As mentioned in Remark \ref{Estimator.Remark.General}, the estimation principle underlying \(\Lambda_n\) is related to that of $\xi_n$ in \citep{chatterjee2021AoS}, and asymptotic normality of \(\xi_n\) (and related quantities) has been recently studied in a number of publications \citep{deb2020,han2022b,ansari2023MFOCI,shi2021b}. The latter is achieved either under the assumption of independence between \(\XX\) and \(Y\) or by applying a local limit theorem of mean structure for nearest neighbor statistics \cite[Theorem 3.4]{chatterjee2008}. 
As a consequence, to prove the above conjectured asymptotic normality for the statistic in \eqref{Estimator.Remark.AN:Eq} in full generality, a U-statistic-like version of \cite[Theorem 3.4]{chatterjee2008} is required, which is not yet known.
\end{remark}

\medskip
\begin{remark} [Classical rank-based relative effect estimation]~~\label{Estimator.Remark.ClassicalRE}
For $\XX$ discrete with several observations for each possible realization $\xx_1$ and $\xx_2$, an alternative way of estimating $\Lambda$ is to use the canonical rank-based estimator for the relative effect discussed in \citep{Brunner2019} and plug it into \eqref{special_REM}:
Let $y_{11},\ldots , y_{1n_1}$ be all those realizations of $Y$ for which $\XX = \xx_1$ and let $y_{21},\ldots,y_{2n_2}$ be all those realizations of $Y$ for which $\XX = \xx_2$. 
Let further 
\begin{equation*}
  R_{ik} 
  = \frac{1}{2} + 
      \sum_{j=1}^{2} \sum_{l=1}^{n_j} \left( \mathds{1}_{\{y_{jl} < y_{ik}\}} + \frac{1}{2} \mathds{1}_{\{y_{jl} = y_{ik}\}}\right)
\end{equation*}
be the mid rank of $y_{ik}$ among all $y_{11},\ldots,y_{1n_1},y_{21},\ldots,y_{2n_2}$ and let $\bar{R}_{2 \cdot} = \frac{1}{n_2} \sum_{k = 1}^{n_2} R_{2k}$. 
Then, according to \cite[Result 3.1]{Brunner2019},
\begin{equation} \label{Estimator.Def.Naive}
    \Psi_n(\PP^{Y|\XX = \xx_1}, \PP^{Y|\XX = \xx_2}) 
    := \frac{1}{n_1} \left(\bar{R}_{2\cdot} - \frac{n_2 + 1}{2}\right)
\end{equation}
is both an unbiased and consistent estimator for $\Psi(\PP^{Y|\XX = \xx_1}, \PP^{Y|\XX = \xx_2})$.
Furthermore, under some regularity conditions and under the strict assumption of stochastically comparable treatment groups, i.e. $\Psi(\PP^{Y|\XX = \xx_1}, \PP^{Y|\XX = \xx_2}) = 1/2$, the estimator in \eqref{Estimator.Def.Naive} behaves asymptotically normal \citep[Result 3.21]{Brunner2019} and, even for small sample sizes, a $t$-distribution approximates the sample distribution quite well \citep[Result 3.22]{Brunner2019}.
Therefore, and according to \eqref{formula_discrete_x}, for a $(p+1)$-dimensional random vector $(\XX,Y)$ with i.i.d.~copies $(\XX_1,Y_1)$, $\dots$, $(\XX_n,Y_n)$, $n \geq 2$, a suitable candidate for estimating $\Lambda(Y|\XX)$ in the case of discrete \(\XX\) is
\begin{equation}
  \widehat{\Lambda}_n(Y|\XX) 
  := \frac{2 \, \sum_{i=1}^m \sum_{j=i+1}^m \hat{q_i} \, \hat{q_j} \, \left( 2 \, \Psi_n(\PP^{Y|\XX = \xx_i}, \PP^{Y|\XX = \xx_j}) - 1 \right)^2}{1 - \sum_{i=1}^{m} \hat{q_i}^2}
  \label{eq:rank_based}
\end{equation}
with $\hat{q_i} := \frac{1}{n} \sum_{k=1}^{n} \mathds{1}_{\{\XX_k=\xx_i\}}$ and $\Psi_n(\PP^{Y|\XX = \xx_k}, \PP^{Y|\XX = \xx_l})$ being defined as in \eqref{Estimator.Def.Naive}. 
Then, $\widehat{\Lambda}_n (Y|\XX)$ consistently estimates $\Lambda(Y|\XX)$.\pagebreak
\\
Caution is required for the case when one of the coordinates of $\XX$ has a  continuous cdf, since then $n_1 = n_2 = 1$ and therefore $\Psi_n(\PP^{Y|\XX = \xx_1}, \PP^{Y|\XX = \xx_2}) \in \{0, 1\}$, and hence $(2 \Psi_n(\PP^{Y|\XX = \xx_1}, \PP^{Y|\XX = \xx_2}) - 1)^2 = 1$. 
Thus, when attempting to estimate $\Lambda(Y|\XX)$ in this situation via $\widehat{\Lambda}_n(Y|\XX)$ in \eqref{eq:rank_based} one would always obtain a value of $\widehat{\Lambda}_n(Y|\XX) = 1$, making this approach inappropriate for estimating $\Lambda(Y|\XX)$.
\end{remark}

\section{Simulation study} \label{Sec.Sim}
In this section we use synthetic data to evaluate the overall performance and speed of convergence of our estimator \(\Lambda_n\) proposed in \eqref{Estimator.Def.Lambda} in different scenarios (Section \ref{Sec.Sim.Sc1&2}) and when compared to the classical rank based relative effect estimator $\widehat{\Lambda}_n$ proposed in \eqref{eq:rank_based} (Section \ref{Sec.Sim.Sc3}).
Further simulation studies are listed in the Supplementary Material:
In Section \ref{Sec.Sim.Sc4} the robustness of $\Lambda$ against small pertubations of the response variable is illustrated (Scenario 4A) confirming the theoretical findings in Corollary \ref{Cont.Cor.RobustY} and a data example with predictor variables $(X,Z)$ is presented, which shows that the pairwise stochastic comparability of $Y$ relative to $X$ and $Y$ relative to $Z$ does not imply stochastic comparability of $Y$ relative to $(X,Z)$ (Scenario 4B).
Section \ref{App.B.Sc5} provides a comparison of $\Lambda$ with Chatterjee’s correlation coefficient and reveals their potential for examining heteroscedasticity.
Finally, Section \ref{App.Sim.VarSel} presents a model-free variable selection based on $\Lambda_n$ for continuous data according to Section \ref{SubSec:PS.PD:MoD}.

\subsection{Speed of convergence} \label{Sec.Sim.Sc1&2}
To assess the overall performance and speed of convergence of \(\Lambda_n\) we simulate data from a bivariate normal distribution (Scenario 1), and from two treatment groups with underlying normal distributions as described in Example~\ref{Ex:REM:BehrensFisher} (Scenario 2: BF situation).
Sample sizes in both scenarios are $n \in \{30, 100, 500, 1000, 5000, 10000, 50000\}$. 
For each setting we simulate 1000 runs. 
The figures are deferred to the Supplementary Material.

\medskip
{\bf Scenario 1: } We consider $(X,Y) \sim \mathcal{N}({\bf 0}, \Sigma)$ with 
\(\Sigma = \left(
\begin{smallmatrix} 
  1           & \sigma_{XY} \\
  \sigma_{XY} & 1 
\end{smallmatrix}\right)\) and draw an i.i.d.~sample $(X_1, Y_1), \dots, (X_n, Y_n)$ from $(X,Y)$. 
Figure \ref{fig:sim_1} depicts boxplots of the obtained estimates $\Lambda_n(Y|X)$ for different sample sizes and correlation parameters $\rho \in \{0,0.4,0.75,1\}$, and illustrates the fast convergence of $\Lambda_n(Y|X)$ to the true value $\Lambda(Y|X) = \frac{2}{\pi} \arcsin(\sigma_{XY}^2)$ derived in Proposition~\ref{Prop:REM:Kendall}.

\medskip
{\bf Scenario 2:} In this scenario, we generate two treatment groups by setting $\PP(X=0) = 1/2 = \PP(X=1)$ and analyze the two subscenarios:
\begin{itemize}
\item Subscenario A: We consider $Y|X=0 \sim \mathcal{N}(0,1)$ and $Y|X=1 \sim \mathcal{N}(2,1)$.
\item Subscenario B: We consider $Y|X=0 \sim \mathcal{N}(0, 1)$ and $Y|X=1 \sim \mathcal{N}(2,4)$.
\end{itemize}
For each subscenario we draw an i.i.d.~sample $(X_1, Y_1), \dots, (X_n, Y_n)$ from $(X,Y)$ so that $\frac{n}{2}$ observations belong to the case $X=0$ and $\frac{n}{2}$ observations belong to the case $X=1$.
We use the classical relative effect estimator in \eqref{eq:rank_based} for comparison. The results of Scenario 2 are depicted in Figure~\ref{fig:sim_2}. 
As before, the fast convergence of the estimator $\Lambda_n$ in \eqref{Estimator.Def.Lambda} towards the true value can be observed. 
The two estimators behave similarly, except that the variance of $\widehat{\Lambda}_n$ appears to be slightly smaller compared to that of $\Lambda_n$.

\subsection{Comparison with classical relative effect estimation: From discrete to continuous data} \label{Sec.Sim.Sc3}

\begin{figure}[t!]
    \centering
    \includegraphics[width=1.00\linewidth]{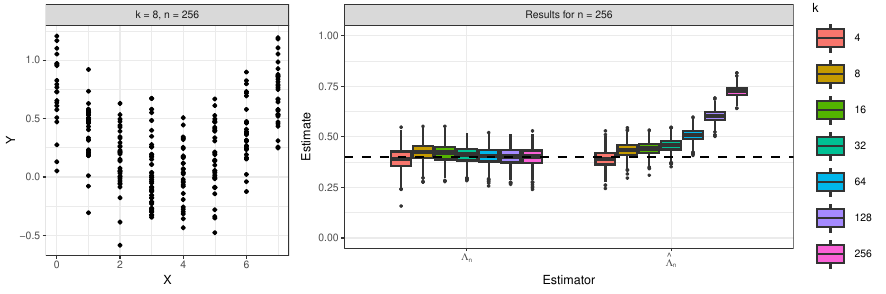}
    \caption{Results of Scenario 3 in Section \ref{Sec.Sim.Sc3} in which variable $X$ is discretized in $k \in \{4, 8, 16, 32, 64, 128, 256\}$ subintervals. An example dataset with $n = 256$ and $8$ subintervals is given in the left panel. The right panel shows boxplots of the obtained estimates $\Lambda_n(Y|X)$ and $\widehat{\Lambda}_n(Y|X)$ for varying $k$. The dashed line indicates the true value $\Lambda(Y|X)$ for the original data without grouping.}
    \label{fig:sim_3}
\end{figure}
{\bf Scenario 3:}  Next we explore how discretization of a continuous predictor variable influences our estimator $\Lambda_n$ and the classical relative effect estimator $\widehat{\Lambda}_n$. 
To this end, we first draw an i.i.d.~sample $X_1, ..., X_n$ from $X \sim \mathcal{U}(-1, 1)$ of size $n=256$ and set $Y_i = X_i^2 + \epsilon_i$ with $\epsilon_1, ..., \epsilon_n \overset{i.i.d}{\sim} \mathcal{N}(0, 0.1^2)$. 
Then we divide the interval $[-1, 1]$ into $k \in \{4, 8, 16, 32, 64, 128, 256\}$ subintervals of equal length and change the original value of the $X_1,...,X_n$ to the number of the subinterval in which the $x$-value of each observation falls. 
The results of Scenario 3 are depicted in Figure~\ref{fig:sim_3}. It can be clearly seen that, as the number of subintervals increases above a certain threshold, the classical relative effect estimator $\widehat{\Lambda}_n$ starts to converge towards $1$, whereas an increase in subintervals does not reduce the precision of $\Lambda_n$.
This behavior is also observed for deviating sample sizes and confirms the theoretical considerations in Remark \ref{Estimator.Remark.ClassicalRE}.
It also speaks against the use of $\widehat{\Lambda}_n$ and in favor of $\Lambda_n$ in case of several treatment groups with only few observations.

\section{Real Data Example} \label{Sec.RealDataEx}

As a dataset to demonstrate the previously described methods we use data from the 2021-2023 National Health and Nutrition Examination Survey (NHANES)~\cite{nhanes2021}. This sample comprises a large number of health related variables from $11,933$ US-citizens. 
We choose age (in years), systolic and diastolic blood pressure (mmHg) from the third measurement and body mass index (BMI, $\frac{kg}{m^2}$) as the variables included in this analysis.  
We filter the dataset for all respondents for which all four variables are present, resulting in $7,416$ complete observations. 
\begin{figure}[t]
    \centering
    \includegraphics[height=0.75\linewidth,width=0.85\linewidth]{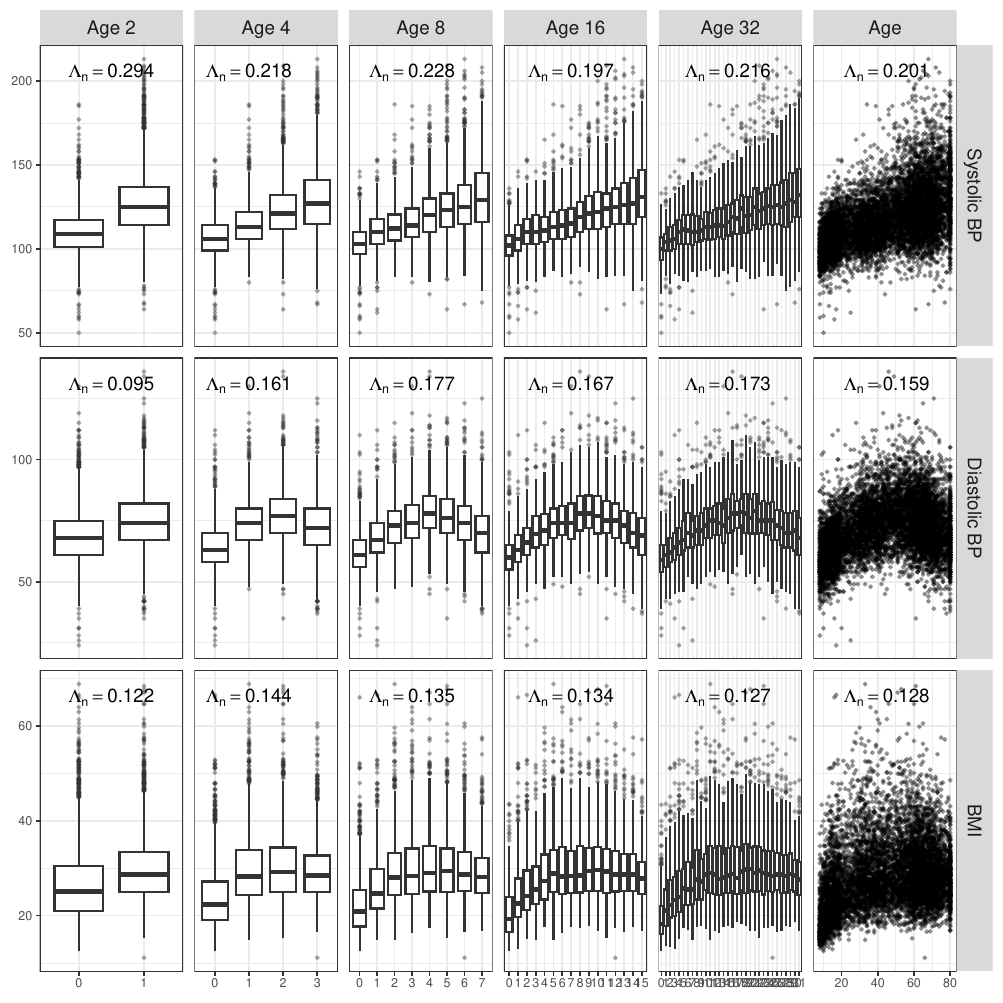}
    \caption{Results of the Nutrition Examination Survey evaluation in Section \ref{Sec.RealDataEx}. 
    The first five columns contain boxplots of the discretized age variables versus the continuous variables Systolic BP, Diastolic BP and BMI. For each level of discretization the estimated values $\Lambda_n$ are given at the top of the panel. 
    The rightmost column depicts the original data and the estimates $\Lambda_n$ in the case of no discretization.}
    \label{fig:pairwise_discrete}
\end{figure}

To study the effects of discretization of continuous variables we group age (ranging from $12$ to $80$ years in our sample) into $2, 4, 8, 16, 32$ equidistant age ranges respectively and refer to these discrete variables as Age 2, Age 4, Age 8, Age 16 and Age 32.
Relationships between all continuous variables (that is, excluding the discretized age variables) as well as the respective estimates $\Lambda_n(Y|X)$ are shown in Figure~\ref{fig:pairwise_plots} in the Supplementary Material. The effect of discretizing age is shown in Figure~\ref{fig:pairwise_discrete}:
Clearly, the discretization into groups compared to the original data set can cause the extent of separation to be both underestimated (the estimates in the case of age to diastolic blood pressure varies between 0.095 in the case of two groups to 0.159 for the original data) and overestimated (the estimates in the case of age to systolic blood pressure varies between 0.294 in the case of two groups to 0.201 for the original data). 
For the sake of completeness, it is noted that, in the case of age to BMI, the estimates remain constant across all discretization variants.
Therefore, when attempting to quantify the extent of the separation, this discovery argues decisively against the discretization of the variables that is often carried out in practice and in favor of using $\Lambda_n$ directly on the original data.
We also conducted permutation tests, testing for deviances of $\Lambda_n(Y|\XX)$ from $0$ under the null hypothesis of equality of all conditional distributions of $Y|X = x$. To this end we simply permuted the $x$-values of the data $N$ times and calculated $\Lambda_n(Y|\XX^{*}_k)$, where $\XX^{*}_k$ is the permuted dataset of $x$-values, $k = 1,...,N$. The permutation p-value was then the number of permuted datasets producing a larger $\Lambda_n(Y|\XX)$ than the original dataset, i.e.
\begin{equation*}
p-value = \frac{|\{k \in \{1,...,N\}:\: \Lambda_n(Y|\XX^{*}_k) \geq \Lambda_n(Y|\XX)\}|}{N}.
\end{equation*}
We ran $10,000$ permutations for each of the pairs of variables depicted in Figure~\ref{fig:pairwise_discrete} but in every case no permutation produced a higher $\Lambda_n(Y|\XX^{*}_k)$ than the original data, leaving us with the conclusion that in each case $p < .0001$.


\section*{Acknowledgment}
The authors gratefully acknowledge the support of the WISS 2025 project 'IDA-lab Salzburg' (20204-WISS/225/197-2019 and 20102-F1901166-KZP). The first and second author further gratefully acknowledge the support of the Austrian Science Fund (FWF) project {P 36155-N} \emph{ReDim: Quantifying Dependence via Dimension Reduction}.


\clearpage

\section*{Supplementary Material}

\begin{appendix}

\section{Additional material for Section \ref{Sec:Coeff}}
\label{App.A}

Example \ref{Ex.REFF.Not1} verifies that there exist Fr{\'e}chet classes with no completely separated representative.

\begin{example}[Fr{\'e}chet classes with no completely separated representative]~~\label{Ex.REFF.Not1} 
Let \(X\) and \(Y\) be two random variables with \(\PP(X=0) = \PP(X=1) = 1/2\) and \(\PP(Y=0) = \PP(Y=1) = \PP(Y=2) = 1/3\)
and joint probabilities given in Table \ref{Ex.REFF.Not1.T}.
\begin{table}[ht] 
  \centering
  \begin{tabular}{|c||c|c||c|}
    \hline
    & \( X=0 \) & \( X=1 \) &  \\
    \hline\hline
    \( Y=0 \) & \( a_{0,0} = \frac{1}{2} - a_{1,0} - a_{2,0}\) & \( a_{0,1} = \frac{1}{3} - a_{0,0}\) & \( \frac{1}{3} \) \\
    \hline
    \( Y=1 \) & \( a_{1,0} \) & \( a_{1,1} = \frac{1}{3} - a_{1,0}\) & \( \frac{1}{3} \) \\
    \hline
    \( Y=2 \) & \( a_{2,0} \) & \( a_{2,1} = \frac{1}{3} - a_{2,0}\) & \( \frac{1}{3} \) \\
    \hline\hline
    & \( \frac{1}{2} \) & \( \frac{1}{2} \) & 1 \\
    \hline 
  \end{tabular}
  \caption{Joint probability distribution of random vector \((X,Y)\) in Example \ref{Ex.REFF.Not1}.}
  \label{Ex.REFF.Not1.T}
\end{table}
Then 
\begin{align*}
    & \Psi \big(\PP^{Y|X=0}, \PP^{Y|X=1} \big)
    \\
    &   =  \int_{\R} \PP(Y < y|X=0) \de \PP^{Y|X=1}(y) + \frac{1}{2} \; \int_{\R} \PP(Y = y|X=0) \de \PP^{Y|X=1}(y)
    \\
    &   =  \PP(Y < 0|X=0) \, 2 \, a_{0,1} + \PP(Y < 1|X=0) \, 2 \, a_{1,1} + \PP(Y < 2|X=0) \, 2 \, a_{2,1} 
    \\
    & \quad + \frac{1}{2} \; \Big( \PP(Y = 0|X=0) \, 2 \, a_{0,1} + \PP(Y = 1|X=0) \, 2 \, a_{1,1} + \PP(Y = 2|X=0) \, 2 \, a_{2,1} \Big) 
    \\
    &   =  \frac{4}{3} \, a_{1,1} + \frac{8}{3} \, a_{2,1} - \frac{1}{6}\,,
\end{align*}
and \eqref{special_REM} finally yields 
\begin{align*}
    \Lambda(Y|X) 
    & 
      = \left( \frac{8}{3} \, a_{1,1} + \frac{16}{3} \, a_{2,1} - \frac{4}{3} \right)^2
      =  \left(\frac{8}{3}\right)^2 (a_{2,1} - a_{0,1})^2
      \leq \left(\frac{8}{9}\right)^2 < 1\,.
\end{align*}
Thus, within this Fr{\'e}chet class there is no completely separated representative.
\end{example}

Remark \ref{Rem:REM:Main.2App} below resumes the discussion in Remark \ref{Rem:REM:Main.2} in which the situation of a finite number of treatment groups is discussed.

\begin{remark}[Discrete predictor variables]~~ 
\label{Rem:REM:Main.2App}
Remark \ref{Rem:REM:Main.2} covers the situation when \(\XX\) is discrete with finite range \(\{\xx_1 \dots, \xx_m\}\), \(m \geq 2\), such that \(\PP(\XX=\xx_i) = q_i\), \(i \in \{1,\dots,m\}\).
It further states that, for $m = 2$, the value of $\Lambda(Y|\XX)$ does not depend on the distribution of $\XX$. 
This is not the case in general. If for example $m = 3$ then~\eqref{formula_discrete_x} simplifies to
\begin{align*}
    \Lambda(Y|\XX) 
    & = \frac{8 [q_1 q_2 (w_{12} - \frac{1}{2})^2 + q_1 q_3 (w_{13} - \frac{1}{2})^2 + q_2 q_3 (w_{23} - \frac{1}{2})^2]}{1 - q_{1}^2 - q_{2}^2 - q_{3}^2}
\end{align*}
where $w_{ij} = \Psi \big(\PP^{Y|\XX=\xx_i}, \PP^{Y|\XX=\xx_j}\big)$. 
Now, choosing
\begin{align*}
      Y|\XX = \xx_1   &\sim \mathcal{U}(0, 1),
    & Y|\XX = \xx_2 &\sim \mathcal{U}((-0.5, 0) \cup (2, 2.5)),
    & Y|\XX = \xx_3 &\sim \mathcal{U}(1, 2),
\end{align*}
gives $w_{12} = w_{23} = 0.5$ and $w_{13} = 1$, further simplifying~\eqref{formula_discrete_x} to
\begin{align*}
    \Lambda(Y|\XX) 
    & = \frac{2 \, q_1 q_3}{1 - q_{1}^2 - q_{3}^2 - (1 - q_1 - q_3)^2}\,.
\end{align*}
For simplicity further setting $q_1 = q_3 = q \in [0,1/2]$ yields
\(\Lambda(Y|\XX) = \frac{q}{2-3q}\)
which is a continuous function in $q$ ranging from $0$ if $q=0$ to $1$ if $q=0.5$. Therefore, in this example, $\Lambda(Y|\XX)$ can take any possible value in $[0, 1]$, depending on $q$. 
\\
From the viewpoint of quantifying the degree of separation of $Y$ relative to $\XX$, this, of course, describes a natural behaviour; however, under the premise of weighting the pairwise relative effects equally, this behaviour appears \emph{paradoxical} and is regularly observed for measures derived from the relative effect. It is due to the relative effect's intransitivity that has been described previously in~\citep{brunner2021ranks}. It can become a problem in statistical analysis when the distribution of groups in a sample varies between experiments and/or is not the same as in the general population of interest. One proposed solution is the use of so-called pseudoranks~\citep{brunner2021ranks, zimmermann2022pseudo} which essentially set $q_i = q = 1/m$ for all $i$, irrespective of the actual distribution of $\XX$. In terms of the functional $\Lambda(Y|\XX)$, this would be equivalent to taking the arithmetic mean over the pairwise relative effects instead of weighting according to the distribution of $\XX$.\end{remark}

The BF situation discussed in Example \ref{Ex:REM:BehrensFisher} is now 
illustrated under alternative assumptions on the distributions of $Y|X=1$ and $Y|X=2$ (Example \ref{Ex:REM:BehrensFisher2}).

\begin{example}[BF situation]\label{Ex:REM:BehrensFisher2}~~
Consider the random variables \(X\) and \(Y\) with $\PP(X = 1) = 1 - \PP(X = 2) = q > 0$ and the two treatment groups represented by the conditional distributions \(\PP^{Y \mid X=1}\) and \(\PP^{Y \mid X=2}\).
\begin{enumerate}[1.]
\item \label{Ex:REM:BehrensFisher2.1} (Uniform distribution)
If $\PP^{Y \mid X=1} = U[0,1]$ and $\PP^{Y \mid X=2} = U[\delta, 1+\delta]$ 
where \(\delta \in [0,1]\), then \eqref{special_REM} yields
\begin{align*}
  \Lambda(Y|X)
  & = \delta^2 \, (2 - \delta)^2\,.
\end{align*}

\item (Bernoulli distribution)
If $\PP^{Y \mid X=1} = B(p_1)$ and $\PP^{Y \mid X=2} = B(p_2)$
where $p_1,p_2 \in [0,1]$, then \eqref{special_REM} yields
\begin{align*}
  \Lambda(Y\mid X) 
  & = \left(p_2-p_1\right)^2\,.
\end{align*}

\item (Exponential distribution)
If $ \PP^{Y \mid X=1} = Exp(\lambda_1)$ and $\PP^{Y \mid X=2} = Exp(\lambda_2)$
where $\lambda_1,\lambda_2 \in (0,\infty)$, then \eqref{special_REM} yields
\begin{align*}
  \Lambda(Y\mid X) 
  & = \left( \frac{\lambda_2-\lambda_1}{\lambda_1+\lambda_2}\right)^2\,.
\end{align*}
\end{enumerate}
\end{example}

\begin{figure}[h]
  \centering
  \includegraphics[scale=0.450]{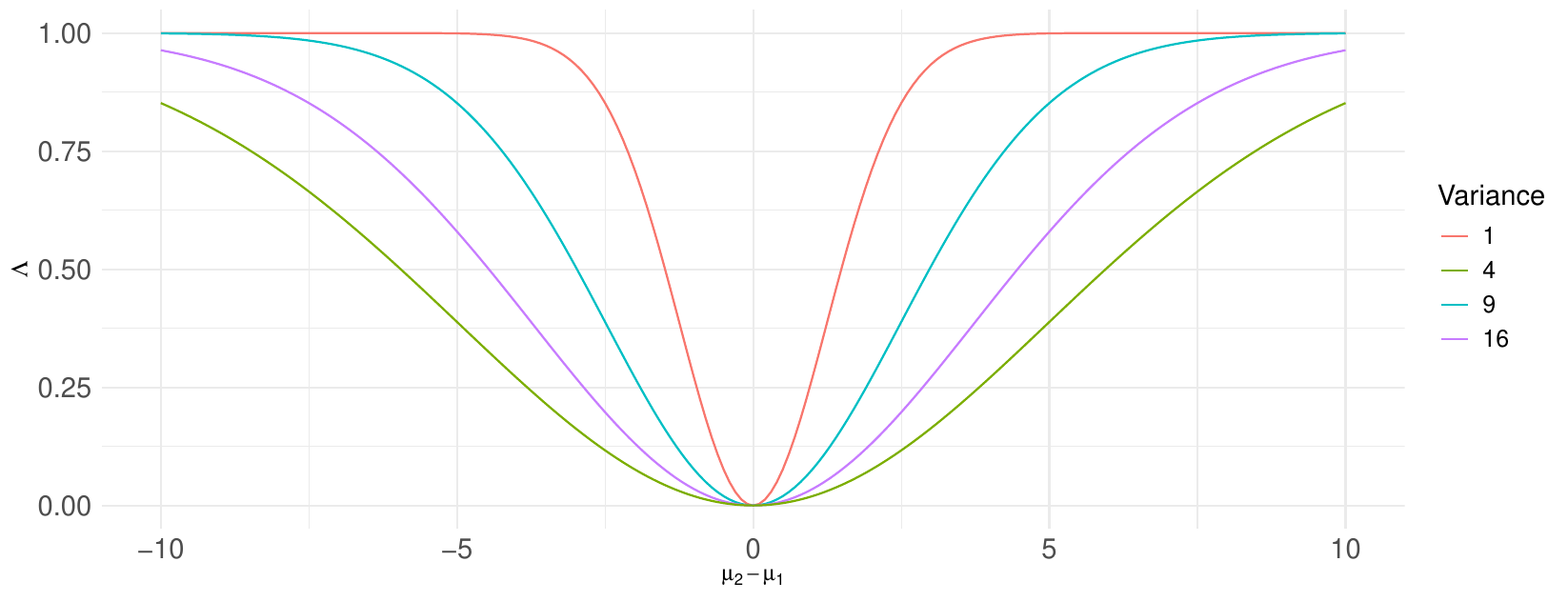}
  \caption{Graph depicting \(\Lambda(Y|X)\) for the BF situation and normally distributed treatment groups $N(\mu_1,\sigma_1^2)$ and $N(\mu_2,\sigma_2^2)$ in Example \ref{Ex:REM:BehrensFisher} as a function of the distance between mean values for various choices of variance \(\sigma^2 = \sigma_1^2 = \sigma_2^2\).}
  \label{Ex.BehrensFisher.Pic}
\end{figure}

Closed-form expressions for \(\Lambda(Y|\XX)\) similar to that in Proposition \ref{Prop:REM:Kendall} but different assumptions on the dependence structure of \((\XX,Y)\) are presented in Example \ref{Ex:REM:Kendall2} below. 
The results are immediate from Remark \ref{Rem:REM:Kendall} in combination with \cite[Example 1]{fuchs2023JMVA}.

\begin{example}[Parametric copula families]\label{Ex:REM:Kendall2}~~
\begin{enumerate}[1.]
\item (Marshall-Olkin copula) 
Assume \(X\) and \(Y\) are continuous with connecting copula $A_{\alpha,\beta}$ that is a Marshall-Olkin copula with parameters $\alpha=1$ and $\beta \in [0,1]$ \citep{durante2016}.
Then $\psi(A_{1,\beta}) = A_{\beta,\beta}$ and 
\begin{align*}
  \Lambda(Y|X)
  & 
	= \tau(A_{\beta,\beta})
	= \frac{\beta}{2-\beta}\,.
\end{align*}
Thus, $\Lambda(Y|X)=0$ if and only if $\beta=0$ if and only if $X$ and $Y$ are independent,
and $\Lambda(Y|X)=1$ if and only if $\beta=1$.
\item (Fr{\'e}chet copula) 
Assume \(X\) and \(Y\) are continuous with connecting copula $A_{\alpha,\beta}$ that is a \emph{Fr{\'e}chet copula} with parameter $\alpha,\beta \in [0,1]$ such that $\alpha+\beta\leq 1$ \citep{durante2016}.
Then $ \psi(A_{\alpha,\beta}) = A_{\alpha^2 + \beta^2,2\alpha\beta} $ and 
\begin{align*}
  \Lambda(Y|X)
  & 
	= \tau(A_{\alpha^2 + \beta^2,2\alpha\beta})
	= \frac{1}{3} \, (\alpha-\beta)^2 \, \big( (\alpha+\beta)^2 + 2 \big)\,.
\end{align*}
Thus, $\Lambda(Y|X)=0$ if and only if $\alpha=\beta$,
and $\Lambda(Y|X)=1$ if and only if $(\alpha,\beta) \in \{\{0,1\},\{1,0\}\}$. 
It is worth mentioning that, if $\alpha=\beta>0$, then $Y$ is stochastically comparable relative to $X$, but $X$ and $Y$ fail to be independent.
\item (EFGM copula) 
Assume \(\XX\) and \(Y\) are continuous with connecting copula $A_\alpha$ that is an EFGM copula with parameter $\alpha \in [-1,1]$ \citep{durante2016}.
Then $\psi(A_\alpha)$ is of EFGM type with parameter $\alpha^2/3^p$, 
and 
\begin{align*}
  \Lambda(Y|\XX)
  & 
	= \tau(A_{\alpha^2/3^p})
	= \frac{2 \, \alpha^2}{3^{p+2}}.
\end{align*}
Thus, $\Lambda(Y|\XX)=0$ if and only if $\alpha=0$ if and only if $X$ and $Y$ are independent, and $\Lambda(Y|\XX) \leq \tfrac{2}{3^{p+2}}$.
\end{enumerate}
\end{example}{}

Example \ref{Ex.REM.PS.vs.PD} verifies that, in general, neither complete separation implies perfect dependence nor vice versa.
Figure \ref{Ex.REM.PS.vs.PD.Pic} depicts scatterplots of the distributions used in Example \ref{Ex.REM.PS.vs.PD}.

\begin{example}[Complete separation \(\neq\) perfect dependence]~~ \label{Ex.REM.PS.vs.PD}
\begin{enumerate}[(i)]
\item 
Consider first the random variables $X$ and $Y$ with $\PP(X = 1) = \PP(X = 2) = 1/2$ and 
\begin{align*}
  \PP^{Y \mid X=1}   & = U([0,1])\,, 
  & \PP^{Y \mid X=2} & = U([1,2])\,.  
\end{align*}
Then, according to Example \ref{Ex:REM:BehrensFisher2} \eqref{Ex:REM:BehrensFisher2.1},
\(\Lambda(Y|X) = 1\), but Chatterjee's rank correlation \citep{chatterjee2020} fulfills \(\xi(Y,X) < 1\).
Thus, due to Theorem \ref{Thm:REM:Main} and \cite[Theorem 1.1]{chatterjee2020}, \(Y\) is completely separated relative to \(X\), but fails to be perfectly dependent on \(X\).

\item 
Now, consider the random variables $X$ and $Y$ with $\PP(Y = 1) = \PP(Y = 2) = 1/2$ and 
\begin{align*}
  \PP^{X \mid Y=1}   & = U([0,1])\,, 
  & \PP^{X \mid Y=2} & = U([1,2])\,.
\end{align*}
Then \(\PP^{X} = U([0,2])\), hence \(\alpha(X)=1\) and \(\Psi \big(\PP^{Y|X = x_1}, \PP^{Y|X = x_2} \big) = 1\) for almost all \(x_1 \in [0,1]\) and \(x_2 \in [1,2]\), and thus
\begin{align*}
  \Lambda(Y|X)
  & = \int_{[0,1] \times [1,2]} \big( 2 \, \Psi \big(\PP^{Y|X=x_1}, \PP^{Y|X=x_2}\big) - 1) \big)^2 \de (\PP^{X} \otimes \PP^{X}) (x_1,x_2) 
  \\
  & \qquad + \int_{[1,2] \times [0,1]} \big( 2 \, \Psi \big(\PP^{Y|X=x_1}, \PP^{Y|X=x_2}\big) - 1 \big)^2 \de (\PP^{X} \otimes \PP^{X}) (x_1,x_2) 
  \\
  & = \int_{[0,1] \times [1,2]} 1 \de (\PP^{X} \otimes \PP^{X}) (x_1,x_2) 
      + \int_{[1,2] \times [0,1]} 1 \de (\PP^{X} \otimes \PP^{X}) (x_1,x_2) 
  \\
  & = \frac{1}{2} < 1\,.
\end{align*}
Instead, Chatterjee's rank correlation \citep{chatterjee2020} equals \(\xi(Y,X) = 1\).
Thus, \(Y\) perfectly depends on \(X\), but, due to Theorem \ref{Thm:REM:Main},
$Y$ fails to be completely separated relative to \(X\).
\end{enumerate}
\begin{figure}[t!]
  \centering
  \vspace{-7mm}
  \includegraphics[scale=0.15]{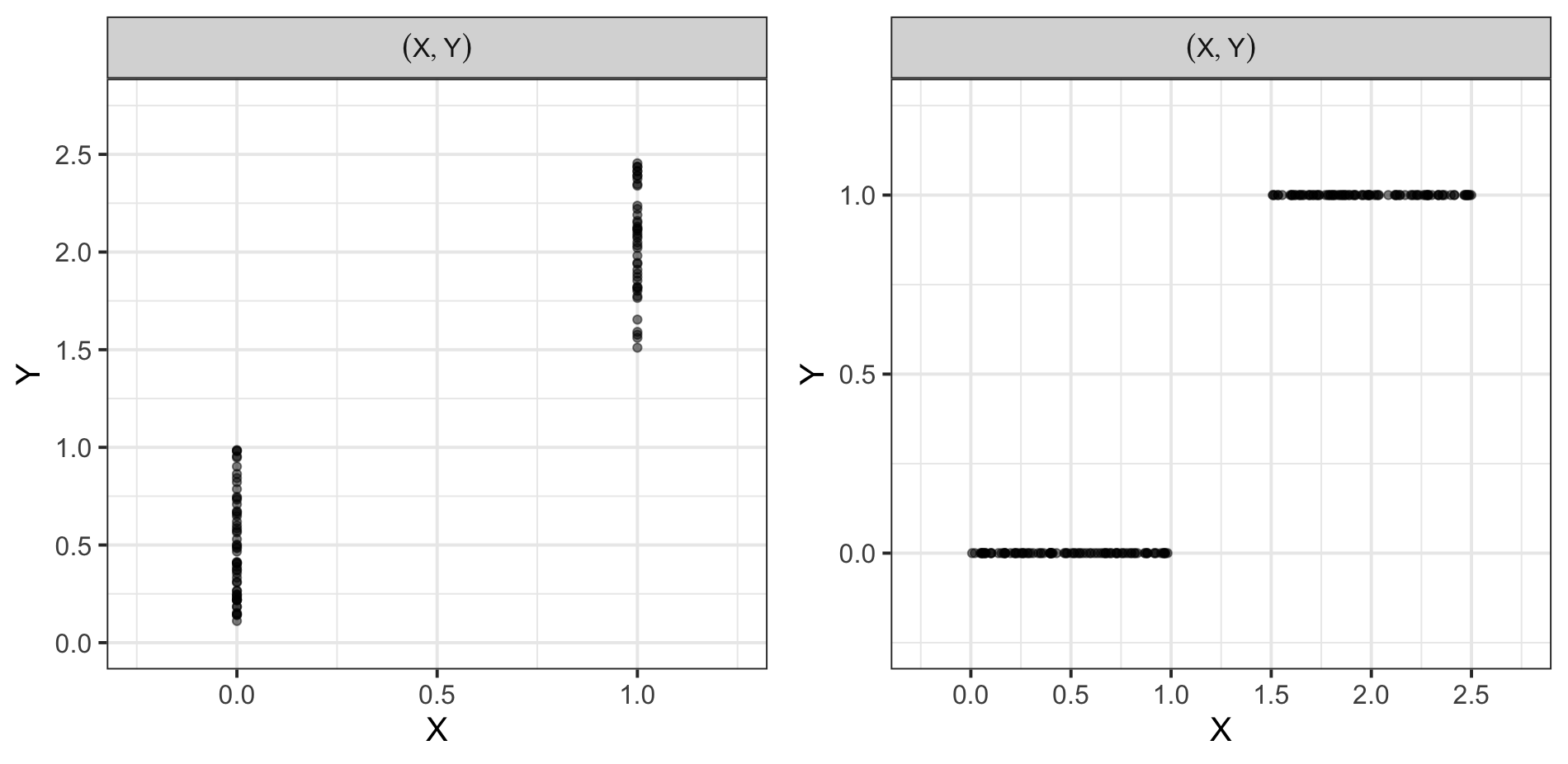}
  \caption{Scatterplots depicting the distributions used in Example \ref{Ex.REM.PS.vs.PD} for verifying that, in general, neither complete separation implies perfect dependence nor vice versa:
  In the left panel $Y$ is completely separated relative to $X$ but fails to be perfectly dependent on $X$.
  In the right panel, instead, $Y$ is perfectly dependent on $X$ but fails to be completely separated relative to $X$.}
  \label{Ex.REM.PS.vs.PD.Pic}
\end{figure}
\end{example}

A slightly more general version of Proposition \ref{Thm.REM:IGI_CI} not assuming continuity is given in Proposition \ref{Prop.REM:IGI_CI} below.

\begin{proposition}[Information gain inequality; conditional independence property]
\label{Prop.REM:IGI_CI}~~
\begin{enumerate}[(i)]
\item (Information gain inequality) 
If \(\alpha(\XX,\ZZ) = \alpha(\XX)\), then the inequality 
\begin{equation}\label{Prop.REM:IGI_CI:Eq1}
  \Lambda(Y|\XX) 
  \leq \Lambda(Y|(\XX,\ZZ))
\end{equation}
holds for all \(\XX,\ZZ\) and \(Y\).

\item (Conditional independence property)
If \(\alpha(\XX,\ZZ) = \alpha(\XX)\), then the identity 
\begin{equation}\label{Prop.REM:IGI_CI:Eq2}
  \Lambda(Y|\XX) 
  = \Lambda(Y|(\XX,\ZZ))
\end{equation}
holds for all \(\XX,\ZZ\) and \(Y\) such that $Y$ and $\ZZ$ are conditionally independent given $\XX$.
\end{enumerate}
\end{proposition}

Recall that \(\alpha(\XX,\ZZ) = \alpha(\XX)\) is fulfilled whenever at least one coordinates of \(\XX\) has a continuous cdf.
However, if \(\alpha(\XX,\ZZ) \neq \alpha(\XX)\), then the information gain inequality in \eqref{Prop.REM:IGI_CI:Eq1} can fail, as the following example demonstrates.

\begin{example}\label{Ex.Prop.IGI}
We resume the situation given in Example \ref{Ex.REM.PS.vs.PD} (i) and add the random variable \(Z\) being continuous and independent of \((X,Y)\).
Then \(\alpha(X,Z) = 1 > 0.5 = \alpha(X)\) and hence 
\begin{align*}
  \Lambda(Y|X)
  & = \frac{0.5}{0.5} =  1 > 0.5 = \frac{0.5}{1} = \Lambda(Y|(X,Z))\,,
\end{align*}
i.e.~the information gain inequality in \eqref{Prop.REM:IGI_CI:Eq1} fails.
Since \(Z\) is independent of \((X,Y)\), we have that \(Y\) and \(Z\) are conditionally independent given \(X\) implying that also \eqref{Prop.REM:IGI_CI:Eq2} fails in this situation.
\end{example}

A slightly more general version of Corollary \ref{Cor.REM:DPI} not assuming continuity is presented in Corollary \ref{CorApp.REM:DPI} below.

\begin{corollary}[Data processing inequality; self-equitability]\label{CorApp.REM:DPI}~~
\begin{enumerate}
\item (Data processing inequality)
The inequality
\begin{align}\label{Eq.DPI}
  \Lambda(Y|{\bf h}(\XX))
  \leq \Lambda(Y|\XX)
\end{align}
holds for all $\XX,Y$ and all measurable functions ${\bf h}$ for which \(\alpha({\bf h}(\XX)) = \alpha(\XX)\).

\item (Self-equitability)
The identity
\begin{align}\label{Eq.SE}
  \Lambda(Y|{\bf h}(\XX)) = \Lambda(Y|\XX)
\end{align}
holds for all $\XX,Y$ and all measurable functions ${\bf h}$ for which \(\alpha({\bf h}(\XX)) = \alpha(\XX)\) and such that $Y$ and $\XX$ are conditionally independent given ${\bf h}(\XX)$. 
\end{enumerate}    
\end{corollary}

However, if \(\alpha({\bf h}(\XX)) \neq \alpha(\XX)\), then the data processing inequality \eqref{Eq.DPI} may fail which is immediate from Example \ref{Ex.Prop.IGI} by letting \(\XX = (X,Z)\) and \({\bf h}(\XX) = X\).

\section{Additional material for Sections \ref{Sec.Sim} and \ref{Sec.RealDataEx}}
\label{App.B}

\begin{figure}[h!]
    \centering
    \includegraphics[width=0.87\linewidth]{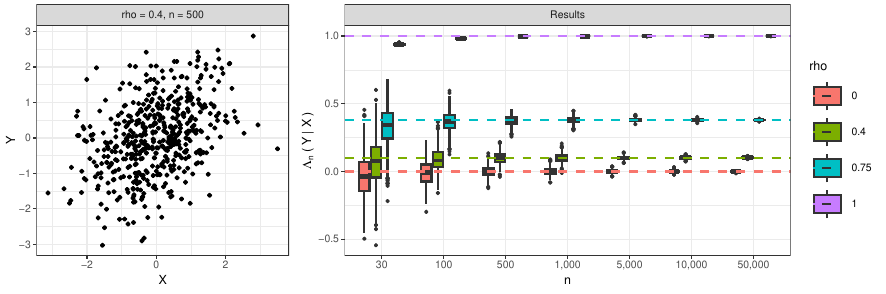}
    \caption{Results of Scenario 1 in Section \ref{Sec.Sim.Sc1&2}. An example dataset with $n = 500$ and $\rho = 0.4$ is given in the left panel. The right panel shows boxplots of the obtained estimates $\Lambda_n(Y|X)$. The dashed lines indicate the true values $\Lambda(Y|X)$.}
    \label{fig:sim_1}
\end{figure}
\begin{figure}[h!]
    \centering
    \includegraphics[width=0.87\linewidth]{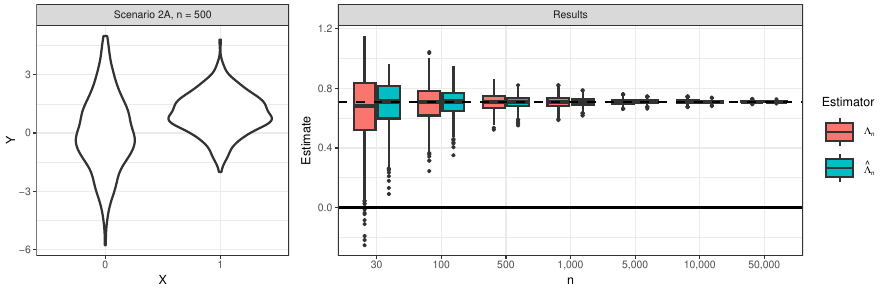}
    \includegraphics[width=0.87\linewidth]{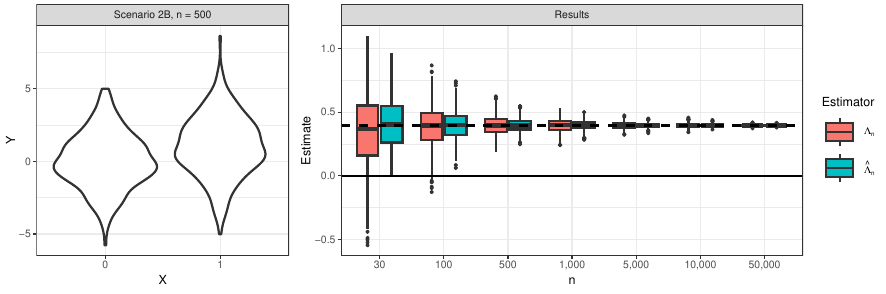}
    \caption{Results of Scenarios 2A (top) and 2B (bottom) in Section \ref{Sec.Sim.Sc1&2}. Example datasets with $n = 500$ are given in the panels on the left. The right panels show boxplots of the obtained estimates $\Lambda_n(Y|X)$ and $\widehat{\Lambda}_n(Y|X)$. The dashed lines indicate the true values $\Lambda(Y|X)$.}
    \label{fig:sim_2}
\end{figure}

\subsection{Multiple predictor variables} \label{Sec.Sim.Sc4}

{\bf Scenario 4:} 
The performance of $\Lambda_n(Y|\XX)$ for multivariate $\XX$ is now examined under the following subscenarios:
\begin{itemize}
    \item Subscenario A (Robustness): 
    We consider the random vector $\XX$ with $\PP(\XX=(k,l)) = 1/4$, $k,l \in \{1,2\}$, and draw an i.i.d.~sample $\XX_1, ..., \XX_n$ from $\XX$.
    Then, for $i \in \{1,\dots,n\}$, we sample $Y_i \sim \mathcal{U}(l(\XX_i), l(\XX_i) + 1) + \sigma \, \mathcal{U}(-1, 1)$ where $l(1, 1) = 0,\, l(1,2) = 1,\, l(2,1) = 2,\, l(2, 2) = 3$. $\sigma$ takes the values in $\{1, \frac{1}{2}, \frac{1}{4}, \frac{1}{8}, 0\}$.
    \item Subscenario B (Stochastic comparability): 
    We consider the independent random vector $(X,Z)$ with $Z \sim \mathcal{U}(0, 2 \pi)$ and $\PP(X=-1) = 1/2 = \PP(X=1)$ and draw i.i.d.~samples $X_1, ..., X_n$ from $X$ and $Z_1, ..., Z_n$ from $Z$.
    Then, for $i \in \{1,\dots,n\}$, we set $Y_{i} = X_{i} \cdot \cos{(X_{i} \cdot Z_{i})} + \epsilon_i$ where $\epsilon_i \overset{i.i.d}{\sim} \mathcal{U}(0, 1)$. 
\end{itemize}
The results of Scenarios 4A and 4B are depicted in Figures \ref{fig:tda1} and \ref{fig:tda2}.
Subscenario 4A covers the situation of complete separation and noise is added to the response variable. It can be clearly seen that the estimates converge quickly and continuously towards the true value $1$ as the noise parameter $\sigma$ decreases, which illustrates and confirms the theoretical findings in Corollary \ref{Cont.Cor.RobustY}.
Subscenario 4B presents a data set in which caution is required as the pairwise comparisons of $(X,Y)$ and $(Z,Y)$ suggest stochastic comparability, however, $Y$ exhibits a high degree of separation relative to $(X,Z)$. Therefore, the determination of the extent of separation of $Y$ relative to $(X,Z)$ cannot be replaced by the pairwise comparisons of $Y$ with the marginals $X$ and $Z$.
\begin{figure}[ht]
    \centering
    \includegraphics[width=1.00\linewidth]{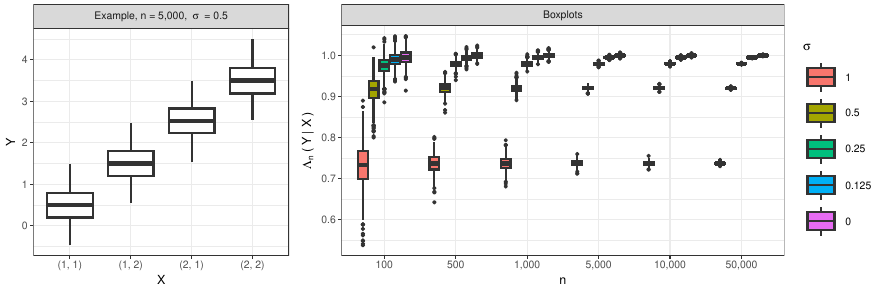}
    \caption{Results of Scenario 4A in Section \ref{Sec.Sim.Sc4}. An example dataset with $n = 5,000$ and noise parameter $\sigma = 0.5$ is given in the left panel. The right panel shows boxplots of the obtained estimates $\Lambda_n(Y|X)$ for varying noise parameter $\sigma \in \{0, 1/8, 1/4, 1/2, 1\}$ and varying sample size.}
    \label{fig:tda1}
\end{figure}

\bigskip
\begin{figure}[ht]
    \centering
    \includegraphics[width=1.00\linewidth]{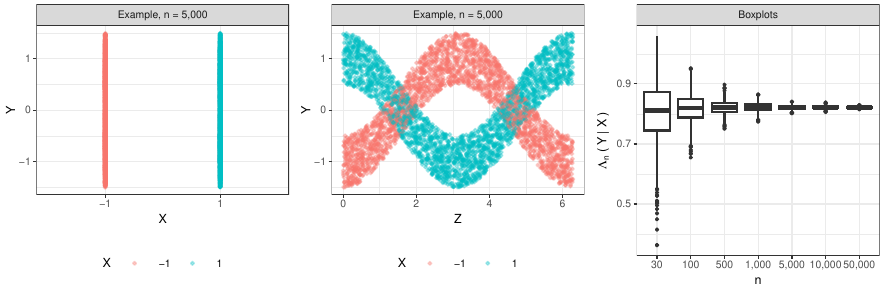}
    \caption{Results of Scenario 4B in Section \ref{Sec.Sim.Sc4}. The left panel shows a scatterplot of $Y$ against $X$, the middle panel of $Y$ against $Z$, both indicating stochastic comparability. The right plot depicts boxplots of the obtained estimates $\Lambda_n(Y|(X, Z))$ for varying sample size.}
    \label{fig:tda2}
\end{figure}

\subsection{Comparison with Chatterjee's correlation coefficient and \\ examination of heteroscedasticity} \label{App.B.Sc5}

{\bf Scenario 5:} 
Since $\Lambda(Y|\XX)$ quantifies the degree of separation of $Y$ relative to $\XX$, 
it captures mostly location effects and is insensitive to scale effects, in so far as they do not influence the extent of separation. 
To illustrate this we create data from four different models and compare the behavior of $\Lambda_n$ with that of Chatterjee's correlation coefficient $\xi_n$ \citep{chatterjee2021AoS}, which directly compares conditional distributions and is thus sensitive to any difference in distribution. 
In all three scenarios we draw an i.i.d.~sample $X_1, \dots, X_n$ from $X \sim \mathcal{U}(0, 1)$.:
\begin{itemize}
    \item Subscenario A: 
    We draw an i.i.d.~sample from $Y \sim \mathcal{U}{(0, 1)}$. 
    \item Subscenario B: 
    For $i \in \{1,\dots,n\}$, we sample $Y_i \sim \mathcal{U}(-|1 - 2X_i|, |1 - 2X_i|)$.
    \item Subscenario C: 
    For $i \in \{1,\dots,n\}$, we set $Y_i = R_{i} \cdot Z_{i}$, where $R_{1},...,R_{n} \overset{i.i.d}{\sim} R$ and $Z_{i} \sim \mathcal{U}(\frac{1}{2} (-X_i^2 + 2X_i)^{0.9}, (-X_i^2 + 2X_i)^{0.9})$.
    Here $R$ is the Rademacher distribution, that is, the discrete probability distribution which takes the values $-1$ and $1$ each with probability $\frac{1}{2}$ and is independent of $Z$.
    \item Subscenario D: Here we first generate data following a simple non-linear regression model. We then fit a linear, and thus misspecified, regression model to the data and determine $\Lambda_n(\epsilon|\XX)$, where $\epsilon_1,...,\epsilon_n$ are the residuals of the linear model.
\end{itemize}
In Subscenario A both coefficients are zero since $X$ and $Y$ are independent and hence neither a location nor a scale effect is present.
In Subscenarios B and C, instead, $Y$ is stochastically comparable relative to $X$, however, since a scale effect is visible the two variables are not independent and knowledge of the value of $X$ reveals information about the value of $Y$ resulting in a strictly positive value of Chatterjee's correlation coefficient.
\begin{figure}[h!]
    \centering
    \includegraphics[width=0.85\linewidth]{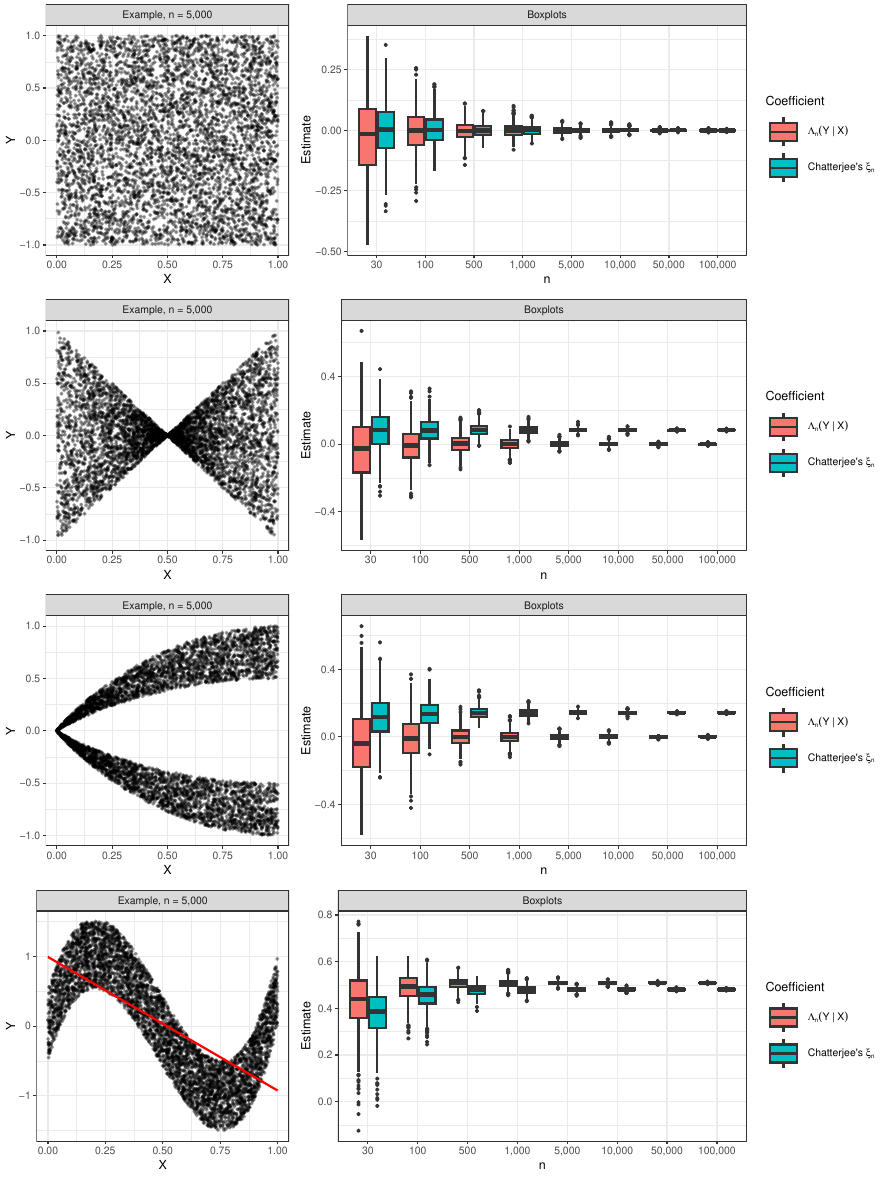}
    \caption{Results of Scenarios 5A (top) - 5D (bottom) in Section \ref{App.B.Sc5}. Example datasets with $n = 5,000$ are given in the left panel. The right panel shows boxplots of the obtained estimates $\Lambda_n(Y|X)$ and Chatterjee's correlation coefficient $\xi_n(X,Y)$ for varying sample sizes.}
    \label{fig:heteroscedasticity}
\end{figure}
Finally, in Subscenario D both coefficients are strictly positive, indicating heteroscedasticity in the residuals, with both location and scale effects.

\subsection{Variable selection} \label{App.Sim.VarSel}

Here we illustrate the use of $\Lambda(Y|\XX)$ as a method for variable selection. To this end we use a high resolution climatological dataset (CHELSA)~\cite{karger2017} containing several climate variables from different locations on the planet. We focus on annual precipitation  (AP\_R), which is our target variable, and the possible predictor variables mean temperature in the warmest (MTWaQ\_R), coldest (MTQC\_R), wettest (MTWeQ\_R) and driest (MTDQ\_R) quarter of the year, as well as the precipitation in the warmest (PWaQ\_R), coldest (PCQ\_R), wettest (PWeQ\_R) and driest (PDQ\_R) quarter of the year.

We conduct both a forward selection as well as a best subsample selection. In the forward selection, starting from an empty set of predictors, we successively add the variable which produces the highest $\Lambda_n(Y|\XX)$ when added to the already selected predictors. The procedure terminates when improvement of $\Lambda_n(Y|\XX)$ is no longer possible. In the best subsample selection we choose among all possible subsamples of the predictors the one which produces the highest $\Lambda_n(Y|\XX)$. Both methods choose exactly the four precipitation parameters with a $\Lambda_n(Y|\XX)$ of $0.9340$. The results of the forward selection are shown in more detail in Table~\ref{tab:forward}. The five best subsets from the best subset selection are shown in Table~\ref{tab:subset}.

\begin{table}[h!]
\label{tab:forward}
\caption{Results of the forward selection procedure in Subsection \ref{App.Sim.VarSel}. Including the precipitation variables successively improves $\Lambda_n(Y|\XX)$, whereas improvement is no longer possible afterwards. Thus the precipitation variables are selected and the temperature variables discarded.}
\bigskip
\centering
\begin{tabular}{llrl}
            \cmidrule{1-3}
            Step & Chosen variable & $\Lambda_n(Y|\XX)$\\
            \cmidrule{1-3}
            1 & PWeQ\_R & $0.8037$ & \multirow{4}{*}{\hspace{-1em}$\left.\begin{array}{l}
                \\
                \\
                \\
                \end{array}\right\rbrace \text{selected}$} \\
            2 & PDQ\_R & $0.9243$&\\
            3 & PCQ\_R & $0.9338$&\\
            4 & PWaQ\_R & $0.9340$\\
            5 & MTWaQ\_R & $0.9278$ & \multirow{4}{*}{\hspace{-1em}$\left.\begin{array}{l}
                \\
                \\
                \\
                \end{array}\right\rbrace \text{discarded}$} \\
            6 & MTCQ\_R & $0.9204$ & \\
            7 & MTWeQ\_R & $0.9156$ &\\
            8 & MTDQ\_R & $0.9066$\\
            \cmidrule{1-3}
        \end{tabular}
\end{table}

\begin{table}[h]
\label{tab:subset}
\caption{Results of the best subset selection in Subsection \ref{App.Sim.VarSel}. Here the five sets of predictors which produce the highest $\Lambda_n(Y|\XX)$ are shown.}
\bigskip
\centering
\begin{tabular}{llr}
\hline
Rank & Included Variables & $\Lambda_n(Y|\XX)$\\
\hline
$255$ & PWeQ\_R, PDQ\_R, PWaQ\_R, PCQ\_R & $0.9340$\\
$254$ & PWeQ\_R, PDQ\_R, PCQ\_R & $0.9338$\\
$253$ & PWeQ\_R, PDQ\_R, PCQ\_R, MTCQ\_R & $0.9286$\\
$252$ & PWeQ\_R, PDQ\_R, PWaQ\_R, PCQ\_R, MTWaQ\_R & $0.9278$\\
$251$ & PWeQ\_R, PDQ\_R, PWaQ\_R, PCQ\_R, MTWeQ\_R & $0.9278$\\
\hline
\end{tabular}
\end{table}

\begin{figure}[h!]
    \centering
    \includegraphics[width=\linewidth]{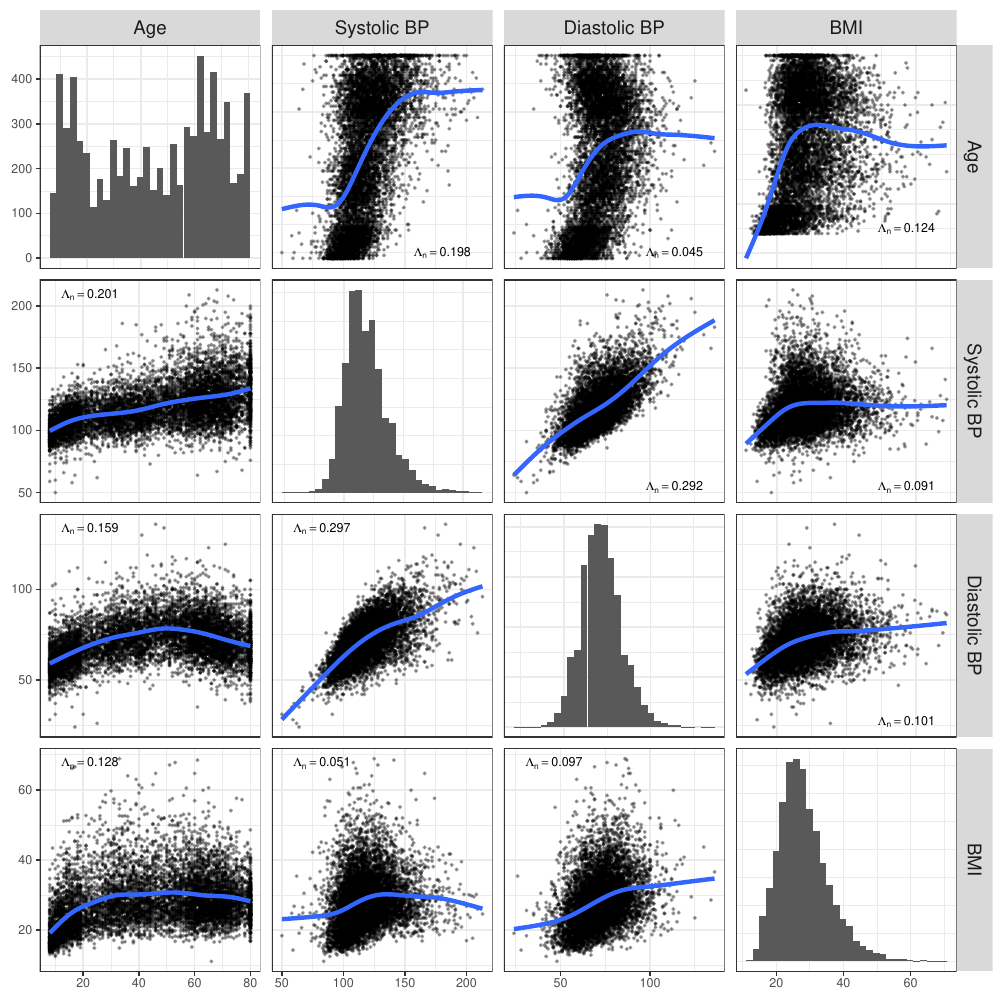}
    \caption{Scatterplots of the data used in Section \ref{Sec.RealDataEx} with loess regression line for every combination of the considered variables. In the diagonals histograms are shown. $\Lambda(Y|X)$ is provided in every scatterplot.}
    \label{fig:pairwise_plots}
\end{figure}

\clearpage
\section{Proofs from Section \ref{Sec:Coeff} and Appendix \ref{App.A}}
\label{Appendx.Proofs:I}

The order of proofs is determined by how they are used in relation to each other. \\

We first repeat some useful facts about the relative effect defined in \eqref{Def.Rel.Eff}.
For the random variable \(Y\) and two realizations \(\xx_1\) and \(\xx_2\) of the random vector \(\XX\), the relative effect of the conditional distributions \(\PP^{Y|\XX=\xx_1} \) and \(\PP^{Y|\XX=\xx_2}\) fulfills
\begin{align}\label{def_REFF}
  \Psi \big(\PP^{Y|\XX=\xx_1}, \PP^{Y|\XX=\xx_2}\big)  
  & = \int_{\R} \PP (Y < y|\XX=\xx_1) + \frac{1}{2} \; \PP (Y = y|\XX=\xx_1) \de \PP^{Y|\XX=\xx_2}(y) 
  \\
  &  = \int_{\R \times \R} \mathds{1}_{\{y_1 < y_2\}} + \frac{1}{2} \; \mathds{1}_{\{y_1 = y_2\}} \de (\PP^{Y|\XX=\xx_1} \otimes \PP^{Y|\XX=\xx_2})(y_1,y_2) \notag
  \\
  &  = \int_{\R \times \R} 1 - \mathds{1}_{\{y_2 < y_1\}} - \frac{1}{2} \; \mathds{1}_{\{y_2 = y_1\}} \de (\PP^{Y|\XX=\xx_2} \otimes \PP^{Y|\XX=\xx_1})(y_2,y_1) \notag
  \\
  &  =  1 - \Psi \big(\PP^{Y|\XX=\xx_2}, \PP^{Y|\XX=\xx_1}\big)\,, \notag
\end{align}
due to Fubini's theorem, and \(\Psi \big(\PP^{Y|\XX=\xx_1}, \PP^{Y|\XX=\xx_2}\big) \in [0,1]\).

\subsection{Proof of Theorem \ref{Thm:REM:Main}}

The following characterization of complete separation paves the way for its quantification.

\begin{lemma}\label{Lemma:Char.PerfSep}~~
Let \(\XX^\ast\) denote an independent copy of \(\XX\).
Then 
\begin{align} \label{Lemma:Char.PerfSep:Eq1}
  & \int_{\R^p \times \R^p} \big( \Psi \big(\PP^{Y|\XX=\xx_1}, \PP^{Y|\XX=\xx_2}\big) - \Psi \big(\PP^{Y|\XX=\xx_2}, \PP^{Y|\XX=\xx_1}\big) \big)^2 
  \de (\PP^\XX \otimes \PP^\XX) (\xx_1,\xx_2) \notag
  \\
  & \leq 1 - \PP(\XX = \XX^\ast)\,, 
\end{align}
with equality if and only if \(Y\) is completely separated relative to \(\XX\).
\end{lemma}
\begin{proof} 
Due to disintegration and Jensen's inequality, we first obtain
\begin{align*}
\hspace{-15mm}
    0 
    &= \left( \int_{\mathbb{R}} \PP(Y < y) \de \PP^{Y}(y) - \int_{\mathbb{R}} \PP(Y < y) \de \PP^{Y}(y) \right)^2 
    \\
    &= \left( \int_{\R^p \times \R^p} \int_{\R} \PP(Y < y|\XX=\xx_1) \de \PP^{Y|\XX=\xx_2}(y) 
            - \int_{\R} \PP(Y < y|\XX=\xx_2) \de \PP^{Y|\XX=\xx_1}(y) 
              \de (\PP^\XX \otimes \PP^\XX)(\xx_1,\xx_2) \right)^2 
    \\
    &\leq \int_{\R^p \times \R^p} \left( \int_{\R} \PP(Y < y|\XX=\xx_1) \de \PP^{Y|\XX=\xx_2}(y) 
        - \int_{\R} \PP(Y < y|\XX=\xx_2) \de \PP^{Y|\XX=\xx_1}(y) \right)^2 
          \de (\PP^\XX \otimes \PP^\XX)(\xx_1,\xx_2)
    \\
    &  = \int_{\R^p \times \R^p} \big( \Psi \big(\PP^{Y|\XX=\xx_1}, \PP^{Y|\XX=\xx_2}\big) - \Psi \big(\PP^{Y|\XX=\xx_2}, \PP^{Y|\XX=\xx_1}\big) \big)^2 
      \de (\PP^\XX \otimes \PP^\XX) (\xx_1,\xx_2)
    \\
    &  = \int_{\R^p \times \R^p} \big( \Psi \big(\PP^{Y|\XX=\xx_1}, \PP^{Y|\XX=\xx_2}\big) - \Psi \big(\PP^{Y|\XX=\xx_2}, \PP^{Y|\XX=\xx_1}\big) \big)^2 \; \mathds{1}_{\{\xx_1 \neq \xx_2\}}
      \de (\PP^\XX \otimes \PP^\XX) (\xx_1,\xx_2)
    \\
    & \leq \int_{\R^p \times \R^p} \mathds{1}_{\{\xx_1 \neq \xx_2\}}
      \de (\PP^\XX \otimes \PP^\XX) (\xx_1,\xx_2) = 1 - \PP(\XX = \XX^\ast)\,,
\end{align*}
where we make use of $\int_{\R} \PP(Y = y|\XX=\xx_1) \de \PP^{Y|\XX=\xx_2}(y) = \int_{\R} \PP(Y = y|\XX=\xx_2) \de \PP^{Y|\XX=\xx_1}(y)$ in the third identity.
This implies \eqref{Lemma:Char.PerfSep:Eq1}. 
\\
We now prove the equivalence. Therefore, recall that \(Y\) is completely separated relative to \(\XX\) if and only if there exists some $G \subseteq \mathcal{B}(\R^p \times \R^p)$ with $(\PP^\XX \otimes \PP^\XX)(G) = 1$ such that identity 
\begin{align} \label{Def.Compl.Sep.III}
  \big( \Psi \big(\PP^{Y|\XX=\xx_1}, \PP^{Y|\XX=\xx_2}\big) - \Psi \big(\PP^{Y|\XX=\xx_2}, \PP^{Y|\XX=\xx_1}\big) \big)^2 
  & = \mathds{1}_{\{\xx_1 \neq \xx_2\}}
\end{align}
holds for all $(\xx_1,\xx_2) \in G$, due to the definition of complete separation in (A\ref{REM:MainProp:3}).
Thus, complete separation of \(Y\) relative to \(\XX\) implies
\begin{align*}
    & \int_{\R^p \times \R^p} \big( \Psi \big(\PP^{Y|\XX=\xx_1}, \PP^{Y|\XX=\xx_2}\big) - \Psi \big(\PP^{Y|\XX=\xx_2}, \PP^{Y|\XX=\xx_1}\big) \big)^2 
    \de (\PP^\XX \otimes \PP^\XX) (\xx_1,\xx_2)
    \\
    & = \int_{G} \mathds{1}_{\{\xx_1 \neq \xx_2\}}
    \de (\PP^\XX \otimes \PP^\XX) (\xx_1,\xx_2)
      = 1 - \PP(\XX = \XX^\ast)\,.
\end{align*}
Now, assume that equality in \eqref{Lemma:Char.PerfSep:Eq1} holds.
Then 
\begin{align*}
    & \int_{\R^p \times \R^p} \left| \mathds{1}_{\{\xx_1 \neq \xx_2\}} - \big( \Psi \big(\PP^{Y|\XX=\xx_1}, \PP^{Y|\XX=\xx_2}\big) - \Psi \big(\PP^{Y|\XX=\xx_2}, \PP^{Y|\XX=\xx_1}\big) \big)^2 \right| 
    \de (\PP^\XX \otimes \PP^\XX) (\xx_1,\xx_2)
    \\
    & = \int_{\R^p \times \R^p} \mathds{1}_{\{\xx_1 \neq \xx_2\}} - \big( \Psi \big(\PP^{Y|\XX=\xx_1}, \PP^{Y|\XX=\xx_2}\big) - \Psi \big(\PP^{Y|\XX=\xx_2}, \PP^{Y|\XX=\xx_1}\big) \big)^2 \, 
    \de (\PP^\XX \otimes \PP^\XX) (\xx_1,\xx_2)
    \\
    & = (1 - \PP(\XX = \XX^\ast)) 
    \\
    & \qquad - \int_{\R^p \times \R^p} \big( \Psi \big(\PP^{Y|\XX=\xx_1}, \PP^{Y|\XX=\xx_2}\big) - \Psi \big(\PP^{Y|\XX=\xx_2}, \PP^{Y|\XX=\xx_1}\big) \big)^2 
        \de (\PP^\XX \otimes \PP^\XX) (\xx_1,\xx_2) 
      = 0\,,
\end{align*}
hence there exists some $G \subseteq \mathcal{B}(\R^p \times \R^p)$ with $(\PP^\XX \otimes \PP^\XX)(G) = 1$ such that \eqref{Def.Compl.Sep.III} holds for all $(\xx_1,\xx_2) \in G$. 
This completes the proof.
\end{proof}

\begin{proof}[Proof of Theorem \ref{Thm:REM:Main}.]~~
The identities are straightforward to verify.
Assertions (A\ref{REM:MainProp:1}) and (A\ref{REM:MainProp:3}) are immediate from Lemma \ref{Lemma:Char.PerfSep}, and the equivalence in (A\ref{REM:MainProp:2}) is evident.
\end{proof}

\subsection{Proof of Theorem \ref{Thm:REM:Kendall}}

We first prove an intermediate step by showing that
\begin{align*}
  & \int_{\R^p \times \R^p}  
    \left(\int_{\R \times \R} \mathds{1}_{\{y_1 < y_2\}} \de (\PP^{Y|\XX=\xx_1} \otimes \PP^{Y|\XX=\xx_2})(y_1,y_2)\right)^2 
    \de (\PP^\XX \otimes \PP^\XX) (\xx_1,\xx_2)
  \\
  &  = \PP (Y_1 < Y_2, Y_1' < Y_2')\,,
\end{align*}
where \((Y_1,Y_1')\) and \((Y_2,Y_2')\) are independent copies of \((Y,Y')\). 
Because of $(Y|\XX) \stackrel{d}{=} (Y'|\XX)$ and  $Y \perp Y' \,|\, \XX$, 
change of coordinates implies
\begin{align*}
  & \int_{\R^p \times \R^p}  
    \left(\int_{\R \times \R} \mathds{1}_{\{y_1 < y_2\}} \de (\PP^{Y|\XX=\xx_1} \otimes \PP^{Y|\XX=\xx_2})(y_1,y_2)\right)^2 
    \de (\PP^\XX \otimes \PP^\XX) (\xx_1,\xx_2)
  \\
  & = \int_{\R^p \times \R^p}  
       \int_{\R \times \R} \mathds{1}_{\{y_1 < y_2\}} \de (\PP^{Y|\XX=\xx_1} \otimes \PP^{Y|\XX=\xx_2})(y_1,y_2)
  \\
  & \qquad \cdot \int_{\R \times \R} \mathds{1}_{\{y'_1 < y'_2\}} \de (\PP^{Y'|\XX=\xx_1} \otimes \PP^{Y'|\XX=\xx_2})(y'_1,y'_2)
  \de (\PP^\XX \otimes \PP^\XX) (\xx_1,\xx_2)
  \\
  &  = \int_{\R^p \times \R^p}  
       \int_{\R^2 \times \R^2} \mathds{1}_{\{y_1 < y_2\}} \mathds{1}_{\{y'_1 < y'_2\}}
       \de (\PP^{(Y,Y')|\XX=\xx_1} \otimes \PP^{(Y,Y')|\XX=\xx_2})((y_1,y_1'),(y_2,y_2'))
  \\
  & \qquad  \de (\PP^\XX \otimes \PP^\XX) (\xx_1,\xx_2)
  \\
  &  = \int_{\R^2 \times \R^2} \mathds{1}_{\{y_1 < y_2\}} \mathds{1}_{\{y'_1 < y'_2\}}
       \de (\PP^{(Y,Y')} \otimes \PP^{(Y,Y')})((y_1,y_1'),(y_2,y_2'))
  \\
  &  = \PP (Y_1 < Y_2, Y_1' < Y_2')\,.
\end{align*}
Analogous representations are obtained for
$\PP (Y_1 > Y_2, Y_1' < Y_2')$, $\PP (Y_1 < Y_2, Y_1' > Y_2')$ and $\PP (Y_1 > Y_2, Y_1' > Y_2')$.
\eqref{def_REFF} and change of coordinates finally yield
\begin{align*}
  & \int_{\R^p \times \R^p} \big( \Psi \big(\PP^{Y|\XX=\xx_1}, \PP^{Y|\XX=\xx_2}\big) - \Psi \big(\PP^{Y|\XX=\xx_2}, \PP^{Y|\XX=\xx_1}\big) \big)^2 \de (\PP^\XX \otimes \PP^\XX) (\xx_1,\xx_2)
  \\
  & = \int_{\R^p \times \R^p} \left( \int_{\R \times \R} \mathds{1}_{\{y_1 < y_2\}} - \mathds{1}_{\{y_1 > y_2\}} \de (\PP^{Y|\XX=\xx_1} \otimes \PP^{Y|\XX=\xx_2})(y_1,y_2) \right)^2 \de (\PP^\XX \otimes \PP^\XX) (\xx_1,\xx_2)
  \\
  & = \PP (Y_1 < Y_2, Y_1' < Y_2') - \PP (Y_1 < Y_2, Y_1' > Y_2') 
      - \PP (Y_1 > Y_2, Y_1' < Y_2') + \PP (Y_1 > Y_2, Y_1' > Y_2')
  \\
  & = \PP \big( (Y_1 - Y_2) (Y_1' - Y_2') > 0 \big) - \PP \big( (Y_1 - Y_2) (Y_1' - Y_2') < 0 \big)\,.
\end{align*}
This completes the proof.

\subsection{Proof of Proposition \ref{Prop.REM:InvB}}

The invariance of $\Lambda(\,\cdot\,|\XX)$ w.r.t. a bijective function ${\bold h}$ of $\XX$ follows from the definition of $\Lambda$ as the \(\sigma\)-algebras generated by \(\XX\) and \({\bf h}(\XX)\) coincide.
The invariance of $\Lambda(Y|\,\cdot\,)$ w.r.t. the strictly increasing function $g$ follows from the representation \eqref{Thm:REM:Kendall.Eq1} of \(\Lambda\) in Theorem \ref{Thm:REM:Kendall}.

\subsection{Proof of Proposition \ref{Prop:REM:Kendall}}

We follow the line of arguments used in \cite[Proof of Proposition 2.7]{ansari2023MFOCI}.

Assume that $\sigma_Y = 1$; otherwise, replace $Y$ by $Y/\sigma_Y$ and use scale invariance of $\Lambda$ according to Proposition \ref{Prop.REM:InvB}.
If $\Sigma_{21}$ is the null matrix, then $\rho = 0$ and \eqref{Prop:REM:Kendall:Eq} yields $\Lambda(Y|\XX) = 0$. This is the correct value since $\Sigma_{21}$ being the null matrix characterizes independence of $Y$ and $\XX$ in the multivariate normal model due to \cite[Proposition 2.8]{ansari2023MFOCI} which implies stochastic comparability of almost all conditional distributions according to Remark \ref{Rem:REM:Main2}.

If $\Sigma_{21}$ is not the null matrix, define the random variable $S:= A{\bf X}$ with $A := \frac{\Sigma_{21} \, \Sigma_{11}^{-1}}{\Sigma_{21} \, \Sigma_{11}^{-1} \, \Sigma_{12}}$. 
According to \cite[Proof of Proposition 2.7]{ansari2023MFOCI}, 
the pair $(S,Y)$ is bivariate normal with 
Pearson correlation \(\rho = \sqrt{\Sigma_{21} \Sigma_{11}^{-1} \Sigma_{12}}\,,\)
and \cite[Example 1]{fuchs2023JMVA} then implies that \((Y,Y')\) is bivariate normal as well, with correlation \(\rho^2\).
Finally, since \(\PP^{Y|\XX=\xx} = \PP^{Y|S=A\xx}\) for \(\PP^\XX\)-almost all $\xx \in \mathbb{R}^p$ due to \cite[Proof of Proposition 2.7]{ansari2023MFOCI}, it hence follows from Remark \ref{Rem:REM:Kendall} in combination with \cite[Example 6.7.2]{durante2016} that 
\(\Lambda(Y|\XX) = \Lambda(Y|S) =  \frac{2}{\pi} \, \arcsin(\rho^2)\).
This proves the result.

\subsection{Proof of Proposition \ref{Prop.REM:IGI_CI}}

We first show an intermediate step and prove a disintegration result for the relative effect in \eqref{Def.Rel.Eff}.
For every \(r\)-dimensional random vector \(\ZZ\) and \(\PP^{\XX} \otimes \PP^{\XX}\)-almost all \(\xx_1,\xx_2\), disintegration yields
\begin{align*}
  \lefteqn{\Psi \big(\PP^{Y|\XX=\xx_1}, \PP^{Y|\XX=\xx_2}\big)}
  \\
  & = \int_{\R \times \R} \mathds{1}_{\{y_1 < y_2\}} + \frac{1}{2} \; \mathds{1}_{\{y_1 = y_2\}} \de (\PP^{Y|\XX=\xx_1} \otimes \PP^{Y|\XX=\xx_2})(y_1,y_2)
  \\
  & = \int_{\R^r \times \R^r} \int_{\R \times \R} \mathds{1}_{\{y_1 < y_2\}} + \frac{1}{2} \; \mathds{1}_{\{y_1 = y_2\}} \de (\PP^{Y|\XX=\xx_1,\ZZ=\zz_1} \otimes \PP^{Y|\XX=\xx_2,\ZZ=\zz_2})(y_1,y_2) 
  \\
  & \qquad \de (\PP^{\ZZ|\XX=\xx_1} \otimes \PP^{\ZZ|\XX=\xx_2})(\zz_1,\zz_2)
  \\
  & = \int_{\R^r \times \R^r} \Psi \big(\PP^{Y|\XX=\xx_1,\ZZ=\zz_1}, \PP^{Y|\XX=\xx_2,\ZZ=\zz_2}\big) \de (\PP^{\ZZ|\XX=\xx_1} \otimes \PP^{\ZZ|\XX=\xx_2})(\zz_1,\zz_2)\,.
\end{align*}
Due to the representation shown in Theorem \ref{Thm:REM:Main} and inequality of Jensen, we then obtain 
\begin{align*}
  \lefteqn{(\alpha(\XX) \cdot \Lambda(Y|\XX) + 1) / 4}
  \\
  &   =  \int_{\R^p \times \R^p} \Psi \big(\PP^{Y|\XX=\xx_1}, \PP^{Y|\XX=\xx_2}\big)^2 \de (\PP^\XX \otimes \PP^\XX) (\xx_1,\xx_2)
  \\
  &   =  \int_{\R^p \times \R^p} \left( \int_{\R^r \times \R^r} \Psi \big(\PP^{Y|\XX=\xx_1,\ZZ=\zz_1}, \PP^{Y|\XX=\xx_2,\ZZ=\zz_2}\big) \de (\PP^{\ZZ|\XX=\xx_1} \otimes \PP^{\ZZ|\XX=\xx_2})(\zz_1,\zz_2) \right)^2 
  \\
  & \qquad \de (\PP^\XX \otimes \PP^\XX) (\xx_1,\xx_2) 
  \\
  & \leq \int_{\R^p \times \R^p} \int_{\R^r \times \R^r} \Psi \big(\PP^{Y|\XX=\xx_1,\ZZ=\zz_1}, \PP^{Y|\XX=\xx_2,\ZZ=\zz_2}\big)^2 \de (\PP^{\ZZ|\XX=\xx_1} \otimes \PP^{\ZZ|\XX=\xx_2})(\zz_1,\zz_2) \\
  & \qquad \de (\PP^\XX \otimes \PP^\XX) (\xx_1,\xx_2) 
  \\
  &   = \int_{\R^{p+r} \times \R^{p+r}} \Psi \big(\PP^{Y|\XX=\xx_1,\ZZ=\zz_1}, \PP^{Y|\XX=\xx_2,\ZZ=\zz_2}\big)^2 \de (\PP^{(\XX,\ZZ)} \otimes \PP^{(\XX,\ZZ)})((\xx_1,\zz_1),(\xx_2,\zz_2))
  \\
  &   = (\alpha(\XX,\ZZ) \cdot \Lambda(Y|(\XX,\ZZ)) + 1)/4\,.
\end{align*}
This proves the first assertion.
Now, assume that $Y$ and $\ZZ$ are conditionally independent given $\XX$.
Then
\begin{align*}
  \lefteqn{(\alpha(\XX,\ZZ) \cdot \Lambda(Y|(\XX,\ZZ)) + 1)/4}
  \\
  &   = \int_{\R^{p+r} \times \R^{p+r}} \Psi \big(\PP^{Y|\XX=\xx_1,\ZZ=\zz_1}, \PP^{Y|\XX=\xx_2,\ZZ=\zz_2}\big)^2 \de (\PP^{(\XX,\ZZ)} \otimes \PP^{(\XX,\ZZ)})((\xx_1,\zz_1),(\xx_2,\zz_2))
  \\
  &   = \int_{\R^{p+r} \times \R^{p+r}} \Psi \big(\PP^{Y|\XX=\xx_1}, \PP^{Y|\XX=\xx_2}\big)^2 \de (\PP^{(\XX,\ZZ)} \otimes \PP^{(\XX,\ZZ)})((\xx_1,\zz_1),(\xx_2,\zz_2)) 
  \\
  &   = \int_{\R^{p} \times \R^{p}} \Psi \big(\PP^{Y|\XX=\xx_1}, \PP^{Y|\XX=\xx_2}\big)^2 \de (\PP^{\XX} \otimes \PP^{\XX})(\xx_1,\xx_2) 
  \\
  &   = (\alpha(\XX) \cdot \Lambda(Y|\XX) + 1)/4
\end{align*}
which proves the second assertion.

\subsection{Proof of Corollary \ref{CorApp.REM:DPI}}

Let ${\bf h}$ be a measurable function for which \(\alpha({\bf h}(\XX)) = \alpha(\XX)\).
Then \(\alpha({\bf h}(\XX)) = \alpha(\XX) = \alpha(\XX,{\bf h}(\XX))\) and $Y$ and ${\bf h}(\XX)$ are conditionally independent given $\XX$. 
Proposition \ref{Prop.REM:IGI_CI} hence yields
\begin{align*}
  \Lambda(Y|{\bf h}(\XX)) 
  & \leq \Lambda(Y|(\XX,{\bf h}(\XX)))
      =  \Lambda(Y|\XX)\,.
\end{align*}
This proves the first assertion, and the second assertion is immediate from the data processing inequality in combination with Proposition \ref{Prop.REM:IGI_CI}.

\subsection{Proof of Proposition \ref{Prop.REM:DI}}

Recall that we denote by \(F_Z\) the cdf of a random variable \(Z\) and let \(F_Z^{-1}\) denote its left-continuous inverse.

We first prove the invariance of the normalizing constant.
Therefore, recall that every random variable \(Z\) fulfills \(F_Z^{-1} \circ F_Z \circ Z = Z\) almost surely which gives 
\begin{align*}
  \PP (F_Z (Z) = F_Z (Z^\ast))
  & \leq \PP (F_Z^{-1} \circ F_Z(Z) = F_Z^{-1} \circ F_Z(Z^\ast))
  \\
  &   =  \int_{\R} \PP (F_Z^{-1} \circ F_Z(Z) = z) \de \PP^{F_Z^{-1} \circ F_Z \circ Z^\ast}(z)
  \\
  &   =  \int_{\R} \PP (Z = z) \de \PP^{Z^\ast}(z)
  \\
  &   =  \PP (Z = Z^\ast)
  \\
  & \leq \PP (F_Z (Z) = F_Z (Z^\ast))\,,
\end{align*}
with \(Z^\ast\) being an independent copy of \(Z\), hence \(\PP (Z = Z^\ast) = \PP (F_Z (Z) = F_Z (Z^\ast))\).
Since the same result applies to random vectors as well, this implies \(\alpha ({\bf F}_\XX (\XX)) = \alpha(\XX)\).
\\
Now, we prove invariance in the predictor variable.
The above almost sure identity in combination with the data processing inequality \eqref{Eq.DPI} given in Corollary \ref{CorApp.REM:DPI} yields
\begin{align*}
  \Lambda \big(Y| {\bf F}_\XX (\XX) \big)
  & \leq \Lambda(Y|\XX)
  \\
  &   =  \Lambda(Y|(F_{X_1}^{-1} \circ F_{X_1} \circ X_1, \dots, F_{X_p}^{-1} \circ F_{X_p} \circ X_p))
  \\
  & \leq \Lambda \big(Y| {\bf F}_\XX (\XX) \big)\,.
\end{align*}

Finally, we prove invariance in the response by applying \eqref{Thm:REM:Kendall.Eq1} in Theorem \ref{Thm:REM:Kendall}.
As shown in the proof of Theorem \ref{Thm:REM:Kendall}, the numerator in \eqref{Thm:REM:Kendall.Eq1} can be decomposed into the four probabilities 
\begin{align}\label{Prop.DI.Proof1}
  &   \PP \big( (Y_1 - Y_2) (Y_1' - Y_2') > 0 \big) - \PP \big( (Y_1 - Y_2) (Y_1' - Y_2') < 0 \big)
  \\
  & = \PP \big( Y_1 < Y_2, Y_1' < Y_2' \big) 
      -\PP \big( Y_1 < Y_2, Y_1' > Y_2' \big)
      - \PP \big( Y_1 > Y_2, Y_1' < Y_2' \big) 
      + \PP \big( Y_1 > Y_2, Y_1' > Y_2' \big) \notag
  \\
  & = \PP \big( Y_1 \leq Y_2, Y_1' \leq Y_2' \big)
      - \PP \big( Y_1 \leq Y_2, Y_1' \geq Y_2' \big)
      - \PP \big( Y_1 \geq Y_2, Y_1' \leq Y_2' \big) 
      + \PP \big( Y_1 \geq Y_2, Y_1' \geq Y_2' \big)\,,\notag
\end{align}
so it is enough to show the identity for each of the above four probabilities.
Using again the fact that every random variable \(Z\) fulfills \(F_Z^{-1} \circ F_Z \circ Z = Z\) almost surely yields
\begin{align*}
  \PP (Y_1 \geq Y_2, Y_1' \geq Y_2')
  & \leq \PP \big( F_Y(Y_1) \geq F_Y(Y_2), Y_1' \geq Y_2' \big)
  \\
  & \leq \PP \big( (F_Y^{-1} \circ F_Y)(Y_1) \geq (F_Y^{-1} \circ F_Y)(Y_2), Y_1' \geq Y_2' \big)
  \\
  &   =  \PP (Y_1 \geq Y_2, Y_1' \geq Y_2')\,,
\end{align*}
where the inequalities follow from the fact that \(F_Y\) and \(F_Y^{-1}\) are nondecreasing.
Applying the same argument to all probabilities in \eqref{Prop.DI.Proof1} eventually gives \(\Lambda(Y|\XX) = \Lambda \big(F_Y(Y)| \XX \big)\).

\subsection{Proof of Theorem \ref{Thm1.REM.PS.vs.PD}}

The proof of Theorem \ref{Thm1.REM.PS.vs.PD} \eqref{Thm1.REM.PS.vs.PD:I4}, in which it is shown that complete separation implies perfect dependence, is a slight modification of \cite[Proof of Theorem 2.2 (iii)]{ansari2023DepM}.

We first prove \eqref{Thm1.REM.PS.vs.PD:I1} by constructing some random vector \((\XX^\circ,Y^\circ)\) with \(\XX^\circ \eqd \XX\) and \(Y^\circ \eqd Y\) such that \(Y^\circ = h(\XX^\circ)\) a.s. for some measurable function \(h\).
Therefore, let $F$ be the cdf of the continuous random variable $X_i$ and let $G$ be that of $Y$.
Then $F(X_i) \sim U[0,1]$. 
Now, define \(\XX^\circ := \XX\) and \(Y^\circ := h(\XX^\circ)\) where \(h(x_1,\dots,x_p) := (G^{-} \circ F) (x_i)\) with $G^{-}$ denoting the left-continuous inverse of $G$.
Then 
$$ 
  \PP(Y^\circ \leq y)
  = \PP((G^{-} \circ F) (X_i) \leq y)
  = \PP(F (X_i) \leq G(y))
  = G(y)
  = \PP(Y \leq y)
$$
from which assertion \eqref{Thm1.REM.PS.vs.PD:I1} immediately follows.

The maximum value for \(\Lambda(Y|X)\) within the Fr{\'e}chet class considered in Example \ref{Ex.REM.PS.vs.PD} (ii) is $1/2$ and is achieved for the joint distribution of $(X,Y)$ used there.
Assertion \eqref{Thm1.REM.PS.vs.PD:I2} hence is a consequence of Example \ref{Ex.REM.PS.vs.PD} (ii).

We now prove \eqref{Thm1.REM.PS.vs.PD:I3}. 
Therefore, assume that \(Y\) is completely separated relative to \(\XX\) from which \(\Lambda(Y|\XX) = 1\) follows by definition, and recall that \(\alpha(\XX)=1\).
Then, by Jensen's inequality 
\begin{align*} 
  \lefteqn{\Lambda(Y|\XX)} 
  \\
  &   =  \int_{\R^p \times \R^p} \big( \Psi \big(\PP^{Y|\XX=\xx_1}, \PP^{Y|\XX=\xx_2}\big) - \Psi \big(\PP^{Y|\XX=\xx_2}, \PP^{Y|\XX=\xx_1}\big) \big)^2 \de (\PP^\XX \otimes \PP^\XX) (\xx_1,\xx_2) 
  \\
  &   =  \int_{\R^p \times \R^p} \left(  \int_{\R \times \R} \mathds{1}_{\{y_1 < y_2\}} - \mathds{1}_{\{y_1 > y_2\}} \de (\PP^{Y|\XX=\xx_1} \otimes \PP^{Y|\XX=\xx_2})(y_1,y_2) \right)^2 \de (\PP^{\XX} \otimes \PP^{\XX}) (\xx_1,\xx_2) 
  \\
  & \leq \int_{\R^p \times \R^p} \int_{\R \times \R} \big( \mathds{1}_{\{y_1 < y_2\}} - \mathds{1}_{\{y_1 > y_2\}} \big)^2 \de (\PP^{Y|\XX=\xx_1} \otimes \PP^{Y|\XX=\xx_2})(y_1,y_2) \de (\PP^{\XX} \otimes \PP^{\XX}) (\xx_1,\xx_2) 
  \\
  &   =  \int_{\R \times \R}  \mathds{1}_{\{y_1 < y_2\}} + \mathds{1}_{\{y_1 > y_2\}} \de (\PP^{Y} \otimes \PP^{Y})(y_1,y_2) 
  \\
  &   =  \PP(Y < Y^\ast) + \PP(Y > Y^\ast) \leq 1 
\end{align*}
where \(Y^\ast\) denotes an independent copy of \(Y\).
For \(\Lambda(Y|\XX) = 1\) the above inequality becomes an equality so that
\(\PP(Y = Y^\ast) 
  = 1 - \PP(Y \neq Y^\ast) 
  = 0\) from which it follows that \(Y\) is continuous.

For proving the second part of \eqref{Thm1.REM.PS.vs.PD:I4} we again assume that \(Y\) is completely separated relative to \(\XX\).
Then there exists some Borel set \(G\subseteq \R^p \times \R^p\) with \((\PP^\XX\otimes \PP^\XX)(G)=1\) such that for \((\xx_1,\xx_2)\in G\) with \(\xx_1 \neq \xx_2\) and \(Z_k\sim \PP^{Y|\XX=\xx_k}\,,\) \(k\in \{1,2\}\,,\) \(Z_1\) and \(Z_2\) independent,
\begin{align}\label{cor_REM_Xcont:P1}
    \text{either} \quad Z_1\leq Z_2 ~\, \text{almost surely} 
    \quad \text{or} \quad Z_1\geq Z_2 ~\, \text{almost surely.}
\end{align}
Letting \(l(\xx)\) and \(u(\xx)\) denote the infimum and the supremum of the support of \(\PP^{Y|\XX=\xx}\,,\) i.e.
\begin{align}
    l(\xx):= \inf(\supp(\PP^{Y|\XX=\xx})) \leq \sup(\supp(\PP^{Y|\XX=\xx}))=:u(\xx)\,,
\end{align}
statement \eqref{cor_REM_Xcont:P1} hence implies that, for each \((\xx_1,\xx_2)\in G\) with \(\xx_1 \neq \xx_2\), the open intervals fulfill \((l(\xx_1),u(\xx_1))\cap (l(\xx_2),u(\xx_2)) = \emptyset\).
\\
For completing the proof, it suffices to demonstrate that the set 
\(V:=\{\xx\in \R^p \colon \PP^{Y|\XX=\xx} \text{ is non-degenerate}\}\) is a null set, i.e.~\(\PP^\XX(V)=0\), since it then follows that \(\PP^{Y|\XX=\xx}\) is degenerate for \(\PP^\XX\)-almost all \(\xx\in \R^p\), thus \(Y\) perfectly depends on \(\XX\).
\\
Let us assume, on the contrary, that \(\PP^\XX(V)>0\). 
Denoting by \(G_\xx:=\{\xx^\circ\in \R^p\colon (\xx,\xx^\circ)\in G\}\) the $\xx$-cut of $G$, defining $W:=\{\xx\in \R^p\colon \PP^\XX(G_\xx)=1\}$ and using disintegration first gives \(\PP^\XX(W)=1\), from which it then immediately follows that 
\begin{align}\label{cor_REM_Xcont:P3}
    \PP^\XX(V\cap W)=\PP^\XX(V)>0\,.
\end{align}
Now, let \(\xx_1\in V\cap W\) be arbitrary but fixed. 
Then, by construction, 
\begin{align}\label{cor_REM_Xcont:P2}
    (l(\xx_1),u(\xx_1))\cap (l(\xx_2),u(\xx_2)) = \emptyset
\end{align}
for all \(\xx_2\in \R^p\) with \((\xx_1,\xx_2)\in G\) such that \(\xx_2 \neq \xx_1\), 
so in particular for all \(\xx_2\in V\cap W \cap G_{\xx_1}\). 
Since \(\PP^\XX(G_{\xx_1})=1=\PP^\XX(W)\) it follows that \(0<\PP^\XX(V) = \PP^\XX(V\cap W \cap G_{\xx_1})\). 
Hence, \eqref{cor_REM_Xcont:P2} holds in particular for \(\PP^\XX\)-almost every \(\xx_2 \in V\cap W\). \\
Finally, let \(q_1,q_2,q_3,\ldots\) be an enumeration of the rational numbers in $\R$.
Then defining \(M_k:=\{\xx\in V\cap W \colon q_k\in (l(\xx),u(\xx))\}\) yields \(\bigcup_{k=1}^\infty M_k = V\cap W\) since the rationals are dense in $\R$. 
Using \eqref{cor_REM_Xcont:P3} as well as sub-\(\sigma\)-additivity yields
\begin{align}
    0 
    & < \PP^\XX(V\cap W)
      = \PP^\XX\Big(\bigcup_{k=1}^\infty M_k\Big) 
      \leq \sum_{k=1}^\infty \PP^\XX(M_k),
\end{align}
so there exists some \(k_0\in \N\) with \(\PP^\XX(M_{k_0})>0\).
The latter in combination with the assumption that $X_i$ is continuous, however, contradicts \eqref{cor_REM_Xcont:P2}, so \(\PP^\XX(V)>0\) can not hold, and 
the proof of \eqref{Thm1.REM.PS.vs.PD:I4} is complete.

Assertion \eqref{Thm1.REM.PS.vs.PD:I5} again is a consequence of Example \ref{Ex.REM.PS.vs.PD} (ii).

\subsection{Proof of Theorem \ref{Thm2.REM.PS.vs.PD}}

We first prove \eqref{Thm2.REM.PS.vs.PD:I1} by constructing some random vector \((\XX^\circ,Y^\circ)\) with \(\XX^\circ \eqd \XX\) and \(Y^\circ \eqd Y\) such that \(Y^\circ\) is completely separated relative to \(\XX^\circ\).
Therefore, let $F$ be the cdf of $X_1$, let $G$ be that of $Y$, and denote by \(\Phi: \R \times [0,1] \mapsto [0,1]\) the generalized probability integral transform given by
\begin{align*}
  \Phi(X_1,R)
  & := F(X_1-) + R \cdot (F(X_1) - F(X_1-))\,,
\end{align*}
where \(R \sim U[0,1]\) is a random variable independent of \(\XX\).
Then \(\Phi(X_1,R) \sim U[0,1]\) and \(G^{-}\) is strictly increasing.
Now, define \(\XX^\circ := \XX\) and \(Y^\circ:= (G^{-} \circ \Phi) (X_1,R)\) so that \(Y^\circ = h(X_1,R)\) with \(h = G^{-} \circ \Phi\).
Then \(h(x_1,r_1) > h(x_2,r_2)\) if and only if either \(x_1=x_2\) and \(r_1>r_2\) or \(x_1>x_2\),
and disintegration yields
\begin{align*}
  \lefteqn{\alpha(\XX^\circ) \, \Lambda(Y^\circ|\XX^\circ)}
  \\*
  & = \int_{\R^p \times \R^p} \big( \Psi \big(\PP^{Y^\circ|\XX^\circ=\xx_1}, \PP^{Y^\circ|\XX^\circ=\xx_2}\big) - \Psi \big(\PP^{Y^\circ|\XX^\circ=\xx_2}, \PP^{Y^\circ|\XX^\circ=\xx_1}\big) \big)^2 \de (\PP^{\XX^\circ} \otimes \PP^{\XX^\circ}) (\xx_1,\xx_2)
  \\
  & = \int_{\R^p \times \R^p} \left(  \int_{\R \times \R} \mathds{1}_{\{y_1 < y_2\}} - \mathds{1}_{\{y_1 > y_2\}} \de (\PP^{Y^\circ|\XX^\circ=\xx_1} \otimes \PP^{Y^\circ|\XX^\circ=\xx_2})(y_1,y_2) \right)^2 \de (\PP^{\XX^\circ} \otimes \PP^{\XX^\circ}) (\xx_1,\xx_2)
  \\
  & = \int_{\R^p \times \R^p} \biggl(  \int_{\R \times \R} \int_{\R \times \R}     \mathds{1}_{\{y_1 < y_2\}} - \mathds{1}_{\{y_1 > y_2\}} \de (\PP^{Y^\circ|\XX^\circ=\xx_1,R=r_1} \otimes \PP^{Y^\circ|\XX^\circ=\xx_2,R=r_2})(y_1,y_2)  
  \\
  & \qquad \de (\PP^{R|\XX^\circ=\xx_1} \otimes \PP^{R|\XX^\circ=\xx_2})(r_1,r_2) \biggr)^2 \de (\PP^{\XX^\circ} \otimes \PP^{\XX^\circ}) (\xx_1,\xx_2)
  \\
  & = \int_{\R^p \times \R^p} \mathds{1}_{\{\xx_1 \neq \xx_2\}} \left(  \int_{\R \times \R} \mathds{1}_{\{h(x_1,r_1) < h(x_2,r_2)\}} - \mathds{1}_{\{h(x_1,r_1) > h(x_2,r_2)\}} \de (\PP^R \otimes \PP^R)(r_1,r_2) \right)^2 
  \\
  & \qquad \de (\PP^{\XX^\circ} \otimes \PP^{\XX^\circ}) (\xx_1,\xx_2)
  \\
  & = \int_{\R^p \times \R^p} \mathds{1}_{\{\xx_1 \neq \xx_2\}} \left(  \int_{\R \times \R} \mathds{1}_{\{x_1 < x_2\}} - \mathds{1}_{\{x_1 > x_2\}} \de (\PP^R \otimes \PP^R)(r_1,r_2) \right)^2 \de (\PP^{\XX^\circ} \otimes \PP^{\XX^\circ}) (\xx_1,\xx_2)
  \\
  & = \int_{\R^p \times \R^p} \mathds{1}_{\{\xx_1 \neq \xx_2\}} \left( \mathds{1}_{\{x_1 < x_2\}} + \mathds{1}_{\{x_1 > x_2\}} \right) \de (\PP^{\XX^\circ} \otimes \PP^{\XX^\circ}) (\xx_1,\xx_2)
  \\
  & = \int_{\R^p \times \R^p} \mathds{1}_{\{\xx_1 \neq \xx_2\}} \de (\PP^{\XX^\circ} \otimes \PP^{\XX^\circ}) (\xx_1,\xx_2) = \alpha(\XX^\circ)\,.
\end{align*}
Thus, \(\Lambda(Y^\circ|\XX^\circ) = 1\) and it follows that \(Y^\circ\) is completely separated relative to \(\XX^\circ\). This proves \eqref{Thm2.REM.PS.vs.PD:I1}.

Assertion \eqref{Thm2.REM.PS.vs.PD:I2} is a consequence of Example \ref{Ex.REM.PS.vs.PD} (i) as within the considered Fr{\'e}chet class there is no function $h$ such that $Y = h(X)$ almost surely.

We now prove \eqref{Thm2.REM.PS.vs.PD:I3}.
Therefore, assume that \(Y\) perfectly depends on \(\XX\), i.e.~\(Y = h(\XX)\) almost surely. 
Then continuity of $Y$ yields
\(0 
  \leq \PP(\XX=\xx)
  \leq \PP(h(\XX)=h(\xx))
    =  \PP(Y=h(\xx))
    =  0\)
for all \(\xx \in \R^p\).
Thus, there exists some $i \in \{1,\dots,p\}$ such that $X_i$ is continuous.

For proving the second part of \eqref{Thm2.REM.PS.vs.PD:I4} we again assume that \(Y\) perfectly depends on \(\XX\) and show that \(\Lambda(Y|\XX) = 1\).
Since there exists some $i \in \{1,\dots,p\}$ such that $X_i$ is continuous we have \(\alpha(\XX)=1\), and change of coordinates in combination with $Y$ being continuous yields
\begin{align*}
  \lefteqn{\Lambda(Y|\XX)}
  \\
  & = \int_{\R^p \times \R^p} \big( \Psi \big(\PP^{Y|\XX=\xx_1}, \PP^{Y|\XX=\xx_2}\big) - \Psi \big(\PP^{Y|\XX=\xx_2}, \PP^{Y|\XX=\xx_1}\big) \big)^2 \de (\PP^\XX \otimes \PP^\XX) (\xx_1,\xx_2) \notag
  \\
  & = \int_{\R^p \times \R^p} \left( \int_{\R \times \R} \mathds{1}_{\{y_1 < y_2\}} - \mathds{1}_{\{y_1 > y_2\}} \de (\PP^{Y|\XX=\xx_1} \otimes \PP^{Y|\XX=\xx_2})(y_1,y_2) \right)^2 \de (\PP^{\XX} \otimes \PP^{\XX}) (\xx_1,\xx_2)
  \\
  & = \int_{\R^p \times \R^p} \left( \mathds{1}_{\{h(\xx_1) < h(\xx_2)\}} - \mathds{1}_{\{h(\xx_1) > h(\xx_2)\}} \right)^2 \de (\PP^{\XX} \otimes \PP^{\XX}) (\xx_1,\xx_2)
  \\
  & = \int_{\R^p \times \R^p} \mathds{1}_{\{h(\xx_1) < h(\xx_2)\}} + \mathds{1}_{\{h(\xx_1) > h(\xx_2)\}} \de (\PP^{\XX} \otimes \PP^{\XX}) (\xx_1,\xx_2)
  \\
  & = \int_{\R \times \R} \mathds{1}_{\{y_1 < y_2\}} + \mathds{1}_{\{y_1 > y_2\}} \de (\PP^{Y} \otimes \PP^{Y}) (y_1,y_2) = 1
\end{align*}
from which it follows that \(Y\) is completely separated relative to \(\XX\).

Assertion \eqref{Thm2.REM.PS.vs.PD:I5} is again a consequence of Example \ref{Ex.REM.PS.vs.PD} (i).

\section{Proofs from Section \ref{Sec:Cont}}
\label{Appendx.Proofs:II}

\subsection{Proof of Proposition \ref{Cont.Prop.NWC}}

Since $(X,Y)$ has independent and continuous marginal cdfs, its corresponding copula is the independence copula $\Pi$ given by $\Pi(u,v)=uv$.
\\
Now, let $\lambda_0 \in [0,1]$ and consider the sequence $(C_n)_{n \in \N}$ of shuffles of Min used in \cite[Theorem 2.1]{buecher2024}, \cite[Theorem 5.2.10]{durante2016} and going back to \citep{mikusinski1992} converging uniformly to $\Pi$ while each $C_n$ is perfectly dependent.
For $n \in \N$, define further the random vector $(X_n,Y_n)$ having the same marginal cdfs than $(X,Y)$ but connecting copula $A_n := \lambda_0 \, \Pi + (1-\lambda_0) C_n$.
Then, the sequence $(X_n,Y_n)_{n \in \N}$ weakly converges to the vector $(X,Y)$.
\\
For evaluating $\Lambda(Y_n|X_n)$ we make use of Remark \ref{Rem:REM:Kendall} stating that $\Lambda(Y_n|X_n) = \tau(\psi(A_n))$ where $\psi$ is the mapping introduced in \citep{fuchs2023JMVA} that transforms every bivariate copula $A$ into its Markov product $A^T \ast A$ with $A^T (u,v) := A(v,u)$.
According to the calculation rules for the Markov product (see, e.g., \citep{durante2016})
\begin{align*}
  & A_n^T \ast A_n
  \\
  & = \lambda_0^2 \, (\Pi^T \ast \Pi)
      + \lambda_0\,(1-\lambda_0) \, (\Pi^T \ast C_n)
      + (1-\lambda_0)\, \lambda_0 \, (C_n^T \ast \Pi)
      + (1-\lambda_0)^2 \, (C_n^T \ast C_n)
  \\
  & = \lambda_0^2 \, \Pi
      + \lambda_0\,(1-\lambda_0) \, \Pi
      + (1-\lambda_0)\, \lambda_0 \, \Pi
      + (1-\lambda_0)^2 \, M
  \\
  & = (1-\beta) \, \Pi + \beta M\,,
\end{align*}
where $\beta = (1-\lambda_0)^2$ and where the second identity follows from the fact that $\Pi$ is a null element of the collection of all copulas equipped with the Markov product and the Markov product $A^T \ast A$ of a perfectly dependent copula $A$ equals the comonotonicity copula $M$ given by $M(u,v) := \min\{u,v\}$.
Due to the symmetry of the copula integral (see, e.g., \citep{fuchs2016}) we eventually obtain
\begin{align*}
  \Lambda(Y_n|X_n) 
  & = \tau \big( (1-\beta) \, \Pi + \beta \,  M \big)
  \\
  & = 4 \, \left( (1-\beta)^2 \int \Pi \de \mu_\Pi + 2 \, \beta \, (1-\beta) \int \Pi \de \mu_M + \beta^2 \int M \de \mu_M \right) - 1
  \\
  & = 4 \, \left( (1-\beta)^2 \frac{1}{4} + 2 \, \beta \, (1-\beta) \frac{1}{3} + \beta^2 \frac{1}{2} \right) - 1
  \\
  & = \frac{(1-\lambda_0)^2}{3} \, \left( 2 + (1-\lambda_0)^2 \right)\,.
\end{align*}
Setting $\lambda := \frac{(1-\lambda_0)^2}{3} \, \left( 2 + (1-\lambda_0)^2 \right) \in [0,1]$ completes the proof.

\bigskip
For a cdf \(F\), 
denote by $F^{-}$ its left-continuous version, i.e.~$F^{-}(x) := \lim_{z \uparrow x} F(z)$ for all $z \in \mathbb{R}$.

\subsection{Proof of Proposition \ref{Cont.Lemma.1C:1}}

The following result addresses the univariate case.

\begin{lemma} \label{Cont.Lemma.1A}
Consider the random variable \(X\) and a sequence of random variables \((X_n)_{n \in \mathbb{N}}\). 
If \(F_{X_n} \circ F_{X_n}^{-1} (t) \to F_{X} \circ F_{X}^{-1}(t)\) for \(\lambda\)-almost all \(t \in (0,1)\),
then 
\begin{align*}
    \PP(X_n=X_n^\ast) \xrightarrow{} \PP(X=X^\ast)\,,
\end{align*}
where \(X^\ast\) and \(X_n^\ast\) denote independent copies of \(X\) and \(X_n\), respectively.
\end{lemma}
\begin{proof}
Since \(F_{X_n} \circ F_{X_n}^{-1} (t) \to F_{X} \circ F_{X}^{-1}(t)\) for \(\lambda\)-almost all \(t \in (0,1)\) by assumption, 
we obtain convergence also for the left-continuous versions, i.e.~$F_{X_n}^- \circ F_{X_n}^{-1} (t) \to F_X^-\circ F_X^{-1}(t)$ for $\lambda$-almost all $t\in (0,1)$; see, \cite[Lemma 2.17.(ii)]{Ansari-2021}.
Dominated convergence then yields
\begin{align*}
  \PP(X_n= X_n^\ast)
  & = \int_{\mathbb{R}} F_{X_n} (x) - F^{-}_{X_n} (x) \de \PP^{X_n}(x)
  \\
  & = \int_{(0,1)} F_{X_n}(F_{X_n}^{-1}(t)) - F^{-}_{X_n} (F_{X_n}^{-1}(t)) \de \lambda(t)
  \\
  & \xrightarrow{} \int_{(0,1)} F_X(F_X^{-1}(t)) - F_X^{-}(F_X^{-1}(t)) \de \lambda(t)
    = \PP(X=X^\ast)\,.
\end{align*}
This proves the assertion.
\end{proof}

Proposition \ref{Cont.Lemma.1C:1} is a straightforward extension of Lemma \ref{Cont.Lemma.1A} to an arbitrary number of predictor variables requiring the presence of at least one continuous marginal distribution function.
Its proof is immediate from
\begin{align*}
    \left| \PP(\XX_n=\XX_n^\ast) - \PP(\XX=\XX^\ast)\right|
    &   =  \PP(\XX_n=\XX_n^\ast)
      \leq \PP(X_{k,n} = X_{k,n}^\ast)
      \xrightarrow{} 0
\end{align*}
where convergence follows from Lemma \ref{Cont.Lemma.1A}.

\subsection{Proof of Proposition \ref{Cont.Lemma.1C}}

We shall need the following lemma, which addresses the univariate case.

\begin{lemma} \label{Cont.Lemma.1B}
Consider the random variable \(X\) and a sequence of random variables \((X_n)_{n \in \mathbb{N}}\) with \(X_n \xrightarrow{d} X\) and \(F_{X_n} \circ F_{X_n}^{-1} (t) \to F_{X} \circ F_{X}^{-1}(t)\) for \(\lambda\)-almost all \(t \in (0,1)\).
Then, for each $x \in \mathbb{R}$ with \(\PP(X=x) > 0\),
there exists a sequence \(x_n \xrightarrow{} x\) such that \(\PP(X_n=x_n) \xrightarrow{} \PP(X=x)\).
\end{lemma}
\begin{proof}
Since \(F_{X_n} \circ F_{X_n}^{-1} (t) \to F_{X} \circ F_{X}^{-1}(t)\) for \(\lambda\)-almost all \(t \in (0,1)\) by assumption, 
we obtain convergence also for the left-continuous versions; see, \cite[Lemma 2.17.(ii)]{Ansari-2021}.
\\
Fix \(x \in \mathbb{R}\) such that $\PP(X=x) > 0$ and let $(t_{\min}, t_{\max})$ be the interval given by $t_{\min} := F_X^{-}(x)$ and $t_{\max} := F_X(x)$. 
For 
$t \in (t_{\min},t_{\max})$ chosen appropriately (i.e.~ such that \(F_{X_n} \circ F_{X_n}^{-1} (t) \to F_{X} \circ F_{X}^{-1}(t)\) and \(F^-_{X_n} \circ F_{X_n}^{-1} (t) \to F^-_{X} \circ F_{X}^{-1}(t)\)) and  $x_n := x_n(t) := F_{X_n}^{-1}(t)$ we then obtain
\begin{align*}
  \mathbb{P}(X_n = x_n) 
  &= F_{X_n}(x_n) - F^-_{X_n}(x_n) 
  \\
  &= F_{X_n} \circ F_{X_n}^{-1}(t) - F^-_{X_n} \circ F_{X_n}^{-1}(t) \notag
  \\
  &\xrightarrow{} F_X \circ F_X^{-1}(t) - F^-_X \circ F_X^{-1}(t) \notag
  \\
  &= t_{\max} - t_{\min} = \mathbb{P}(X=x) > 0\,, \notag
\end{align*}
which implies that the sequence \((x_n)_{n \in \mathbb{N}}\) fulfills both \(\mathbb{P}(X_n = x_n) \xrightarrow{} \mathbb{P}(X=x)\) and \(x_n = F_{X_n}^{-1}(t) \xrightarrow{} F_{X}^{-1}(t) = x\), where the latter is due to \cite[Lemma 21.2]{vanderVaart-1998} and the fact that \(t\) is a continuity point of \(F_{X}^{-1}\).
\end{proof}

\begin{remark}
Notice that, given $x \in \mathbb{R}$ with \(\PP(X=x) > 0\),
the sequence \((x_n(t))_{n \in \mathbb{N}} = (x_n)_{n \in \mathbb{N}}\) used in the proof of Lemma \ref{Cont.Lemma.1B} converges for $\lambda$-almost every choice of $t \in (t_{\min},t_{\max}) = (F_X^{-}(x),F_X(x))$ to $x$ such that \(\PP(X_n=x_n(t)) \xrightarrow{} \PP(X=x)\).
The sequences \((x_n(t_1))_{n \in \mathbb{N}}\) and \((x_n(t_2))_{n \in \mathbb{N}}\) for appropriate choices of $t_1,t_2 \in (F_X^{-}(x),F_X(x))$ with $t_1 \neq t_2$ can only differ for finitely many $n$.
\end{remark}

\begin{proof}[Proof of Theorem \ref{Cont.Lemma.1C}.]
We prove the result for the case \(p=2\) as a generalization of Lemma \ref{Cont.Lemma.1B}. 
The general result then follows by induction and the same reasoning.

Therefore, fix $\xx = (x_1,x_2) \in \mathbb{R}^2$ such that  $\PP(X_1=x_1, X_2=x_2) > 0$. Then Lemma \ref{Cont.Lemma.1B} ensures the existence of sequences \((x_{k,n})_{n \in \mathbb{N}}\), \(k \in \{1,2\}\), such that $x_{k,n} \xrightarrow{} x_k$ and $\PP(X_{k,n} = x_{k,n}) \xrightarrow{} \PP(X_{k}=x_{k})$. 
Thus, 
\(\PP(\{X_{1,n}=x_{1,n}\} \cap  \{X_{2,n}=x_{2,n}\})\) converges to \(\PP(\{X_{1}=x_{1}\} \cap  \{X_{2}=x_{2}\})\) if and only if 
\begin{align}\label{Cont.Lemma.1C.Eq1}
    \PP(\{X_{1,n}=x_{1,n}\} \cup  \{X_{2,n}=x_{2,n}\}) \xrightarrow{} \PP(\{X_{1}=x_{1}\} \cup  \{X_{2}=x_{2}\})\,.
\end{align}
To prove \eqref{Cont.Lemma.1C.Eq1}, fix $\epsilon>0$.
Since, for each \(k \in \{1,2\}\), \(\PP(X_k=x_k) > 0\), there exists some $\delta > 0$ such that, for all $k \in \{1,2\}$, 
\begin{align}\label{Cont.Lemma.1C.Eq3}
  | \PP(X_k \in [x_k-\delta,x_k+\delta]) - \PP(X_k=x_k) |
  & 
  < \epsilon
  \qquad \textrm{and} \qquad 
  \PP(X_k \in \{x_k-\delta,x_k+\delta\}) 
  = 0\,.
\end{align}
Since $x_{k,n} \xrightarrow{} x_k$, we further have
  \(x_{k,n} \in [x_k-\delta,x_k+\delta]\)
for all \(k \in \{1,2\}\) and all \(n \in \mathbb{N}\) sufficiently large,
which then implies 
\begin{align*}
  0 
  & \leq \PP( \{X_{1,n} \in [x_1-\delta,x_1+\delta]\} \cup \{X_{2,n} \in [x_2-\delta,x_2+\delta]\}) - \PP(\{X_{1,n} = x_{1,n}\} \cup \{X_{2,n}=x_{2,n}\} )
  \\
  & \leq \PP( \{X_{1,n} \in [x_1-\delta,x_1+\delta] \setminus \{x_{1,n}\}\} \cup \{X_{2,n} \in [x_2-\delta,x_2+\delta] \setminus \{x_{2,n}\}\} ) 
  \\
  & \leq \PP( X_{1,n} \in [x_1-\delta,x_1+\delta] \setminus \{x_{1,n}\}) + \PP (X_{2,n} \in [x_2-\delta,x_2+\delta] \setminus \{x_{2,n}\}) 
  \\
  & = \big( \PP(X_{1,n} \in [x_1-\delta,x_1+\delta]) - \PP(X_{1,n}=x_{1,n}) \big) + \big( \PP (X_{2,n} \in [x_2-\delta,x_2+\delta])-\PP(X_{2,n}=x_{2,n}) \big)
\end{align*}
for all \(n \in \mathbb{N}\) sufficiently large.
This can now be used to show \eqref{Cont.Lemma.1C.Eq1}. We obtain
\begin{align*}
  0 
  & \leq \big| \PP(\{X_{1,n} =x_{1,n}\}\cup \{X_{2,n} =x_{2,n}\}) - \PP(\{X_{1}=x_{1}\} \cup \{X_{2} = x_{2}\}) \big|
  \\
  & \leq \big| \PP(\{X_{1,n} = x_{1,n}\}\cup \{X_{2,n} =x_{2,n}\}) - \PP( \{X_{1,n} \in [x_1-\delta,x_1+\delta]\} \cup \{X_{2,n}\in [x_2-\delta,x_2+\delta] \}) \big|
  \\
  & \quad + \big| \PP( \{X_{1,n} \in [x_1-\delta,x_1+\delta]\} \cup \{X_{2,n}\in [x_2-\delta,x_2+\delta] \}) 
  \\
  & \qquad - \PP( \{X_{1} \in [x_1-\delta,x_1+\delta]\} \cup \{X_{2}\in [x_2-\delta,x_2+\delta] \}) \big|
  \\
  & \quad + \big| \PP( \{X_{1} \in [x_1-\delta,x_1+\delta]\} \cup \{X_{2}\in [x_2-\delta,x_2+\delta] \}) - \PP(\{X_{1}=x_{1}\} \cup \{X_{2} = x_{2}\}) \big|
  \\
  & \leq \big| \big( \PP(X_{1,n} \in [x_1-\delta,x_1+\delta]) - \PP(X_{1,n}=x_{1,n}) \big) + \big( \PP (X_{2,n} \in [x_2-\delta,x_2+\delta])-\PP(X_{2,n}=x_{2,n}) \big) \big|
  \\
  & \quad + \big| \PP( \{X_{1,n} \in [x_1-\delta,x_1+\delta]\} \cup \{X_{2,n}\in [x_2-\delta,x_2+\delta] \}) 
  \\
  & \qquad - \PP( \{X_{1} \in [x_1-\delta,x_1+\delta]\} \cup \{X_{2}\in [x_2-\delta,x_2+\delta] \}) \big|
  \\
  & \quad + \big| \big( \PP(X_{1} \in [x_1-\delta,x_1+\delta]) - \PP(X_{1}=x_{1}) \big) + \big( \PP (X_{2} \in [x_2-\delta,x_2+\delta])-\PP(X_{2}=x_{2}) \big) \big|
  \\
  & \xrightarrow{} \big| \big( \PP(X_{1} \in [x_1-\delta,x_1+\delta]) - \PP(X_{1}=x_{1}) \big) + \big( \PP (X_{2} \in [x_2-\delta,x_2+\delta])-\PP(X_{2}=x_{2}) \big) \big|
  \\
  & \quad + \big| \PP( \{X_{1} \in [x_1-\delta,x_1+\delta]\} \cup \{X_{2} \in [x_2-\delta,x_2+\delta] \}) 
  \\
  & \qquad - \PP( \{X_{1} \in [x_1-\delta,x_1+\delta]\} \cup \{X_{2}\in [x_2-\delta,x_2+\delta] \}) \big|
  \\
  & \quad + \big| \big( \PP(X_{1} \in [x_1-\delta,x_1+\delta]) - \PP(X_{1}=x_{1}) \big) + \big( \PP (X_{2} \in [x_2-\delta,x_2+\delta])-\PP(X_{2}=x_{2}) \big) \big|
  \\
  & < 4 \, \epsilon \,, 
\end{align*}
where convergence follows from $\PP(X_{k,n} = x_{k,n}) \xrightarrow{} \PP(X_{k}=x_{k})$, \(k \in \{1,2\}\), \eqref{Cont.Lemma.1C.Eq3} and \(\XX_{n} \xrightarrow{d} \XX\) in combination with Portemanteau theorem \cite[Theorem 2.1]{Billingsley-1999}, and the last inequality is due to \eqref{Cont.Lemma.1C.Eq3}.
This proves \eqref{Cont.Lemma.1C.Eq1}, i.e.~for each $\xx \in \mathbb{R}^2$ with $\PP(\XX=\xx) > 0$ there exists a sequence \(\xx_n \xrightarrow{} \xx\) such that
$\PP(\XX_{n}=\xx_n) \xrightarrow{} \PP(\XX =\xx) > 0$ implying $\PP(\XX_{n}=\xx_n)>0$ for all \(n \in \mathbb{N}\) sufficiently large.
Assigning to each such \(\xx\) the just used sequence \((\xx_n(\xx))_{n \in \mathbb{N}} = (\xx_n)_{n \in \mathbb{N}}\), we then obtain
\begin{align*}
  \PP(\XX = \XX^\ast) 
  &   =  \sum_{\{\xx \,:\, \PP(\XX = \xx) > 0\}} \PP(\XX = \xx)^2 
  \\
  &   =  \sum_{\{\xx \,:\, \PP(\XX = \xx) > 0\}} \lim_{n \to \infty} \PP(\XX_{n}=\xx_n(\xx))^2
  \\
  & \leq \liminf_{n \to \infty} \sum_{\{\xx \,:\, \PP(\XX = \xx) > 0\}} \PP(\XX_{n}=\xx_n(\xx))^2
  \\
  & \leq \liminf_{n \rightarrow \infty} \sum_{\{\zz \,:\, \PP(\XX_n = \zz) > 0\}} \PP(\XX_{n} = \zz)^2 
  \\
  &   =  \liminf_{n \rightarrow \infty} \PP(\XX_{n} = \XX_{n}^\ast)
  \\
  & \leq \limsup_{n \rightarrow \infty} \PP(\XX_{n} = \XX_{n}^\ast)\,,
\end{align*}
where the first inequality follows from Fatou's lemma.
We now prove the reverse inequality, i.e.~$\limsup_{n \rightarrow \infty} \PP(\XX_{n} = \XX_{n}^\ast) \leq \PP(\XX = \XX^\ast) $.
Therefore, define 
\begin{align*}
  A := \bigcup_{n=1}^{\infty}\{ \zz \in \mathbb{R}^2 \,:\,  \PP(\XX_n = \zz) > 0\}\,. 
\end{align*}
Then $A$ is countable. 
For each $\zz \in A$ with \(\PP(\XX=\zz) = 0\) Portmanteau theorem yields $\limsup_{n \to \infty} \PP(\XX_{n}=\zz)^2 = \PP(\XX=\zz)^2 = 0$, and for each $z \in A$ with \(\PP(\XX=\zz) > 0\) Portmanteau theorem further gives $\limsup_{n \to \infty} \PP(\XX_{n}=\zz)^2 \leq \PP(\XX=\zz)^2$.
Therefore, 
$\sum_{\zz \in A } \limsup_{n \to \infty} \PP(\XX_{n}=\zz)^2
  \leq \sum_{\zz \in A } \PP(\XX=\zz)^2
  \leq 1$ and hence
\begin{align*}
  \limsup_{n \rightarrow \infty} \PP(\XX_{n} = \XX^\ast_{n})
  &   =  \limsup_{n \to \infty} \sum_{ \{ \zz \,:\,  \PP(\XX_n = \zz) > 0\}} \PP(\XX_{n}=\zz)^2
  \\
  & \leq \limsup_{n \to \infty} \sum_{\zz \in A } \PP(\XX_{n}=\zz)^2
  \\
  & \leq \sum_{\zz \in A } \limsup_{n \to \infty} \PP(\XX_{n}=\zz)^2
\end{align*}
where the last inequality follows from Fatou's lemma.
Finally,
\begin{align*}
  \limsup_{n \rightarrow \infty} \PP(\XX_{n} = \XX_{n}^\ast)
  & \leq \sum_{\{\xx \,:\, \PP(\XX = \xx) > 0\}} \PP(\XX = \xx)^2
       = \PP(\XX = \XX^\ast)\,,
\end{align*}
and hence 
\begin{align*}
  \PP(\XX = \XX^\ast)
  & \leq \liminf_{n\rightarrow \infty} \PP(\XX_n = \XX_n^\ast) 
    \leq \limsup_{n\rightarrow \infty} \PP(\XX_n = \XX_n^\ast) 
    \leq \PP(\XX = \XX^\ast)\,.
\end{align*}
This proves the result.
\end{proof}

\subsection{Proof of Theorem \ref{Cont.Thm.1}}

It remains to prove convergence in the numerator for which we use expression \eqref{Thm:REM:Kendall.Eq1} of $\Lambda(Y|\XX)$ derived in Theorem \ref{Thm:REM:Kendall}.
That \eqref{Cont.Thm.1.A1} and \eqref{Cont.Thm.1.A2} imply \eqref{Cont.Thm.1.A12} follows from \cite[Theorem 3.1]{ansari2025Cont}. 
Thus, the sequence of Markov products \((Y_n,Y_n')_{n \in \mathbb{N}}\) of \((\XX_n,Y_n)_{n \in \mathbb{N}}\) weakly converges to the Markov product \((Y,Y')\) of \((\XX,Y)\) (cf. \eqref{DefMarkovProduct}).
Then, with the same reasoning as used in the proof of \citep[Theorem 2.2]{ansari2025Cont}, the (on \(\Ran(F_{Y_n})\times \Ran(F_{Y_n})\) and \(\Ran(F_Y)\times \Ran(F_Y)\) uniquely determined) copulas \(C_n\) and \(C\) 
associated with \((Y_n,Y_n')\) and \((Y,Y')\), respectively, fulfill
\begin{align}\label{equnifconvuv}
    C_n \to C  ~~~\text{uniformly on } \overline{\Ran(F_Y)}\times \overline{\Ran(F_Y)}\,.
\end{align}
We now show that the numerator 
\begin{align*}
  \PP \big( (Y_{n,1} - Y_{n,2}) (Y_{n,1}' - Y_{n,2}') > 0 \big) - \PP \big( (Y_{n,1} - Y_{n,2}) (Y_{n,1}' - Y_{n,2}') < 0 \big)
\end{align*}
in \eqref{Thm:REM:Kendall.Eq1}, where \((Y_{n,1},Y_{n,1}')\) and \((Y_{n,2},Y_{n,2}')\) denote independent copies of \((Y_n,Y_n')\), converges to the respective expression in the limit.
Due to the fact that \((Y_{n,1},Y_{n,1}')\) and \((Y_{n,2},Y_{n,2}')\) share the same distribution and are independent, we have 
\begin{align*}
  \lefteqn{\PP \big( (Y_{n,1} - Y_{n,2}) (Y_{n,1}' - Y_{n,2}') > 0 \big)}
  \\
  & = \PP (Y_{n,1} > Y_{n,2}, Y_{n,1}' > Y_{n,2}') + \PP (Y_{n,1} < Y_{n,2}, Y_{n,1}' < Y_{n,2}')
  \\
  & = 2 \, \PP (Y_{n,1} > Y_{n,2}, Y_{n,1}' > Y_{n,2}')
  \\
  & = 2 \, \left( 1 - \PP (Y_{n,1} \leq Y_{n,2}) - \PP (Y_{n,1}' \leq Y_{n,2}') + \PP (Y_{n,1} \leq Y_{n,2}, Y_{n,1}' \leq Y_{n,2}') \right)
  \\
  & = 2 \, \left( 1 - \frac{1}{2} \left( 1 + \PP (Y_{n,1} = Y_{n,2}) \right) - \frac{1}{2} \left( 1 + \PP (Y_{n,1}' = Y_{n,2}') \right) + \PP (Y_{n,1} \leq Y_{n,2}, Y_{n,1}' \leq Y_{n,2}') \right)
  \\
  & = 2 \, \PP (Y_{n,1} \leq Y_{n,2}, Y_{n,1}' \leq Y_{n,2}') 
      - 2 \, \PP (Y_{n,1} = Y_{n,2})
\end{align*}
and 
\begin{align*}
  \lefteqn{\PP \big( (Y_{n,1} - Y_{n,2}) (Y_{n,1}' - Y_{n,2}') < 0 \big)}
  \\
  & = \PP (Y_{n,1} > Y_{n,2}, Y_{n,1}' < Y_{n,2}') + \PP (Y_{n,1} < Y_{n,2}, Y_{n,1}' > Y_{n,2}')
  \\
  & = 2 \, \PP (Y_{n,1} > Y_{n,2}, Y_{n,1}' < Y_{n,2}')
  \\
  & = 2 \, \left( \PP (Y_{n,1}' < Y_{n,2}') - \PP (Y_{n,1} \leq Y_{n,2}, Y_{n,1}' < Y_{n,2}') \right)
  \\
  & = 2 \, \left( \frac{1}{2} \left( 1 - \PP (Y_{n,1} = Y_{n,2}) \right)  - \PP (Y_{n,1} \leq Y_{n,2}, Y_{n,1}' < Y_{n,2}') \right)
  \\
  & = 1 - \PP (Y_{n,1} = Y_{n,2}) - 2 \, \PP (Y_{n,1} \leq Y_{n,2}, Y_{n,1}' < Y_{n,2}') 
  \\
  & = 1 - \PP (Y_{n,1}' = Y_{n,2}') - 2 \, \PP (Y_{n,1} \leq Y_{n,2}, Y_{n,1}' \leq Y_{n,2}') + 2 \, \PP (Y_{n,1} \leq Y_{n,2}, Y_{n,1}' = Y_{n,2}')
  \\
  & = 1 - 2 \, \PP (Y_{n,1} \leq Y_{n,2}, Y_{n,1}' \leq Y_{n,2}') + \PP (Y_{n,1} = Y_{n,2}, Y_{n,1}' = Y_{n,2}')\,,
\end{align*}
and hence 
\begin{align*}
  & \alpha_n \, \Lambda(Y_n|\XX_n)
  \\
  & = \PP \big( (Y_{n,1} - Y_{n,2}) (Y_{n,1}' - Y_{n,2}') > 0 \big) - \PP \big( (Y_{n,1} - Y_{n,2}) (Y_{n,1}' - Y_{n,2}') < 0 \big)
  \\
  & = 4 \, \PP (Y_{n,1} \leq Y_{n,2}, Y_{n,1}' \leq Y_{n,2}') 
      - 1 - 2 \, \PP (Y_{n,1} = Y_{n,2}) - \PP (Y_{n,1} = Y_{n,2}, Y_{n,1}' = Y_{n,2}')\,.
\end{align*}
Since both the terms \(\PP (Y_{n,1} = Y_{n,2})\) and \(\PP (Y_{n,1} = Y_{n,2}, Y_{n,1}' = Y_{n,2}')\) converge to their respective limits according to Proposition \ref{Cont.Lemma.1A} and Proposition \ref{Cont.Lemma.1C} (applied to $(Y,Y')$), 
it thus remains to show \(\PP (Y_{n,1} \leq Y_{n,2}, Y_{n,1}' \leq Y_{n,2}') \to \PP (Y_{1} \leq Y_{2}, Y_{1}' \leq Y_{2}')\), 
where \((Y_{1},Y_{1}')\) and \((Y_{2},Y_{2}')\) denote independent copies of \((Y,Y')\).
Applying Lipschitz continuity of \(C_n\) in combination with condition \eqref{Cont.Thm.1.A3}, and uniform convergence in \eqref{equnifconvuv} yields
\begin{align*}
  \lefteqn{\left| \PP (Y_{n,1} \leq Y_{n,2}, Y_{n,1}' \leq Y_{n,2}') - \PP (Y_{1} \leq Y_{2}, Y_{1}' \leq Y_{2}') \right|}
  \\
  &   =  \biggl| \int_{(0,1)^2} C_n (F_{Y_n} \circ F_{Y_n}^{-1}(u), F_{Y_n'} \circ F_{Y_n'}^{-1}(u'))  \de \mu_{C_n}(u,u') 
  \\
  & \qquad 
         - \int_{(0,1)^2} C (F_{Y} \circ F_{Y}^{-1}(u), F_{Y'} \circ F_{Y'}^{-1}(u'))  \de \mu_{C}(u,u') \biggr|
  \\
  & \leq  \underbrace{\int_{(0,1)^2} \underbrace{\left| C_n (F_{Y_n} \circ F_{Y_n}^{-1}(u), F_{Y_n'} \circ F_{Y_n'}^{-1}(u')) - C_n (F_{Y} \circ F_{Y}^{-1}(u), F_{Y'} \circ F_{Y'}^{-1}(u')) \right|}_{ \leq
  | F_{Y_n} \circ F_{Y_n}^{-1}(u) - F_{Y} \circ F_{Y}^{-1}(u) |
  + | F_{Y_n} \circ F_{Y_n}^{-1}(u') - F_{Y} \circ F_{Y}^{-1}(u') |} \de \mu_{C_n}(u,u')}_{\xrightarrow{} 0\,, \quad \eqref{Cont.Thm.1.A3}} 
  \\
  & \quad + \int_{(0,1)^2} \underbrace{\left| C_n (F_{Y} \circ F_{Y}^{-1}(u), F_{Y'} \circ F_{Y'}^{-1}(u')) - C (F_{Y} \circ F_{Y}^{-1}(u), F_{Y'} \circ F_{Y'}^{-1}(u')) \right|}_{\xrightarrow{} 0\,, \quad \eqref{equnifconvuv}} \de \mu_{C_n}(u,u') 
  \\
  & \quad + \biggl| \int_{(0,1)^2} C (F_{Y} \circ F_{Y}^{-1}(u), F_{Y'} \circ F_{Y'}^{-1}(u')) \de \mu_{C_n}(u,u') 
  \\
  & \qquad
         - \int_{(0,1)^2} C (F_{Y} \circ F_{Y}^{-1}(u), F_{Y'} \circ F_{Y'}^{-1}(u'))  \de \mu_{C}(u,u') \biggr|
\end{align*}
In order to prove that also the last expression converges to \(0\), define 
\begin{align*}
  N := \{(u,u') \in (0,1)^2 \,\mid\, F_{Y} \circ F_{Y}^{-1}(u) \text{ or } F_{Y'} \circ F_{Y'}^{-1}(u') \text{ is discontinuous} \}\,.  
\end{align*}
Since \( F_{Y} \circ F_{Y}^{-1} \) is monotonically increasing, it is continuous except for a countable number of points implying \(\mu_{C_n}(N) = 0 = \mu_{C}(N)\).
Therefore,
\begin{align*}
  \lefteqn{\int_{(0,1)^2} C (F_{Y} \circ F_{Y}^{-1}(u), F_{Y'} \circ F_{Y'}^{-1}(u)) \, \mathrm{d}\mu_{C_n}(u,u')} 
  \\
  & = \int_{(0,1)^2 \setminus N} C (F_{Y} \circ F_{Y}^{-1}(u), F_{Y'} \circ F_{Y'}^{-1}(u')) \, \mathrm{d}\mu_{C_n}(u,u') 
  \\
  & \xrightarrow{} \int_{(0,1)^2 \setminus N} C (F_{Y} \circ F_{Y}^{-1}(u), F_{Y'} \circ F_{Y'}^{-1}(u')) \, \mathrm{d}\mu_{C}(u,u') 
  \\
   &= \int_{(0,1)^2} C (F_{Y} \circ F_{Y}^{-1}(u), F_{Y'} \circ F_{Y'}^{-1}(u')) \, \mathrm{d}\mu_{C}(u,u')\,,
\end{align*}
where convergence follows from Portmanteau theorem \cite[Theorem 2.1]{Billingsley-1999} in combination with \eqref{equnifconvuv} and the fact that $(u,u') \mapsto C (F_{Y} \circ F_{Y}^{-1}(u), F_{Y'} \circ F_{Y'}^{-1}(u'))$ is continuous and bounded on $(0,1)^2 \setminus N$.
This proves convergence of the numerator.
The assertion then follows from condition  \eqref{Cont.Thm.1.A4} which gives \(\alpha_n \xrightarrow{} \alpha\).

\subsection{Proof of Corollary \ref{Cont.Cor.RobustY}}

The result is immediate from \cite[Theorem 3.4]{ansari2025Cont} and Theorem \ref{Cont.Thm.1}.

\section{Proofs from Section \ref{Sec:Estimation}} 
\label{App.Proof.Estimation}

\subsection{Proof of Theorem \ref{Estimator.Thm.Consistency}}

To prove Theorem \ref{Estimator.Thm.Consistency}, we spilt the estimator in \eqref{Estimator.Def.Lambda} into pieces and define
\begin{align} \label{Estimator.Def.Gamma}
  \gamma_n(Y|\XX) 
  & := {n \choose 2}^{-1} \sum_{k < l} \, \sgn(Y_k - Y_l) \, \sgn(Y_{N(k)} - Y_{N(l)})
\end{align}
and 
\begin{align} \label{Estimator.Def.Alpha}
  \alpha_n(\XX) 
  & := 1 - {n \choose 2}^{-1} \sum_{k < l} \, \mathds{1}_{\{\XX_k = \XX_l\}}\,.
\end{align}
If there exist $k,l \in \{1,\dots,n\}$ with $\XX_k \neq \XX_l$, then $\alpha_n(\XX)$ is strictly positive such that 
$$\Lambda_n(Y|\XX) = \frac{\gamma_n(Y|\XX)}{\alpha_n(\XX)}$$
is well-defined (see Remark \ref{Estimator.Remark.General} \eqref{Estimator.Remark.General:2}).
The claimed consistency in Theorem \ref{Estimator.Thm.Consistency} now follows from Lemma \ref{Estimator.Thm.Consistency.Split1} and Theorem \ref{Estimator.Thm.Consistency.Split2} below.

\begin{lemma} \label{Estimator.Thm.Consistency.Split1}
$\lim_{n \to \infty} \alpha_n(\XX) = \alpha(\XX)$ almost surely.
\end{lemma}
\begin{proof}~~
Note that ${n \choose 2}^{-1} \sum_{k < l} \, \mathds{1}_{\{\XX_k = \XX_l\}}$ is a U-statistic with kernel $h(\XX_k, \XX_l) = \mathds{1}_{\{\XX_k = \XX_l\}}$, that is,
\begin{equation*}
    {n \choose 2}^{-1} \sum_{k < l} \, \mathds{1}_{\{\XX_k = \XX_l\}} = {n \choose 2}^{-1} \sum_{k < l} \, h(\XX_k, \XX_l)\,.
\end{equation*}
Clearly $\E(h(\XX_k, \XX_l)) = \PP(\XX = \XX^{\ast})$. 
Since also $\E(|h(\XX_k, \XX_l)|) = \PP(\XX = \XX^{\ast}) < \infty$, according to Theorem 2.3 in \citep{christofides1992}, we have
$\lim_{n \to \infty} \alpha_n(\XX) = \alpha(\XX)$ almost surely.
This proves the result.
\end{proof}

Notice that $\alpha_n(\XX)$ is an unbiased estimator for $\alpha(\XX)$.

\begin{theorem} \label{Estimator.Thm.Consistency.Split2}
$\lim_{n \to \infty} \gamma_n(Y|\XX) = \alpha(\XX) \cdot \Lambda(Y|\XX)$ almost surely.
\end{theorem}

For the prove of Theorem \ref{Estimator.Thm.Consistency.Split2} several intermediate results are required most of which are generalizations of results presented in \cite[Section 11]{chatterjee2021AoS}.
Therefore, let $\XX_1, \XX_2, \dots$ be an infinite sequence of i.i.d. copies of $\XX$. 
For each $n \geq 2$ and each $k \in \{1,\dots,n\}$, let $\XX_{n,k}$ be the Euclidean nearest neighbor of $\XX_k$ among $\{\XX_l \, : \, l \in \{1,\dots,n\}, l \neq k\}$.
Ties are broken at random.

\begin{lemma}[Paired version of Lemma 11.3 in \citep{chatterjee2021AoS}]~~
\label{lem:11.3_m}
It holds that \(\PP(\lim_{n\to \infty} (\XX_{n,1}, \XX_{n,2}) = (\XX_1, \XX_2)) = 1\).
\end{lemma}

\begin{proof}
According to \cite[Lemma 11.3]{chatterjee2021AoS}, we have $\XX_{n,1} \rightarrow \XX_{1}$ almost surely and $\XX_{n,2} \rightarrow \XX_{2}$ almost surely. 
Thus,
\begin{align*}
  & \PP \left( \left\{ \lim_{n\to \infty} \left\| \binom{\XX_1}{\XX_2} - \binom{\XX_{n,1}}{\XX_{n,2}} \right\|_2^2 = 0 \right\} \right)
  \\
  & \geq \PP \left( \left\{\lim_{n\to \infty} \left\| \XX_1 - \XX_{n,1} \right\|_2^2 = 0 \right\} \cap 
                    \left\{\lim_{n\to \infty} \left\| \XX_2 - \XX_{n,2} \right\|_2^2 = 0 \right\} \right)
       = 1
\end{align*}
This completes the proof.
\end{proof}

Now, for $n \geq 2$ and a given nearest neighbor graph, define 
\begin{align*}
  M 
  & := \{(k,l) \in \{1,...,n\}^2 \mid k < l, N(k) = N(l)\}
\end{align*}
as the set of pairs $(k,l)$ with $k < l$ for which $N(k) = N(l)$.
Denote by $C(p)$ the deterministic constant (only depending on $p$), which is used in \cite[Lemma 11.4]{chatterjee2021AoS} as an upper bound for the maximum possible number of observations $\XX_l$ with $\XX_l \neq \XX_k$ for which a vector $\XX_k$ can act as nearest neighbor.

\begin{lemma} \label{lem:App.1}
We have
\begin{align*}
    |M| 
    & \leq n \; \frac{C(p)(C(p)-1)}{2}\,.
\end{align*}
\end{lemma}
\begin{proof}
Fix some index $m$ and let $S_m = \{k \in \{1,...,n\} | N(k) = m\}$ be the set of indices for which $m$ is the nearest neighbor. 
The number of pairs $(k,l)$ with $k < l$ for which $N(k) = N(l) = m$ is  $|S_m| (|S_m| - 1)/2$, which, due to \cite[Lemma 11.4]{chatterjee2021AoS}, is bounded by $C(p) (C(p) + 1) / 2$. Thus,
\begin{align*}
  |M| 
  &   =  \sum_{m = 1}^{n} \frac{|S_m|(|S_m| - 1)}{2} 
    \leq n \, \frac{C(p)(C(p) - 1)}{2}\,.
\end{align*}
\end{proof}

\begin{lemma}[Paired version of Lemma 11.5 in \citep{chatterjee2021AoS}]~~
\label{lem:11.5_m}
The inequality
\begin{align*}
  & \mathbb{E} \left( f(\XX_{n,1}, \XX_{n,2}) \right) 
  \\
  & \leq \left(1 + C(p) (C(p)+2) \; \left( 1 + \frac{C(p) (C(p)-1)}{(n-1) - C(p) (C(p)-1)} \right) \right) \E\left( f(\XX_{1}, \XX_{2}) \right)
\end{align*}
holds for every measurable function $f: \mathbb{R}^p \times \mathbb{R}^p \rightarrow [0,\infty)$ and all $n \geq 2$ for which $N(1) \neq N(2)$, that is, the nearest neighbors of $\XX_1$ and $\XX_2$ differ. 
\end{lemma}

\begin{proof}
Since $f$ is nonnegative, for $k,l \in \{1,\dots,n\}$ with $k \neq l$ and $N(k) \neq N(l)$ we first obtain  
\begin{align*}
  & \E(f(\XX_{n,k}, \XX_{n,l})) 
  \\
  & \leq \E(f(\XX_k, \XX_l)) 
         + \E(f(\XX_{n,k}, \XX_{n,l}) \, \mathds{1}_{\{\XX_{n,k} \neq \XX_{k}\} \cap \{\XX_{n,l} \neq \XX_{l}\}})
  \\
  & \qquad 
         + \E(f(\XX_{n,k}, \XX_{l}) \, \mathds{1}_{\{\XX_{n,k} \neq \XX_{k}\}})
         + \E(f(\XX_{k}, \XX_{n,l}) \, \mathds{1}_{\{\XX_{n,l} \neq \XX_{l}\}})
  \\
  & \leq \E(f(\XX_k, \XX_l)) 
         + \sum_{i=1}^n \sum_{j \neq i} \E\left( f(\XX_{i}, \XX_{j}) \,
          \mathds{1}_{\{\XX_{i} = \XX_{n,k}, \XX_{i} \neq \XX_{k}\}} \mathds{1}_{\{\XX_{j} = \XX_{n,l}, \XX_{j} \neq \XX_{l}\}} \right) 
  \\
  & \qquad   
         + \sum_{i \neq l} \E\left( f(\XX_{i}, \XX_{l}) \,
         \mathds{1}_{\{\XX_{i} = \XX_{n,k}, \XX_{i} \neq \XX_{k}\}} \right) 
         + \sum_{j \neq k} \E\left( f(\XX_{k}, \XX_{j}) \,
           \mathds{1}_{\{\XX_{j} = \XX_{n,l}, \XX_{j} \neq \XX_{l}\}} \right) \,.
\end{align*}
According to \cite[Lemma 11.4]{chatterjee2021AoS}, the number $K_{n,m_1}$ of $m_2$ with $\XX_{m_1}$ being a nearest neighbor of $\XX_{m_2}$ (not necessarily the randomly chosen one) and $\XX_{m_1} \neq \XX_{m_2}$ is bounded from above by $C(p)$, i.e.~$K_{n,m_1} \leq C(p)$.
Further, denote by $L$ the set of pairs $(k,l)$, $k \neq l$, such that $N(k) \neq N(l)$.
Thus,
\begin{align*}
  & \E(f(\XX_{n,1}, \XX_{n,2})) 
    = \frac{1}{|L|} \sum_{k,l \in L} \E(f(\XX_{n,k}, \XX_{n,l}))
  \\
  & \leq \frac{1}{|L|} \sum_{k,l \in L} \E(f(\XX_k, \XX_l)) 
         + \frac{1}{|L|} \sum_{k,l \in L} \sum_{i=1}^n \sum_{j \neq i} \E\left( f(\XX_{i}, \XX_{j}) \,
          \mathds{1}_{\{\XX_{i} = \XX_{n,k}, \XX_{i} \neq \XX_{k}\}} \mathds{1}_{\{\XX_{j} = \XX_{n,l}, \XX_{j} \neq \XX_{l}\}} \right)
  \\
  & \qquad 
         + \frac{1}{|L|} \sum_{k,l \in L} \sum_{i \neq l} \E\left( f(\XX_{i}, \XX_{l}) \,
         \mathds{1}_{\{\XX_{i} = \XX_{n,k}, \XX_{i} \neq \XX_{k}\}} \right)
         + \frac{1}{|L|} \sum_{k,l \in L} \sum_{j \neq k} \E\left( f(\XX_{k}, \XX_{j}) \,
           \mathds{1}_{\{\XX_{j} = \XX_{n,l}, \XX_{j} \neq \XX_{l}\}} \right)
  \\
  & \leq \E(f(\XX_1, \XX_2)) 
         + \frac{1}{|L|} \sum_{i=1}^n \sum_{j \neq i} \E\left( f(\XX_{i}, \XX_{j}) \,
          \sum_{k,l \in L} \mathds{1}_{\{\XX_{i} = \XX_{n,k}, \XX_{i} \neq \XX_{k}\}} \mathds{1}_{\{\XX_{j} = \XX_{n,l}, \XX_{j} \neq \XX_{l}\}} \right)
  \\
  & \qquad 
         + \frac{1}{|L|} \sum_{l=1}^n \sum_{i \neq l} \E\left( f(\XX_{i}, \XX_{l}) \,
         \sum_{k=1}^n \mathds{1}_{\{\XX_{i} = \XX_{n,k}, \XX_{i} \neq \XX_{k}\}} \right) 
         \\
  & \qquad
         + \frac{1}{|L|} \sum_{k=1}^n \sum_{j \neq k} \E\left( f(\XX_{k}, \XX_{j}) \,
         \sum_{l=1}^n \mathds{1}_{\{\XX_{j} = \XX_{n,l}, \XX_{j} \neq \XX_{l}\}} \right)
  \\
  & \leq \E(f(\XX_1, \XX_2)) 
         + \frac{1}{|L|} \sum_{i=1}^n \sum_{j \neq i} \E\left( f(\XX_{i}, \XX_{j}) \,
          K_{n,i} \, K_{n,j} \right)
  \\
  & \qquad 
         + \frac{1}{|L|} \sum_{l=1}^n \sum_{i \neq l} \E\left( f(\XX_{i}, \XX_{l}) \,
         K_{n,i} \right)
         + \frac{1}{|L|} \sum_{k=1}^n \sum_{j \neq k} \E\left( f(\XX_{k}, \XX_{j}) \,
         K_{n,j} \right)
  \\
  & \leq \E(f(\XX_1, \XX_2)) 
         + \left( C(p)^2 + 2 \, C(p)\right)\, \frac{1}{|L|} \sum_{i=1}^n \sum_{j \neq i} \E\left( f(\XX_{i}, \XX_{j}) \right)
  \\
  & \leq \left( 1 + \left( C(p)^2 + 2 \, C(p)\right)\, \frac{n(n-1)}{|L|} \right) \E\left( f(\XX_{1}, \XX_{2}) \right)\,.
\end{align*}
Finally, since $|L| \geq n(n-1) - n \, C(p)(C(p)-1)$, due to Lemma \ref{lem:App.1}, this proves the assertion.
\end{proof}

\begin{lemma}[Paired version of Lemma 11.7 in \citep{chatterjee2021AoS}]
\label{lem:11.7_m}
Suppose that, for each $n \geq 2$, the nearest neighbor graph is such that $N(1) \neq N(2)$, that is, the nearest neighbors of $\XX_1$ and $\XX_2$ differ.
Then, for any measurable function $f: \mathbb{R}^p \times \mathbb{R}^p \to \mathbb{R}$,
$\lim_{n \to \infty} f(\XX_{n,1}, \XX_{n,2}) = f(\XX_1, \XX_2)$ in probability. 
\end{lemma}

\begin{proof}
Fix $\varepsilon > 0$. 
Then there exists some compactly supported continuous function $g: \mathbb{R}^p \times \mathbb{R}^p \to \mathbb{R}$ (cf. \cite[Lemma 11.6]{chatterjee2021AoS}) such that 
\begin{align}\label{lem:11.7_m:Eq1}
  \PP^{(\XX_1,\XX_2)} (\{(\xx_1,\xx_2) \in \mathbb{R}^p \times \mathbb{R}^p \, : \, f(\xx_1,\xx_2) \neq g(\xx_1,\xx_2)\}) < \varepsilon\,.  
\end{align}
For any $\delta >0$, we then obtain 
\begin{align*} 
  & \PP(|f(\XX_1, \XX_2) - f(\XX_{n,1}, \XX_{n,2})| > \delta)
  \\
  & \leq \PP(|g(\XX_1, \XX_2) - g(\XX_{n,1}, \XX_{n,2})| > \delta)
  \\
  & \qquad + \underbrace{\PP(f(\XX_1, \XX_2) \neq g(\XX_1, \XX_2))}_{< \varepsilon; \, \eqref{lem:11.7_m:Eq1}}
           + \PP(f(\XX_{n,1}, \XX_{n,2}) \neq g(\XX_{n,1}, \XX_{n,2}))\,,
\end{align*}
with $\lim_{n \to \infty} \PP(|g(\XX_1, \XX_2) - g(\XX_{n,1}, \XX_{n,2})| > \delta) = 0$ due to continuity of $g$ and Lemma \ref{lem:11.3_m}.
Lemma \ref{lem:11.5_m} further gives 
\begin{align*}
  & \PP(f(\XX_{n,1}, \XX_{n,2}) \neq g(\XX_{n,1}, \XX_{n,2}))
  \\*
  & \leq \left(1 + C(p) (C(p)+2) \; \left( 1 + \frac{C(p) (C(p)-1)}{(n-1) - C(p) (C(p)-1)} \right) \right) \, \PP(f(\XX_{1}, \XX_{2}) \neq g(\XX_{1}, \XX_{2}))
  \\
  & \leq \left(1 + C(p) (C(p)+2) \; \left( 1 + \frac{C(p) (C(p)-1)}{(n-1) - C(p) (C(p)-1)} \right) \right) \, \varepsilon\,.
\end{align*}
Altogether, this yields 
\begin{align*} 
  \limsup_{n \to \infty} \PP(|f(\XX_1, \XX_2) - f(\XX_{n,1}, \XX_{n,2})| > \delta)
  \leq \left(2 + C(p) (C(p)+2)\right) \, \varepsilon\,.
\end{align*}
Since $\varepsilon$ and $\delta$ are arbitrary, this proves the assertion.
\end{proof}

In Lemma \ref{Estimator.Thm.Consistency.Split1} we have already established that $\alpha_n(\XX)$ is a consistent estimator for $\alpha(\XX)$. 
In order to prove consistency of $\gamma_n(Y|\XX)$ for $\alpha(\XX) \, \Lambda(Y|\XX)$ we will first (Lemma \ref{Estimator.Thm.Consistency.Split2.Lem1}) show that
\begin{equation*}
  \lim_{n \rightarrow \infty} \mathbb{E}(\gamma_n(Y|\XX)) 
  = \alpha(\XX) \, \Lambda(Y|\XX)\,,
\end{equation*}
and then (Lemma \ref{Estimator.Thm.Consistency.Split2.Lem2}) that there exist positive constants $C_1,C_2$ only depending on $p$ such that, for all $t \geq 0$ and all $n \in \mathbb{N}$,
\begin{align*}
  \mathbb{P}(|\gamma_n(Y|\XX) - \mathbb{E}(\gamma_n(Y|\XX))| \geq t) 
  & \leq  C_1 e^{- C_2 \, n \, t^2}\,.
\end{align*}
Combining the two results yields the claimed consistency in Theorem \ref{Estimator.Thm.Consistency.Split2}.

\begin{lemma} \label{Estimator.Thm.Consistency.Split2.Lem1}
It holds that 
$\lim_{n \rightarrow \infty} \mathbb{E}(\gamma_n(Y|\XX)) 
  = \alpha(\XX) \, \Lambda(Y|\XX)$\,.
\end{lemma}
\begin{proof}
Clearly
\begin{align*}
  \mathbb{E}(\gamma_n(Y|\XX)) 
  & = {n \choose 2}^{-1} \sum_{k < l} \, \mathbb{E}(\eta_{k,l} \, \eta_{N(k),N(l)})
\end{align*}
where $\eta_{k,l} := \sgn(Y_k - Y_l) = \mathds{1}_{\{Y_k > Y_l\}} - \mathds{1}_{\{Y_k < Y_l\}}$.
Let $\mathcal{F}$ be the $\sigma$-algebra induced by $\{\XX_1,...,\XX_n\}$ and any random variables used to break ties in the nearest neighbour structure.
\\
By definition, we always have $k \neq l$. 
In full generality we will need to distinguish three cases:
\begin{enumerate}
	\item $N(k) = N(l)$
	\item $N(k) = l$ and/or $N(l) = k$
	\item $N(k) \neq l,\, N(l) \neq k, N(k) \neq N(l)$
\end{enumerate}
In all three cases $|\mathbb{E}(\eta_{k,l} \, \eta_{N(k), N(l)}| \mathcal{F})| \leq 4$. 
We will first show that the number of pairs $(k, l)$ for which one of the first two cases is true, grows with order $n$ and therefore converges to $0$ when divided by ${n \choose 2} = n(n-1)/2$ as $n$ goes to infinity.
\\
For the first case, we obtain from Lemma \ref{lem:App.1}
\begin{align*}
  |\{(k,l) \in \{1,...,n\}^2 \mid k < l, N(k) = N(l)\}| 
  &   =  |M|
    \leq n \, \frac{C(p)(C(p) - 1)}{2}\,.
\end{align*}
For the second case it is easy to see that $N(k) = l$ occurs exactly once for every $k$ and $N(l) = k$ occurs exactly once for every $l$. Therefore the second case can occur at most $2n$ times.
\\
For the third case, i.e.~ when all four indices are distinct, we first get
\begin{align*}
  \mathbb{E}(\eta_{k,l} | \mathcal{F}) 
  & = \E \left( \mathds{1}_{\{Y_k > Y_l\}} - \mathds{1}_{\{Y_k < Y_l\}} | \, \mathcal{F} \right) 
  \\
  & = \int_{\mathbb{R}^2} \mathds{1}_{\{y_k>y_l\}} - \mathds{1}_{\{y_k<y_l\}} \de \mathbb{P}^{Y_k|\XX_k} \otimes \mathbb{P}^{Y_l|\XX_l} (y_k,y_l)
\end{align*}
and analogously for $\eta_{N(k), N(l)}$, hence
\begin{align*}
  &\mathbb{E}(\eta_{k,l} \, \eta_{N(k),N(l)}|\mathcal{F}) 
  \\*
  & = \mathbb{E}(\eta_{k,l}|\mathcal{F}) \, \mathbb{E}(\eta_{N(k),N(l)}|\mathcal{F})
 \\
  & = \left(\int_{\mathbb{R}^2} \mathds{1}_{\{y_1>y_2\}} - \mathds{1}_{\{y_1<y_2\}} \de \mathbb{P}^{Y_1|\XX_1} \otimes \mathbb{P}^{Y_2|\XX_2} (y_1,y_2) \right) \,
  \\
  & \qquad \cdot \left(\int_{\mathbb{R}^2} \mathds{1}_{\{y_1>y_2\}} - \mathds{1}_{\{y_1<y_2\}} \de \mathbb{P}^{Y_{N(1)}|\XX_{N(1)}} \otimes \mathbb{P}^{Y_{N(2)}|\XX_{N(2)}} (y_1,y_2) \right)\,.
\end{align*}
Now, let $c_{1,n}, c_{2,n}, c_{3,n}$  be the number of times the cases 1,2 and 3 occur in a given sample of size $n$ and let $e_{1,n}, e_{2,n}, e_{3,n}$ be $\mathbb{E}(\eta_{k,l}, \eta_{N(k), N(l)} | \mathcal{F})$ in the respective cases. Note that $c_{3,n} = {n \choose 2} - c_{1,n} - c_{2,n}$.
Taking all the above considerations together gives
\begin{align*} 
  & \lim_{n \rightarrow \infty} \mathbb{E}(\gamma_n(Y|\XX) | \mathcal{F}) 
  \\
  & = \lim_{n \rightarrow \infty}  {n \choose 2}^{-1} \left( c_{1,n} e_{1,n} + c_{2,n} e_{2,n} + c_{3,n} e_{3,n} \right) = \lim_{n \rightarrow \infty} {n \choose 2}^{-1} \, c_{3,n} \, e_{3,n}
  \\
  & = \lim_{n \rightarrow \infty} {n \choose 2}^{-1} \, c_{3,n} \, \left(\int_{\mathbb{R}^2} \mathds{1}_{\{y_1>y_2\}} - \mathds{1}_{\{y_1<y_2\}} \de \mathbb{P}^{Y_1|\XX_1} \otimes \mathbb{P}^{Y_2|\XX_2} (y_1,y_2) \right)
  \\
  & \qquad \cdot \left(\int_{\mathbb{R}^2} \mathds{1}_{\{y_1>y_2\}} - \mathds{1}_{\{y_1<y_2\}} \de \mathbb{P}^{Y_{n,1}|\XX_{n,1}} \otimes \mathbb{P}^{Y_{n,2}|\XX_{n,2}} (y_1,y_2) \right)
  \\
  & = \left(\int_{\mathbb{R}^2} \mathds{1}_{\{y_1>y_2\}} - \mathds{1}_{\{y_1<y_2\}} \de \mathbb{P}^{Y_1|\XX_1} \otimes \mathbb{P}^{Y_2|\XX_2} (y_1,y_2) \right)^2 
\end{align*}
in probability, where we use Lemma \ref{lem:11.7_m} for the convergence.
Finally, using uniform integrability of $\mathbb{E}(\gamma_n(Y|\XX) | \mathcal{F})$ in combination with \cite[Theorem 3.5]{Billingsley-1999} yields
\begin{align*}
  & \lim_{n \rightarrow \infty} \mathbb{E}(\gamma_n(Y|\XX)) 
    = \lim_{n \rightarrow \infty}  \mathbb{E} \left( \mathbb{E}(\gamma_n(Y|\XX)| \mathcal{F})\right)
  \\
  & = \int_{\mathbb{R}^p \times \mathbb{R}^p} \left(\int_{\mathbb{R}^2} \mathds{1}_{\{y_1>y_2\}} - \mathds{1}_{\{y_1<y_2\}} \de \mathbb{P}^{Y_1|\XX_1=\xx_1} \otimes \mathbb{P}^{Y_2|\XX_2=\xx_2} (y_1,y_2) \right)^2 \de \PP^{(\XX_1,\XX_2)} (\xx_1,\xx_2)
  \\
  & = \alpha(\XX) \, \Lambda(Y|\XX)\,.
\end{align*}
This proves the assertion.
\end{proof}

\begin{lemma} \label{Estimator.Thm.Consistency.Split2.Lem2}
There are constants $C_1 > 0$ and $C_2 > 0$ depending only on dimension $p$ such that 
\begin{align*}
  \mathbb{P}(|\gamma_n(Y|\XX) - \mathbb{E}(\gamma_n(Y|\XX))| \geq t) 
  & \leq  C_1 e^{- C_2 \, n \, t^2}
\end{align*}
for all $t \geq 0$ and all $n \in \mathbb{N}$.
\end{lemma}
\begin{proof}
Similar as in \citep{chatterjee2021AoS} we use McDiarmid's bounded difference inequality \citep{mcdiarmid1989}. 
Thus, we will need to bound the change in $\gamma_n(Y|\XX)$ as we replace a $(Y_m, \XX_m, U_m)$ with $(Y_m', \XX_m', U_m')$. 
Here the $U_i$'s are uniformly $[0,1]$-distributed variables used to break ties among the $\XX_i$'s to determine the nearest neighbour.
\\
Since $\eta_{k,l}\,\eta_{N(k),N(l)} = \sgn(Y_k - Y_l) \, \sgn(Y_{N(k)} - Y_{N(l)}) \in \{-1, 0, 1\}$ it is immediately obvious that any term in the double sum \eqref{Estimator.Def.Gamma} can change by at most $2$.
When determining the potential change in $\gamma_n(Y|\XX)$ as a whole we only need to count the number of terms $\eta_{k,l}\,\eta_{N(k),N(l)}$ which can change. 
$\eta_{k,l}\,\eta_{N(k),N(l)}$ can change when either $\sgn(Y_k - Y_l)$ changes or when $\sgn(Y_{N(k)} - Y_{N(l)})$ changes.
Clearly when we swap $(Y_m, \XX_m, U_m)$ with $(Y_m', \XX_m', U_m')$ then all terms $\sgn(Y_k - Y_l)$ for which either $k = m$ or $l = m$ can change. These are $(n - 1)$ terms.
Let again $S_m = \{k \in \{1,...,n\} | N(k) = m\}$ be the set of indices for which $m$ is the nearest neighbor.
The terms of the form $\sgn(Y_{N(k)} - Y_{N(l)})$ can change whenever $k \in S_m \cup S_{m'}$ or $l \in S_m \cup S_{m'}$.
Since every index occurs $n-1$ times in the double sum \eqref{Estimator.Def.Gamma}, these are at most $4 C(p) (n-1)$ cases.
In total therefore at most $(n - 1)(1 + 4 C(p))$ terms can change by at most $2$ and thus $\gamma_n(Y|\XX)$ changes by at most
\begin{equation*}
  \binom{n}{2}^{-1} \; 2 \, (n-1)(1 + 4C(p)) = \frac{4 \, (1 + 4C(p))}{n} =: \frac{C_1}{n}\,.
\end{equation*}
The bounded difference inequality finally gives
\begin{align*}
  \mathbb{P}(|\gamma_n(Y|\XX) - \mathbb{E}(\gamma_n(Y|\XX))| \geq t) 
  & \leq 2 \, e^{-\frac{2t^2}{\sum_{i=1}^{n} (C_1/n)^2}}\,.
\end{align*}
This proves the assertion.	
\end{proof}
\end{appendix}
\end{document}